\documentclass[a4paper]{article}
\usepackage{enumitem}
\usepackage{calc}

\usepackage{lmodern} 

\usepackage{graphicx} 

\usepackage{amsmath, accents}
\usepackage{amssymb,tikz}
\usepackage{pict2e}
\usepackage{amsthm}
\usepackage{amscd}
\usepackage{mathrsfs}
\usepackage{todonotes}
\usetikzlibrary{calc} 

\usepackage{yfonts}

\usepackage{marvosym}
\usepackage[percent]{overpic}

\newtheorem{theorem}{Theorem} [section]
\newtheorem{proposition}[theorem]{Proposition}	
	
\newtheorem{lemma}[theorem]{Lemma}

\newtheorem{remark}[theorem]{Remark}
\theoremstyle{definition}
\newtheorem{definition}[theorem]{Definition}

\DeclareMathOperator{\tr}{Tr}
\DeclareMathOperator{\dist}{dist}

\newcommand{\C}{\mathbb{C}}
\newcommand{\R}{\mathbb{R}}
\newcommand{\N}{\mathbb{N}}
\newcommand{\re}{\text{\upshape Re\,}}
\newcommand{\im}{\text{\upshape Im\,}}

\newcommand{\Ai}{{\rm Ai}}

\newcommand{\Done}{\mathbb{D}_\delta(b_1)}
\newcommand{\Dtwo}{\mathbb{D}_\delta(b_2)}

\def\XXint#1#2#3{{\setbox0=\hbox{$#1{#2#3}{\int}$}
\vcenter{\hbox{$#2#3$}}\kern-.5\wd0}}

\usepackage{tikz}
\usetikzlibrary{arrows}
\usetikzlibrary{decorations.pathmorphing}
\usetikzlibrary{decorations.markings}
\usetikzlibrary{patterns}
\usetikzlibrary{automata}
\usetikzlibrary{positioning}
\usepackage{tikz-cd}
\tikzset{->-/.style={decoration={
				markings,
				mark=at position #1 with {\arrow{latex}}},postaction={decorate}}}
	
	\tikzset{-<-/.style={decoration={
				markings,
				mark=at position #1 with {\arrowreversed{latex}}},postaction={decorate}}}

\usetikzlibrary{shapes.misc}\tikzset{cross/.style={cross out, draw, 
         minimum size=2*(#1-\pgflinewidth), 
         inner sep=0pt, outer sep=0pt}}

\usepackage{pgfplots}
	
\allowdisplaybreaks	

\numberwithin{equation}{section}

\def\ds{\displaystyle}
\def\bigO{{\cal O}}

\usepackage[colorlinks=true]{hyperref}
\hypersetup{urlcolor=blue, citecolor=red, linkcolor=blue}

\begin{document}
\title{Higher order large gap asymptotics at the hard edge for Muttalib--Borodin ensembles}
\author{Christophe Charlier, Jonatan Lenells, Julian Mauersberger}

\maketitle

\begin{abstract}
We consider the limiting process that arises at the hard edge of Muttalib--Borodin ensembles. This point process depends on $\theta > 0$ and has a kernel built out of Wright's generalized Bessel functions. In a recent paper, Claeys, Girotti and Stivigny have established first and second order asymptotics for large gap probabilities in these ensembles. These asymptotics take the form
\begin{equation*}
\mathbb{P}(\mbox{gap on } [0,s]) = C \exp \left( -a s^{2\rho} + b s^{\rho} + c \ln s \right) (1 + o(1)) \qquad  \mbox{as }s \to + \infty,
\end{equation*}
where the constants $\rho$, $a$, and $b$ have been derived explicitly via a differential identity in $s$ and the analysis of a Riemann--Hilbert problem. Their method can be used to evaluate $c$ (with more efforts), but does not allow for the evaluation of $C$. In this work, we obtain expressions for the constants $c$ and $C$ by employing a differential identity in $\theta$. When $\theta$ is rational, we find that $C$ can be expressed in terms of Barnes' $G$-function.
We also show that the asymptotic formula can be extended to all orders in $s$.
\end{abstract}

\noindent
{\small{\sc AMS Subject Classification (2010)}: 41A60, 60B20, 33B15, 33E20, 35Q15.}

\noindent
{\small{\sc Keywords}: Muttalib--Borodin ensembles, Random matrix theory, asymptotic analysis, large gap probability, Riemann--Hilbert problems.}

\tableofcontents

\section{Introduction and main results}
The Muttalib--Borodin ensembles are joint probability density functions of the form
\begin{equation}\label{density function Muttalib Borodin}
\frac{1}{Z_{n}} \prod_{1 \leq j < k \leq n} (x_{k}-x_{j})(x_{k}^{\theta}-x_{j}^{\theta}) \prod_{j=1}^{n} w(x_{j})dx_{j},
\end{equation}
where the $n$ points $x_{1},\ldots,x_{n}$ belong to the interval $[0,+\infty)$, $\theta > 0$ is a parameter of the model, and $Z_{n}$ is a normalization constant. The positive weight function $w$ is defined on $[0,+\infty)$ and is assumed to have enough decay at $\infty$ to make \eqref{density function Muttalib Borodin} a well-defined density function. 

The  probability density function \eqref{density function Muttalib Borodin} exhibits so-called two-body interactions---in addition to the repulsion between the points $x_{1},\ldots,x_{n}$, there is also repulsion between the points $x_{1}^{\theta},\ldots,x_{n}^{\theta}$. The models defined by \eqref{density function Muttalib Borodin} were introduced by Muttalib in 1995 in the study of disordered conductors in the metallic regime \cite{Muttalib}. They have attracted a lot of attention recently in the random matrix community, partly due to the work of Cheliotis \cite{Cheliotis} who showed that the squared singular values of certain lower triangular random matrices have the same joint density as \eqref{density function Muttalib Borodin} in the case of the Laguerre weight $w(x) = x^{\alpha}e^{-x}$, $\alpha > -1$. Other matrix ensembles whose eigenvalues are distributed according to \eqref{density function Muttalib Borodin} for the Laguerre or Jacobi weight were obtained in \cite{ForWang}. 

As $n \to + \infty$ the macroscopic behavior of the points $x_{1},\ldots,x_{n}$ is well described by an equilibrium measure $\mu$ which depends on the weight $w$. Such measures have been studied in detail in \cite{ClaeysRomano} for general values of $\theta$. In particular, the authors of \cite{ClaeysRomano} found sufficient conditions on $w$ for $\mu$ to be supported on a single cut. If there is a hard edge (that is, if part of the points accumulate near the origin as $n \to +\infty$), the density of $\mu$ behaves as a constant times $x^{-\frac{1}{1+\theta}}$ as $x \to 0^{+}$. On the other hand, near a soft edge, this density vanishes to the order $1/2$ for any value of $\theta$; this is the usual square root behavior that is often encountered in random matrix theory. We also refer to \cite{BLVW, Butez, Kuij} for related results on the equilibrium measure.

The Muttalib--Borodin point process is determinantal for any $\theta >0$. This means that the density \eqref{density function Muttalib Borodin}, as well as all the associated correlation functions, can be expressed as determinants involving a function $\mathbb{K}_{n}$ (general definitions and properties of point processes can be found in \cite{Johansson, Soshnikov, BorodinPoint}). This function $\mathbb{K}_{n}$ is called the kernel and encodes all the probabilistic information about the point process. In the simplest case $\theta = 1$, the point process is a polynomial ensemble. This means that all the correlation functions can be expressed in terms of orthogonal polynomials (associated to $w$), and that there exists a Christoffel--Darboux formula which can be utilized to derive asymptotic formulas as $n \to + \infty$. For $\theta \neq 1$, the point process is still determinantal; however the aforementioned properties become more complicated for rational values of $\theta$, and are lost if $\theta$ is irrational. In fact, for $\theta \neq 1$, the kernel is instead expressed in terms of biorthogonal polynomials \cite{Borodin}, for which there is no simple analog of the Christoffel--Darboux formula (when $\theta$ is an integer, the Christoffel--Darboux formula contains $\theta$ terms, see \cite{yasov}).

As $n \to + \infty$, the local repulsion of the points leads to microscopic limit laws that depend on the location. The term microscopic refers to the fact that the correlation is measured in the unit of the mean level spacing. For $\theta = 1$, three different canonical limiting kernels arise: the sine kernel arises in the bulk, the Airy kernel near soft edges, and the Bessel kernel near (typical) hard edges. The three limiting kernels are independent of the fine details of the weight; this phenomenon is called universality in random matrix theory \cite{KuijUniversality}. Also, the kernels are all \textit{integrable} (of size $2$) in the sense of Its-Izergin-Korepin-Slavnov \cite{IIKS}, and there are $2 \times 2$ matrix Riemann--Hilbert (RH) problems available for the asymptotics analysis. Much less is known for $\theta \neq 1$. In the case of the Laguerre weight $w(x) = x^{\alpha}e^{-x}$, $\alpha > -1$, Borodin proved in his pioneering work \cite{Borodin} that
\begin{equation}\label{hard edge scaling limit}
\lim_{n\to + \infty} \frac{1}{n^{\frac{1}{\theta}}} \mathbb{K}_{n}\Big( \frac{x}{n^{\frac{1}{\theta}}},\frac{y}{n^{\frac{1}{\theta}}} \Big) = \mathbb{K}(x,y), \qquad x,y > 0,
\end{equation}
for any $\alpha > -1$ and $\theta > 0$, where the limiting kernel $\mathbb{K}(x,y)$ depends on $\alpha$ and $\theta$ and can be expressed in terms of Wright's generalized Bessel functions (see also \eqref{def kernel} below). If $\theta = p/q$ with $p,q$ relatively prime integers, then the kernel $\mathbb{K}$ is integrable, but of size $p+q$ \cite{Zhang}, which means that the associated RH problems involve matrices of size $(p+q)\times (p+q)$. In the Jacobi case (i.e., the weight is supported on $[0,1]$ and given by $w(x) = x^{\alpha}$, $\alpha > -1$), Borodin proved that the same limiting kernel $\mathbb{K}$ appears at the hard edge if a slightly different scaling limit is considered (the terms $n^{\frac{1}{\theta}}$ in \eqref{hard edge scaling limit} need to be replaced by $n^{1+\frac{1}{\theta}}$). It seems reasonable to expect some universality of this kernel, in the sense that $\mathbb{K}$ should appear in the hard edge scaling limit for a large class of weights. Moreover, from the behavior of $\mu$ described above, one expects the sine kernel in the bulk and the Airy kernel at the soft edge for a large class of weights. This has been proved in the special case $\theta = \frac{1}{2}$ only recently by Kuijlaars and Molag in \cite{KuijMolag} using a non-standard analysis of a $3 \times 3$ matrix RH problem. The case of general $\theta$ is still open. We also mention that, in case $\theta$ or $1/\theta$ is an integer, the kernel $\mathbb{K}$ can be expressed in terms of a Meijer function and coincides with the limiting kernel at the hard edge of certain product random matrices \cite[Theorem 5.1]{KuijSti2014}.


There are several expressions available in the literature for the kernel $\mathbb{K}(x,y)$ in \eqref{hard edge scaling limit}; in \cite{Borodin} it is written as a series, and also in terms of Wright's generalized Bessel functions. For us, the following double contour integral expression (from \cite{ClaeysGirSti}) will be important:
\begin{equation}\label{def kernel}
\mathbb{K}(x,y) = \frac{1}{4\pi^{2}} \int_{\gamma} du \int_{\tilde{\gamma}} dv \frac{F(u)}{F(v)} \frac{x^{-u}y^{v-1}}{u-v}, \qquad x,y >0,
\end{equation}
where the function $F$ is given by
\begin{equation}\label{def of F}
F(z) = \frac{\Gamma \big( z+\frac{\alpha}{2} \big)}{\Gamma \left( \frac{\frac{\alpha}{2}+1-z}{\theta} \right)}
\end{equation}
with $\Gamma$ denoting the Gamma function (see \cite[Chapter 5]{NIST}). 
The contours $\gamma$ and $\tilde{\gamma}$ are both oriented upward and do not intersect each other; the contour $\gamma$ intersects $\mathbb{R}$ to the right of the poles of $F$ and $\tilde{\gamma}$ intersects $\mathbb{R}$ to the left of the zeros of $F$, see Figure \ref{fig: contours gamma}. The contour $\gamma$ tends to infinity in sectors lying strictly in the left half-plane, and $\tilde{\gamma}$ tends to infinity in sectors lying strictly in the right half-plane. If $\theta = 1$, the kernel $\mathbb{K}$ reduces to
\begin{equation}\label{kernel theta=1}
\left. \mathbb{K}(x,y) \right|_{\theta = 1} = 4\mathbb{K}_{\mathrm{Be}}(4x,4y), \qquad x,y > 0,
\end{equation}
where $\mathbb{K}_{\mathrm{Be}}$ is the well-known Bessel kernel \cite{TraWidLUE} given by
\begin{equation*}
\mathbb{K}_{\mathrm{Be}}(x,y) = \frac{J_{\alpha}(\sqrt{x})\sqrt{y}J_{\alpha}^{\prime}(\sqrt{y})-\sqrt{x}J_{\alpha}^{\prime}(\sqrt{x})J_{\alpha}(\sqrt{y})}{2(x-y)},
\end{equation*}
with $J_\alpha$ the Bessel function of the first kind of order $\alpha$.

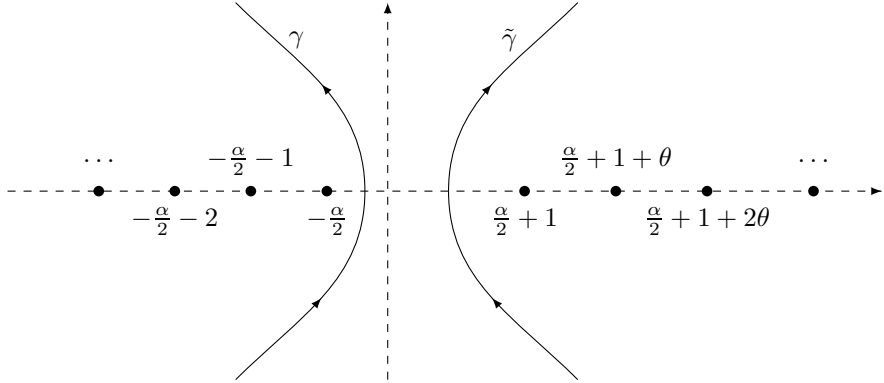
\begin{figure}
\begin{center}
\begin{tikzpicture}
\node at (0,0) {};
\draw[dashed,->-=1,black] (0,-2.5) to [out=90, in=-90] (0,2.5);
\draw[dashed,->-=1,black] (-5,0) to [out=0, in=-180] (6.5,0);

\node at (-1.2,2) {$\gamma$};
\draw[->-=0.5,black] (-2,-2.5) to [out=45, in=-90] (-0.3,0);
\draw[->-=0.5,black] (-0.3,0) to [out=90, in=-45] (-2,2.5);

\node at (1.6,2) {$\tilde{\gamma}$};
\draw[->-=0.5,black] (2.5,-2.5) to [out=135, in=-90] (0.8,0);
\draw[->-=0.5,black] (0.8,0) to [out=90, in=-135] (2.5,2.5);


\fill (-0.8,0) circle (0.07cm);
\node at (-0.8,-0.4) {$-\frac{\alpha}{2}$};
\fill (-1.8,0) circle (0.07cm);
\node at (-1.8,0.4) {$-\frac{\alpha}{2}-1$};
\fill (-2.8,0) circle (0.07cm);
\node at (-2.8,-0.4) {$-\frac{\alpha}{2}-2$};
\fill (-3.8,0) circle (0.07cm);
\node at (-3.8,0.4) {$\ldots$};

\fill (1.8,0) circle (0.07cm);
\node at (1.8,-0.4) {$\frac{\alpha}{2}+1$};
\fill (3,0) circle (0.07cm);
\node at (3,0.4) {$\frac{\alpha}{2}+1+\theta$};
\fill (4.2,0) circle (0.07cm);
\node at (4.2,-0.4) {$\frac{\alpha}{2}+1+2\theta$};
\fill (5.6,0) circle (0.07cm);
\node at (5.6,0.4) {$\ldots$};
\end{tikzpicture}
\end{center}
\caption{The contours $\gamma$ and $\tilde{\gamma}$ for $\alpha = 1.6$ and  $\theta = 1.2$. The dots are the zeros and poles of $F$.  \label{fig: contours gamma}}
\end{figure}
By \cite[equation (1.15)]{ClaeysGirSti}, the finite $n$ probability to observe a gap on $[0, n^{-\frac{1}{\theta}}s]$ converges as $n \to + \infty$ to the probability to observe a gap on $[0,s]$ in the limiting process with kernel $\mathbb{K}$. This is a slightly stronger result than the convergence of the kernel \eqref{hard edge scaling limit}. Let $x^{\star} := \min \{x_{1},\ldots,x_{n}\}$ denote the smallest point. Then the limiting distribution of $x^{\star}$ is given by
\begin{equation}\label{fredholm determinant}
\lim_{n\to + \infty} \mathbb{P}_{n}\big(n^{\frac{1}{\theta}}x^{\star} > s\big) = \det\big(1- \left. \mathbb{K} \right|_{[0,s]}\big), \qquad s > 0,
\end{equation}
where the right-hand side is the Fredholm determinant associated to $\mathbb{K}$ on the interval $[0,s]$.
For $\theta = 1$, Tracy and Widom have shown in \cite{TraWidLUE} that the log $s$-derivative of this Fredholm determinant solves a Painlev\'{e} V equation. In the case of $\theta$ rational, a more involved system of differential equations has been derived recently in \cite{Zhang}.

In the case $\theta = 1$, the large gap asymptotics (i.e., the asymptotics of \eqref{fredholm determinant} as $s \to + \infty$) are known from Deift, Krasovsky and Vasilevska \cite[Theorem 4]{DeiftKrasVasi} where it was shown that\footnote{Note that, due to the re-scaling \eqref{kernel theta=1}, $\left.\det \big(  1-\left.\mathbb{K}\right|_{[0,s]} \big)\right|_{\theta = 1} = \det \big( \left. 1-\mathbb{K}_{\mathrm{Be}} \right|_{[0,4s]} \big)$, and thus one should use \cite[Theorem 4]{DeiftKrasVasi} with $s$ replaced by $4s$ to obtain \eqref{asymp theta=1}.}
\begin{equation}\label{asymp theta=1}
\left.\det \Big(  1-\left.\mathbb{K}\right|_{[0,s]} \Big)\right|_{\theta = 1} = \frac{G(1+\alpha)}{(2\pi)^{\frac{\alpha}{2}}} \exp \Big( -s + 2 \alpha \sqrt{s} - \frac{\alpha^{2}}{4} \ln(4s) + \bigO(s^{-1/2}) \Big)  \quad \mbox{as } s \to +\infty,
\end{equation}
where $G$ is Barnes' $G$-function (see \cite[Chapter 5]{NIST}). The study of the general case $\theta > 0$ has been initiated by Claeys, Girotti and Stivigny in the recent paper \cite{ClaeysGirSti}. They obtained the asymptotic formula
\begin{equation}\label{known result fredholm general theta}
\det \Big( \left. 1-\mathbb{K} \right|_{[0,s]} \Big) = C \exp \Big(  -a s^{2\rho}+b s^{\rho}+c \ln s +\bigO(s^{-\rho}) \Big) \quad \mbox{as } s \to + \infty,
\end{equation}
where the real constants $\rho$, $a$, and $b$ are explicitly given by
\begin{equation}\label{coeff rho a b}
\rho = \frac{\theta}{1+\theta}, \qquad a = \frac{1}{4}(1+\theta)^{2} \theta^{\frac{1-3\theta}{1+\theta}}, \quad \mbox{ and } \quad  b = \frac{1}{2} (1+\theta)(1+2\alpha - \theta) \theta^{-\frac{2\theta}{\theta + 1}}.
\end{equation}
It is quite remarkable that, even though the kernel $\mathbb{K}$ is known to be integrable only for rational $\theta$, they managed to obtain an asymptotic formula valid \textit{for any} fixed $\theta > 0$ (we comment on their method below). 

\subsection{Main results}
The constants $c$ and $C$ in the large gap probability \eqref{known result fredholm general theta} are \textit{multiplicative} constants. Therefore, there is no accurate description of the large gap probability without their explicit expressions. Obtaining such expressions is precisely the purpose of this paper. Our main result is the following.

\begin{theorem}[Explicit expressions for $c$ and $C$]\label{thm:main results}
For any fixed $\theta > 0$ and $\alpha > -1$, the constants $c$ and $C$ that appear in the asymptotic formula \eqref{known result fredholm general theta} are given by
\begin{align}
c = & -\frac{6 \alpha^2-6\alpha(\theta-1) +(\theta-1)^2}{12(1+\theta)}, \label{little c in thm} 
	\\ \nonumber
C = &\; \frac{G ( 1+\alpha )}{(2\pi)^{\frac{\alpha}{2}}}   \exp \big( d(1,\alpha) - d(\theta,\alpha) \big)  \exp\left( \frac{24 \alpha (\alpha +2)+15+3\theta + 4 \theta^{2}}{24(1+\theta)} \ln \theta \right) 
	\\
& \times \exp \left( \frac{6\alpha \theta - 6 \alpha (1+\alpha)-(\theta-1)^{2}}{12 \theta} \ln(1+\theta) \right), \label{big C in thm} 
\end{align}
where $G$ is the Barnes $G$-function and the real quantity $d(\theta,\alpha)$ is defined by the limit
\begin{align}\nonumber
d(\theta,\alpha) = &\; \lim_{N\to + \infty} \Bigg[ \sum_{k=1}^{N} \ln \Gamma (1+\alpha + k \theta) - \Bigg\{\frac{\theta}{2} N^{2} \ln N +  \frac{\theta(2 \ln \theta - 3)}{4}N^{2} 
	 \\ \label{def of the constant d}
& +\left(1+\alpha + \frac{\theta -1}{2}\right) N \ln N + \left( \frac{\ln(2\pi)}{2}-(1+\alpha)+\frac{1-\theta}{2} + \left( \alpha + \frac{1+\theta}{2} \right) \ln \theta \right) N \nonumber \\
& +\frac{1+6\alpha^{2} + \theta (3+\theta) + 6\alpha(1+\theta)}{12 \theta} \ln N \Bigg\} \Bigg]. 
\end{align}
\end{theorem}
\begin{remark}[The case $\theta = 1$]\upshape\label{theta1remark}
For $\theta = 1$, the expressions for the coefficients $\rho, a,b,c$, and $C$ given in (\ref{coeff rho a b})--(\ref{big C in thm}) reduce to
\begin{equation}
\rho = \frac{1}{2}, \qquad a = 1, \qquad b = 2\alpha, \qquad c = -\frac{\alpha^{2}}{4}, \quad \mbox{ and } \quad C = \frac{2^{- \frac{\alpha^{2}}{2}}G(1+\alpha)}{(2\pi)^{\frac{\alpha}{2}}},
\end{equation}
so we recover \eqref{asymp theta=1} as a special case of \eqref{known result fredholm general theta}. 
\end{remark}

\begin{remark}[The constant $d$]\upshape
The constant $d =d(\theta, \alpha)$ (constant in the sense that it is independent of $s$) is defined by the limit in (\ref{def of the constant d}) in a similar way as the Euler gamma constant $\gamma_{\mathrm{E}}$, which appears in the definition of $G$, is defined by (see \cite[Eq. 5.2.3]{NIST})
\begin{equation*}
\gamma_{\mathrm{E}} = \lim_{N\to + \infty} \sum_{k=1}^{N} \left( \frac{1}{k}-\ln \left( 1+\frac{1}{k} \right) \right).
\end{equation*}
The definition of $d$ can also be compared with the following expression for the derivative of the Riemann $\zeta$-function evaluated at $-1$ (see \cite[Eq. 5.17.7]{NIST}):
\begin{align*}
\zeta'(-1) = \lim_{N \to + \infty} \Bigg[ \sum_{k=1}^{N} \ln \Gamma(k) - \Bigg\{ \frac{1}{2}N^{2}\ln N - \frac{3}{4} N^{2} + \frac{\ln(2\pi)}{2}N-\frac{1}{12}\ln(N) \Bigg\} \Bigg].
\end{align*}
In fact, comparing the above expression with the definition (\ref{def of the constant d}) of $d$, we see that\footnote{Note that $d(\theta,\alpha)$ is well-defined for $\alpha = -1$ even though the point process is defined only for $\theta > 0$ and $\alpha > -1$.} 
$$d(1,-1) = \zeta'(-1).$$
More generally, for $\theta = 1$ but any value of $\alpha > -1$, we can use the functional equation for the Barnes $G$-function to write
\begin{equation}
\sum_{k=1}^{N} \ln \Gamma (1+\alpha + k) = \ln G(1+\alpha + N) - \ln G(2+\alpha).
\end{equation}
Using the expansion (see \cite[Eq. 5.17.5]{NIST})
\begin{equation}
\ln G(z+1) = \frac{z^{2}}{2}\ln z - \frac{3}{4}z^{2} + \frac{\ln(2\pi)}{2}z - \frac{1}{12}\ln z + \zeta'(-1) + \bigO(z^{-1}), \qquad z \to + \infty,
\end{equation}
we conclude from  (\ref{def of the constant d}) that
\begin{align}\label{d with theta=1}
& d(1,\alpha) = \zeta'(-1) + \frac{1+\alpha}{2}\ln (2\pi) - \ln G(2+\alpha)
\end{align}
for $\alpha > -1$. In other words, for $\theta = 1$, $d(\theta,\alpha)$ is expressed in terms of already known special functions evaluated at certain points. The next proposition shows that this is still the case if $\theta$ is a rational number, but then the expression becomes more complicated. 

\begin{proposition}[Expression for $d(\theta, \alpha)$ when $\theta = p/q$ is rational]\label{prop: constant d for rational theta}
Let $\alpha > -1$ and $\theta = p/q$ where $p,q \in \mathbb{N}\setminus \{0\}$. Then $d(\theta,\alpha)$ admits the following expression:
\begin{align}\nonumber
d\Big(\theta = \frac{p}{q},\alpha\Big) = &\; pq\zeta'(-1) + \frac{p+(1+2\alpha)q}{4} \ln(2\pi)-\frac{1+6\alpha^{2} + \theta (3+\theta) + 6\alpha(1+\theta)}{12 \theta} \ln q 
	 \\ \label{d in terms of Barnes G intro}
& - \sum_{k=1}^{q} \sum_{j = 1}^{p} \ln G \left( \frac{j+\alpha}{p} + \frac{k}{q} \right). 
\end{align}
\end{proposition}
\begin{proof}
See Appendix \ref{subsection: d rational}.
\end{proof}

Quantities such as $\zeta'(-1)$ or $G(1+\alpha)$ appear in several asymptotic formulas in random matrix theory. For example, $\zeta'(-1)$ appears in the large gap asymptotics of the Airy point process \cite{DIK} and in the asymptotics of the partition function for a large class of random matrix ensembles \cite[equations (1.38)-(1.40)]{CharlierGharakhloo}, while the Barnes $G$-function appears in the large gap asymptotics of the Bessel point process (see \eqref{asymp theta=1}) and in the asymptotics of large Toeplitz and Hankel determinants with Fisher-Hartwig singularities \cite{DIK-FH,CharlierGharakhloo}. However, despite its relatively simple definition, we have not been able to express $d$ in terms of known special functions for irrational values of $\theta$.
\end{remark}

\begin{remark}[The symmetry $\theta \mapsto \frac{1}{\theta}$]\upshape\label{symmremark}
By \cite[page 4]{Borodin}, the determinant on the left-hand side of \eqref{known result fredholm general theta} is invariant under the following changes of the parameters:
\begin{align}\label{symmetry}
s \mapsto s^{\theta}, \qquad \theta \mapsto \frac{1}{\theta}, \quad \mbox{ and } \quad \alpha \mapsto \alpha^{\star}:=\frac{1+\alpha}{\theta}-1.
\end{align} 
It follows that the coefficients $\rho$, $a$, $b$, $c$, and $C$ must obey the following symmetry relations for any $\theta > 0$ and $\alpha > -1$:
\begin{align}\nonumber
& \rho(\theta,\alpha) = \theta \rho\big( \tfrac{1}{\theta},\alpha^{\star} \big), \quad 
a(\theta,\alpha) = a(\tfrac{1}{\theta},\alpha^{\star}), \quad 
b(\theta,\alpha) = b(\tfrac{1}{\theta},\alpha^{\star}), 
	\\ \label{rhoabcCsymm}
& c(\theta,\alpha) = \theta c(\tfrac{1}{\theta},\alpha^{\star}), \quad
C(\theta,\alpha) = C(\tfrac{1}{\theta},\alpha^{\star}),
\end{align}
where we have indicated the dependence of the coefficients on $\theta$ and $\alpha$ explicitly. The first four of these relations are easily verified directly from the definitions \eqref{coeff rho a b}--\eqref{little c in thm} by simple computations. The relation $C(\theta,\alpha) = C(\tfrac{1}{\theta},\alpha^{\star})$ can also be verified directly from the definition (\ref{big C in thm}) of $C$, but the computations are more involved. In fact, a long but straightforward computation which uses (\ref{d with theta=1}) and the functional relation $G(z+1) = \Gamma(z)G(z)$ implies that the relation $C(\theta,\alpha) = C(\tfrac{1}{\theta},\alpha^{\star})$ is equivalent to the symmetry relation for $d$ given in the following proposition.
\end{remark}

\begin{proposition}[Symmetry relation for $d$]\label{dsymmprop}
The constant $d = d(\theta, \alpha)$ defined in (\ref{def of the constant d}) satisfies  
\begin{align} \nonumber
d(\theta,\alpha) = & \; d \left( \frac{1}{\theta},\frac{1+\alpha}{\theta}-1 \right) + \ln \Gamma \left( \frac{1+\alpha}{\theta} \right) - \ln \Gamma \left( 1+\alpha \right) 
	\\ \label{symmetry for d} 
& \; + \frac{13 + 6 \alpha^{2} + \theta(\theta-3) + 6 \alpha (\theta + 3)}{12 \theta}\ln \theta
\end{align}
for $\theta > 0$ and $\alpha > -1$.
\end{proposition}
\begin{proof}
See Appendix \ref{subsection: symmetry for d}.
\end{proof}

Our second main result shows that the expansion \eqref{known result fredholm general theta} of the Fredholm determinant of $\mathbb{K}$ on $[0,s]$ can be extended to all orders in powers of $s^{-\rho}$ as $s \to + \infty$. More precisely, we have the following.
\begin{theorem}[Asymptotics to all orders]\label{thm: all order expansion}
Let $N \geq 1$ be an integer and fix $\theta > 0$ and $\alpha > -1$. As $s \to + \infty$, there exist constants $C_1,\ldots,C_N \in \mathbb{R}$ such that
\begin{equation}\label{all order expansion}
\det \Big( \left. 1-\mathbb{K} \right|_{[0,s]} \Big) = C \exp \Big(  -a s^{2\rho}+b s^{\rho}+c \ln s + \sum_{j=1}^{N}C_j s^{-j\rho} + \bigO\big(s^{-(N+1)\rho}\big) \Big), 
\end{equation}
where $\mathbb{K}$ is the kernel defined in (\ref{def kernel}) and $\rho, a,b,c,C$ are given by (\ref{coeff rho a b})--(\ref{big C in thm}).
\end{theorem}

\subsection{Outline of proofs}\label{outlinesubsec}
Our proof of Theorem \ref{thm:main results} is based on some preliminary results from \cite{ClaeysGirSti}. An important and remarkable ingredient of that paper (inspired by \cite{BertolaCafasso}) is the identity
\begin{equation}\label{det integrable}
\det \Big( \left. 1-\mathbb{K} \right|_{[0,s]} \Big) = \det \Big( 1-\mathbb{M}_{s} \Big),
\end{equation}
where the integrable kernel $\mathbb{M}_{s}$ of size $2 \times 2$ is given for any $\theta > 0$ by
\begin{equation}\label{integrable kernel}
\mathbb{M}_{s}(u,v) = \frac{\textbf{f}(u)^{T}\textbf{g}(v)}{u-v}, \qquad \textbf{f}(u) = \frac{1}{2\pi i} \begin{pmatrix}
\chi_{\gamma}(u) \\
s^{u} \chi_{\tilde{\gamma}}(u)
\end{pmatrix}, \qquad \textbf{g}(v) = \begin{pmatrix}
-F(v)^{-1} \chi_{\tilde{\gamma}}(v) \\
s^{-v}F(v)\chi_{\gamma}(v)
\end{pmatrix}
\end{equation}
with $\chi_{\gamma}$ and $\chi_{\tilde{\gamma}}$ denoting the indicator functions of $\gamma$ and $\tilde{\gamma}$, respectively. Using some results from \cite{Bertola, BertolaCafasso} and following the procedure developed by Its-Izergin-Korepin-Slavnov (IIKS) \cite{IIKS}, the authors of \cite{ClaeysGirSti} obtained a differential identity for
\begin{equation}\label{log der s of fredholm}
\partial_{s} \ln \det \Big( \left. 1-\mathbb{K} \right|_{[0,s]} \Big)
\end{equation}
in terms of the solution $Y$ of a $2 \times 2$ matrix RH problem. Moreover, by performing a (non-standard) Deift/Zhou \cite{DeiftZhou} steepest descent analysis of this RH problem, they computed the large $s$ asymptotics of the expression in \eqref{log der s of fredholm}. The asymptotic formula \eqref{known result fredholm general theta} and the expressions (\ref{coeff rho a b}) for the coefficients $a$ and $b$ were then obtained from the relation
\begin{equation}\label{s integration idea}
\ln \det \Big( \left. 1-\mathbb{K} \right|_{[0,s]} \Big) = \ln \det \Big( \left. 1-\mathbb{K} \right|_{[0,M]} \Big) + \int_{M}^{s} \partial_{s'} \ln \det \Big( \left. 1-\mathbb{K} \right|_{[0,s']} \Big)ds',
\end{equation}
where $M$ is a sufficiently large but fixed constant. 

In principle, the method of \cite{ClaeysGirSti} can be employed to obtain any number of terms in the large $s$ expansion of \eqref{log der s of fredholm} (even though the computations become technically more involved as higher order terms are included). 
In particular, it is possible to compute the constant $c$ by extending the expansion of \eqref{log der s of fredholm} to the next order and then substituting the resulting asymptotics into the integrand of \eqref{s integration idea}. However, the fact that the quantity
\begin{equation*}
\ln \det \Big( \left. 1-\mathbb{K} \right|_{[0,M]} \Big)
\end{equation*}
is an unknown constant (independent of $s$) is an essential obstacle to the computation of $C$, see also \cite[Remark 1.3]{ClaeysGirSti}. Therefore, in the present work we adopt a different approach which takes advantage of the known result for $\theta = 1$ given in \eqref{asymp theta=1}. 

Whereas the approach of \cite{ClaeysGirSti} is based on a differential identity in $s$, our approach relies on a differential identity in $\theta$. More precisely, using \eqref{det integrable}--\eqref{integrable kernel} and results from \cite[Section 5.1]{Bertola}, we apply the IIKS procedure \cite{IIKS} to obtain a differential identity for 
\begin{equation}\label{log der theta of fredholm}
\partial_{\theta} \ln \det \Big( \left. 1-\mathbb{K} \right|_{[0,s]} \Big)
\end{equation}
in terms of the solution $Y$ of the RH problem of \cite{ClaeysGirSti} mentioned above (henceforth referred to as the RH problem for $Y$). 
By recycling the steepest descent analysis of \cite{ClaeysGirSti}, we obtain asymptotics of $Y$ as $s \to + \infty$. The steepest descent analysis in \cite{ClaeysGirSti} was performed for $\theta$ fixed, but we can easily show that the resulting asymptotic formulas are in fact valid uniformly for $\theta$ in any compact subset of $(0,+\infty)$. An integration of \eqref{log der theta of fredholm} from $\theta = 1$ to an arbitrary (but fixed) $\theta > 0$ then gives
\begin{equation}\label{theta integration idea}
\left.\ln \det \Big( \left. 1-\mathbb{K} \right|_{[0,s]} \Big)\right|_{\theta} = \left.\ln \det \Big( \left. 1-\mathbb{K} \right|_{[0,s]} \Big)\right|_{\theta=1} + \int_{1}^{\theta} \left. \partial_{\theta'} \ln \det \Big( \left. 1-\mathbb{K} \right|_{[0,s]} \Big) \right|_{\theta'}d\theta'.
\end{equation}
The main advantage of this approach is that the large $s$ asymptotics of
\begin{equation}
\left.\ln \det \Big( \left. 1-\mathbb{K} \right|_{[0,s]} \Big)\right|_{\theta=1}
\end{equation}
are known (including the constant term), see \eqref{asymp theta=1}. Therefore, if we compute the asymptotics of \eqref{log der theta of fredholm} to sufficiently high order and substitute the resulting expansion into \eqref{theta integration idea} (using the uniformity of this expansion with respect to $\theta$), we can obtain $C$ by performing the integral with respect to $\theta'$. 

\subsubsection{The two cases $\theta \leq 1$ and $\theta \geq 1$}\label{twocasessubsubsec}
The proof of Theorem \ref{thm:main results} naturally splits into the two cases $\theta \in (0,1]$ and $\theta \in [1, \infty)$. Similar techniques can be used to handle both of these cases, but since they are associated with different branch cut structures, slightly different arguments are required. To avoid having to deal with two different cases, we will therefore, for simplicity, give the derivation of Theorem \ref{thm:main results} only in the case $\theta \in (0,1]$ and then appeal to the symmetry (\ref{symmetry}) to extend the result to $\theta \in [1, \infty)$. The extension to $\theta \in [1, \infty)$ can be carried out as follows: Assuming that Theorem \ref{thm:main results} holds for $\theta \in (0,1]$, the invariance of the determinant in \eqref{known result fredholm general theta} under the symmetry (\ref{symmetry}) implies that, for any $\theta \in [1, \infty)$,
\begin{align*}
\ln \det \Big( \left. 1-\mathbb{K} \right|_{[0,s]} \Big) = & 
-a(\tfrac{1}{\theta},\alpha^{\star}) s^{2\theta\rho(\frac{1}{\theta},\alpha^{\star})}+b(\tfrac{1}{\theta},\alpha^{\star}) s^{\theta\rho(\frac{1}{\theta},\alpha^{\star})}+c(\tfrac{1}{\theta},\alpha^{\star}) \ln s 
	\\
& + C(\tfrac{1}{\theta},\alpha^{\star}) +\bigO\Big(\frac{1}{s^{\theta\rho(\frac{1}{\theta},\alpha^{\star})}} \Big), \qquad s \to + \infty.
\end{align*}
Using the symmetries in (\ref{rhoabcCsymm}), which we recall can be verified directly from the explicit expressions for $\rho$, $a$, $b$, $c$, $C$ in (\ref{coeff rho a b})--(\ref{big C in thm}) (see Remark \ref{symmremark}), the statement of Theorem \ref{thm:main results} follows also 
for $\theta \in [1, \infty)$. A similar argument applies to Theorem \ref{thm: all order expansion}. 
The upshot is that it is enough to prove Theorem \ref{thm:main results} and Theorem \ref{thm: all order expansion} for $\theta \in (0,1]$.

\subsubsection{Comparison with the approach of \cite{ClaeysGirSti}}
Even though our approach has the major advantage of opening up a path to the evaluation of the constant $C$, there are several disadvantages of integrating with respect to $\theta$ instead of with respect to $s$. First, in \cite{ClaeysGirSti} the authors were able to obtain expressions for the constants $a$ and $b$ at the hard edge not only for Muttalib--Borodin ensembles, but also for certain other limiting point processes arising from products of random matrices. This was feasible because $s$ is a common parameter in all of these models and the associated differential identities could be analyzed in a similar way in all cases. Since the parameter $\theta$ is not present in the other models, our method of deforming with respect to $\theta$ can only be applied in the case of the Muttalib--Borodin ensembles. Second, integration with respect to $\theta$ requires significantly more computational effort than integration with respect to $s$. This can be seen by taking the logarithm of the asymptotic formula \eqref{known result fredholm general theta} and differentiating the resulting equation with respect to $s$ and $\theta$ respectively:\footnote{From the analysis of \cite{ClaeysGirSti}, we can show that the error term in \eqref{known result fredholm general theta} is indeed differentiable and satisfies $\partial_{s} \bigO(s^{-\rho}) = \bigO(s^{-\rho -1})$ and $\partial_{\theta} \bigO(s^{-\rho}) = \bigO(s^{-\rho} \ln s)$ as $s \to + \infty$.}
\begin{align}\label{a b c in the diff identity for s lol}
\partial_{s} \ln \det\bigg(1 - \mathbb{K}\Big|_{[0,s]}\bigg) = & - 2  \rho a s^{2\rho -1} + \rho b s^{\rho-1} + \frac{c}{s} + \bigO(s^{-\rho -1}),  
	\\\nonumber
\partial_{\theta} \ln \det\bigg(1 - \mathbb{K}\Big|_{[0,s]}\bigg) = & - \frac{2a}{(1+\theta)^2} s^{2\rho} \ln{s} - (\partial_\theta a)s^{2\rho}
+ \frac{b}{(1+\theta)^2} s^\rho \ln{s}
+ (\partial_\theta b) s^{\rho}
	\\  \label{a b c in the diff identity for theta lol}
& + (\partial_\theta c) \ln{s}
+ \frac{\partial_\theta C}{C}
 + \bigO(s^{-\rho} \ln s), 
\end{align}
as $s \to + \infty$. Note that the differentiation with respect to $\theta$ generates additional terms proportional to $\ln s$. Moreover, the expansion in (\ref{a b c in the diff identity for theta lol}) involves the rather complicated first-order derivatives of $a,b,c$, and $C$ with respect to $\theta$. Third, it turns out that the differential identity with respect to $\theta$ is more intricate to analyze: While \eqref{log der s of fredholm} is expressed in terms of the first subleading term in the expansion of $Y(z)$ as $z \to + \infty$ (see \eqref{diff identity in s}), the analogous representation for \eqref{log der theta of fredholm} involves an integral whose integrand also contains the digamma function $\psi$ (see \eqref{first diff identity lemma}). The infinitely many poles of the digamma function $\psi$ (which we recall is defined as the log-derivative of $\Gamma$, see e.g. \cite[Eq. 5.2.2]{NIST}) complicate the analysis considerably.

For all the above reasons, we will in Section \ref{Section: little c} provide an independent derivation of the expression (\ref{little c in thm}) for $c$ by employing the differential identity in $s$. This derivation is significantly shorter than the derivation based on the differential identity in $\theta$ and it can also be generalized to other point processes. In particular, from the formulas we obtain we can straightforwardly determine the constants $c^{(1)}$ and $c^{(2)}$ of \cite[formula 1.20]{ClaeysGirSti} associated with point processes  at the hard edge of certain product random matrices, see Remark \ref{c1c2remark}. Furthermore, several important aspects of this alternative derivation of (\ref{little c in thm}) will be useful in the proofs of Theorem \ref{thm: all order expansion} and the expression (\ref{big C in thm}) for $C$. 

Finally, we note that the fact that the approach based on the differential identity in $\theta$ yields the same expressions (\ref{coeff rho a b}) and (\ref{little c in thm}) for the coefficients $a, b, c$ as the approach based on the differential identity in $s$ provides a nontrivial consistency check of our results.

\subsection{Organization of the paper}
In Section \ref{Section: ClaeysGirSti}, we introduce some notation and recall some results from \cite{ClaeysGirSti} that are needed for our analysis. 
In Sections \ref{psec} and \ref{Rsec}, we establish the existence of large $s$ asymptotics to all orders of three functions 
which play a pivotal role in the RH formulation. 
In Section \ref{Section: little c}, we use these expansions to prove Theorem \ref{thm: all order expansion} and to provide a first proof of the expression (\ref{little c in thm}) for $c$. 

In Section \ref{Section: diff identity in theta}, we derive a differential identity with respect to the parameter $\theta$. This identity expresses the $\theta$-derivative of $\ln \det( \left. 1-\mathbb{K} \right|_{[0,s]})$ as the sum of four integrals which we denote by $I_1$, $I_2$, $I_{3,K}$, and $I_{4,K}$. The arguments required to obtain the large $s$ asymptotics of these integrals are rather long and are presented in Sections \ref{I1sec}-\ref{I3I4sec}.

We complete the proof of Theorem \ref{thm:main results} in Section \ref{section: integration in theta} by substituting the above asymptotics into the differential identity in $\theta$ and integrating the resulting equation with respect to $\theta$. In addition to yielding the expression (\ref{big C in thm}) for $C$, this also provides independent derivations of the expressions (\ref{coeff rho a b}) and (\ref{little c in thm}) for the coefficients $a, b$, and $c$. 

The proofs of Propositions \ref{prop: constant d for rational theta} and \ref{dsymmprop} as well as the proofs of two lemmas (Lemma \ref{lemma: integrals for I1} and Lemma \ref{Zintlemma}) are presented in the four appendices.

\subsection*{Acknowledgements}
Support is acknowledged from the European Research Council, Grant Agreement No. 682537, the Swedish Research Council, Grant No. 2015-05430, the G\"oran Gustafsson Foundation, and the Ruth and Nils-Erik Stenb\"ack Foundation.
 
\section{Preliminary results from \cite{ClaeysGirSti}}\label{Section: ClaeysGirSti}
All the results presented in this section are taken from \cite{ClaeysGirSti}. We use the same notation as in \cite{ClaeysGirSti} except that we use $G$ to denote Barnes' $G$-function and $\mathcal{G}$ to denote the function which is denoted by $G$ in \cite{ClaeysGirSti}.
We start by recalling the RH problem for $Y$, which is central for this paper.

\subsubsection*{RH problem for Y}
\begin{itemize}
\item[(a)] $Y : \mathbb{C}\setminus (\gamma \cup \tilde{\gamma}) \to \mathbb{C}^{2 \times 2}$ is analytic, where $\gamma$ and $\tilde{\gamma}$ are the oriented contours shown in Figure \ref{fig: contours gamma}.
\item[(b)] The limits of $Y(z)$ as $z$ approaches $\gamma \cup \tilde{\gamma}$ from the left (+) and from the right (--) exist, are continuous on $\gamma \cup \tilde{\gamma}$, and are denoted by $Y_+$ and $Y_-$, respectively. Furthermore, they are related by
\begin{align*}
& Y_{+}(z) = Y_{-}(z) \begin{pmatrix}
1 & -s^{-z}F(z) \\
0 & 1
\end{pmatrix}, & & z \in \gamma, \\
& Y_{+}(z) = Y_{-}(z) \begin{pmatrix}
1 & 0 \\
s^{z}F(z)^{-1} & 1
\end{pmatrix}, & & z \in \tilde{\gamma},
\end{align*}
where $F$ is given by \eqref{def of F}.
\item[(c)] As $z \to \infty$, $Y$ admits the expansion
\begin{align*}
Y(z) = I + \frac{Y_{1}}{z} + \bigO(z^{-2}),
\end{align*}
where the $2 \times 2$ matrix $Y_{1}$ depends on $s$, $\alpha$, and $\theta$ but not on $z$.
\end{itemize}
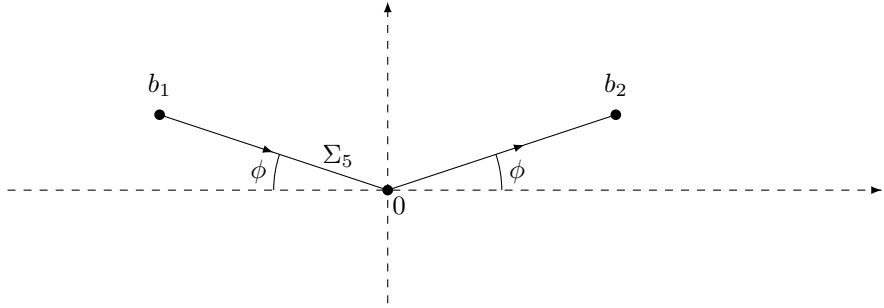
\begin{figure}
\begin{center}
\begin{tikzpicture}
\node at (0,0) {};
\fill (0,0) circle (0.07cm);
\node at (0.15,-0.2) {$0$};

\draw[dashed,->-=1,black] (0,-1.5) to [out=90, in=-90] (0,2.5);
\draw[dashed,->-=1,black] (-5,0) to [out=0, in=-180] (6.5,0);

			 \node at (-0.65,0.45) {$\Sigma_{5}$};

\draw[->-=0.6,black] (0,0)--(3,1);
\fill (3,1) circle (0.07cm);
\node at (3,1.4) {$b_{2}$};
\draw[-<-=0.5,black] (0,0)--(-3,1);
\fill (-3,1) circle (0.07cm);
\node at (-3,1.4) {$b_{1}$};

\draw ([shift=(0:1.5cm)]0,0) arc (0:18.43494:1.5cm);
\node at (1.7,0.25) {$\phi$};
\draw ([shift=(161.565:1.5cm)]0,0) arc (161.565:180:1.5cm);
\node at (-1.7,0.25) {$\phi$};
\end{tikzpicture}
\end{center}
\caption{\label{fig: support b1b2}The points $b_{1}$ and $b_{2}$ lie in the upper half-plane for $0 < \theta < 1$. The contour $\Sigma_{5}$ consists of the two line segments $[b_{1},0]$ and $[0,b_{2}]$.}
\end{figure}
The solution of the RH problem for $Y$ exists and is unique for any choice of the parameters $s > 0$, $\theta > 0$, and $\alpha > -1$, see \cite[below (1.18)]{ClaeysGirSti}. 

We choose the branch for $\ln F$ such that
\begin{align}\label{def of ln F}
\ln F(z) = \ln \Gamma \left( z + \frac{\alpha}{2} \right) - \ln \Gamma \left( \frac{\frac{\alpha}{2}+1-z}{\theta} \right),
\end{align}
where $z \mapsto \ln \Gamma(z)$ is the log-gamma function, which has a branch cut along $(-\infty,0]$. Therefore, $z \mapsto \ln F(z)$ has a branch cut along $(-\infty,-\frac{\alpha}{2}]\cup [1+\frac{\alpha}{2},+\infty)$. Following \cite{ClaeysGirSti}, we introduce a new complex variable $\zeta$ by 
\begin{align}\label{zetadef}
  z = i s^{\rho} \zeta + \tfrac{1}{2}.
\end{align}
As $s^{\rho} \zeta \to \infty$, we have the asymptotics
\begin{align}\nonumber
\ln F(i s^{\rho} \zeta + \tfrac{1}{2}) = &\; i s^{\rho} \ln(s) \zeta + i s^{\rho} [c_{1} \zeta \ln(i\zeta) + c_{2} \zeta \ln(-i\zeta) + c_{3} \zeta] 
	\\ \label{asymp for ln F} 
& + c_{4} \ln(s) + c_{5} \ln(i\zeta) + c_{6} \ln(-i \zeta) + c_{7} + \frac{c_{8}}{i s^{\rho} \zeta} + \bigO \left( \frac{1}{s^{2 \rho}\zeta^{2}} \right),  
\end{align}
where the logarithms on the right-hand side are defined using the principal branch. The real constants $c_{1},\ldots,c_{8}$ are computed in \cite[equation (3.12)]{ClaeysGirSti} and are given by\footnote{Here we have corrected a small typo in \cite[equation (3.12)]{ClaeysGirSti} in the expression for $c_{8}$, which has no consequence for the results of \cite{ClaeysGirSti} as $c_{8}$ does not play any role in the computation of $a$ and $b$.}
\begin{align}\nonumber
& c_{1} = 1, & & c_{2} = \frac{1}{\theta},  \\
& c_{3} = - \frac{\theta + 1 + \ln \theta}{\theta}, & & c_{4} = \frac{\theta + (\theta-1)\alpha - 1}{2(\theta +1)}, \nonumber \\
& c_{5} = \frac{\alpha}{2}, & & c_{6} = \frac{\theta - \alpha - 1}{2 \theta}, \nonumber \\
& c_{7} = - \frac{\theta - \alpha -1}{2 \theta} \ln \theta, & & c_{8} = \frac{3(1+\alpha)^{2}-7\theta-6\alpha \theta+3\alpha^2\theta+2\theta^2}{24\theta}. \label{def of c1...c8}
\end{align}
We also define $\mathcal{G}(\zeta)$ by
\begin{align}\label{calGdef}
\mathcal{G}(\zeta) = F(is^{\rho} \zeta + \tfrac{1}{2})e^{-is^{\rho}(\ln (s)\zeta - h(\zeta))},
\end{align}
where 
\begin{align}\label{hdef}
h(\zeta) = -c_{1} \zeta \ln(i \zeta) - c_{2} \zeta \ln(-i\zeta) - c_{3} \zeta.
\end{align}
The function $\mathcal{G}$ above is denoted by $G$ in \cite[equation (3.13)]{ClaeysGirSti}, while in this paper $G$ denotes Barnes' $G$-function. Note that $\mathcal{G}$ also depends on $s$, $\theta$ and $\alpha$, but we omit this dependence in the notation. Following \cite[Section 3.3]{ClaeysGirSti}, we define $b_{1},b_{2} \in \mathbb{C}$ by
\begin{align}\label{b1b2def}
b_{2} = - \overline{b_{1}} = |b_{2}|e^{i\phi}, \qquad \phi \in \left( - \frac{\pi}{2},\frac{\pi}{2} \right),
\end{align}
with 
\begin{align}\label{def of Re b2}
& -\re b_{1} = \re b_{2} = 2 \left( \frac{c_{2}}{c_{1}} \right)^{-\frac{c_{2}-c_{1}}{2(c_{2}+c_{1})}}e^{-\frac{c_{1}+c_{2}+c_{3}}{c_{1}+c_{2}}} = 2 \theta^{\frac{3-\theta}{2(1+\theta)}} > 0, 	
	 \\
& \sin \phi = \frac{c_{2}-c_{1}}{c_{2}+c_{1}} = \frac{1-\theta}{1+\theta}. \nonumber
\end{align}
Thus $\phi \geq 0$ for $0 < \theta \leq 1$ while $\phi \leq 0$ for $\theta \geq 1$. 
As explained in Section \ref{twocasessubsubsec}, it is enough to prove Theorem \ref{thm:main results} and Theorem \ref{thm: all order expansion} for $\theta \in (0,1]$ thanks to the symmetry (\ref{symmetry}). Therefore we will henceforth restrict ourselves to the case $0 < \theta \leq 1$, for which we have $\phi \in [0, \pi/2)$.

In the steepest descent analysis of the RH problem for $Y$, the so-called $g$-function plays an important role. Using this function, certain jumps of the RH problem can be made exponentially small as $s \to +\infty$. The $g$-function has a jump along the contour $\Sigma_{5}$, which consists of the two line segments $[b_{1},0]\cup[0,b_{2}]$ oriented to the right, see Figure \ref{fig: support b1b2}, and is defined as follows. Define the function $r(\zeta)$ by 
\begin{align}\label{rdef}
r(\zeta) = [(\zeta-b_{1})(\zeta-b_{2})]^{\frac{1}{2}},
\end{align}
where the branch is such that $r$ is analytic in $\mathbb{C}\setminus \Sigma_{5}$ and $r(\zeta) \sim \zeta$ as $\zeta \to \infty$. The second derivative of the $g$-function is given by
\begin{equation}\label{def of g''}
g''(\zeta) = - i \frac{c_{1}+c_{2}}{2} \left( \frac{1}{\zeta} - \frac{1}{r(\zeta)} + \frac{i \im b_{2}}{\zeta r(\zeta)} \right).
\end{equation}
Hence
\begin{align*}
g_{+}''(\zeta) + g_{-}''(\zeta) = - i \frac{c_{1}+c_{2}}{\zeta}, \qquad \zeta \in \Sigma_{5},
\end{align*}
and, as $\zeta \to \infty$, 
\begin{align*}
g''(\zeta) = \frac{2g_{1}}{\zeta^{3}} + \bigO(\zeta^{-4}), \qquad \text{where} \quad g_{1} = \frac{i (\re b_{2})^{2}(c_{1}+c_{2})}{8}.
\end{align*}
The $g$-function is then obtained by
\begin{align*}
g'(\zeta) = \int_{\infty}^{\zeta} g''(\xi)d\xi, \qquad g(\zeta) = \int_{\infty}^{\zeta} g'(\xi)d\xi,
\end{align*}
where the integration paths lie in the complement of $\Sigma_{5}$. The $g$-function is analytic on $\mathbb{C} \setminus \Sigma_{5}$ and has the following jump across $\Sigma_5$:
\begin{align}\label{jump relation g}
g_{+}(\zeta) + g_{-}(\zeta) - i h(\zeta) + \ell = 0, \qquad \zeta \in \Sigma_{5},
\end{align}
where $\ell = ih(b_{1}) - 2 \int_{\infty}^{b_{1}}g'(\xi)d\xi$.
 
\subsection{Steepest descent analysis}
Let $\sigma_{1}$ and $\sigma_{3}$ denote the first and third Pauli matrices given by
\begin{align}\label{Pauli matrices}
& \sigma_{1} = \begin{pmatrix}
0 & 1 \\
1 & 0
\end{pmatrix}, \qquad \sigma_{3} = \begin{pmatrix}
1 & 0 \\
0 & -1
\end{pmatrix}.
\end{align}
The steepest descent analysis of the RH problem for $Y$ involves a sequence of transformations $Y \mapsto U \mapsto T \mapsto S \mapsto R$.
The first transformation $Y \mapsto U$ is defined by
\begin{equation}\label{def of U}
U(\zeta) = s^{\frac{\sigma_{3}}{4}} Y \left( i s^{\rho} \zeta + \tfrac{1}{2} \right) s^{-\frac{\sigma_{3}}{4}}.
\end{equation}
The $2\times 2$ matrix-valued function $U$ is analytic on $\mathbb{C}\setminus (\gamma_{U} \cup \tilde{\gamma}_{U})$, where 
\begin{equation*}
\gamma_{U} = \{\zeta \in \mathbb{C}: i s^{\rho} \zeta + \tfrac{1}{2} \in \gamma \}, \qquad 
\tilde{\gamma}_{U} = \{\zeta \in \mathbb{C}: i s^{\rho} \zeta + \tfrac{1}{2} \in \tilde{\gamma} \},
\end{equation*}
see also \cite[Figure 2]{ClaeysGirSti}. 
Let $\{\Sigma_i\}_1^4$ denote the contours defined by 
\begin{align}\label{Sigmaidef}
\Sigma_{2} = - \overline{\Sigma_{1}} = b_{2} + e^{i(\phi + \epsilon)}\R_{\geq 0}, \qquad \Sigma_{4} = - \overline{\Sigma_{3}} = b_{2}+e^{-i \epsilon}\R_{\geq 0},
\end{align}
with $0 < \epsilon < \pi/10$ and oriented from left to right, see Figure \ref{fig: Sigma 1,2,3,4}. 
Recall that $\Sigma_{5} = [b_{1},0]\cup[0,b_{2}]$.

The second transformation $U \mapsto T$ consists of deforming the contour of the RH problem by considering an analytic continuation of $U$ such that $T$ is analytic in $\mathbb{C}\setminus \cup_{i=1}^{5} \Sigma_{i}$; we refer to \cite[Section 3.2]{ClaeysGirSti} for details.
\begin{figure}
\begin{center}
\begin{tikzpicture}
\node at (0,0) {};
\fill (0,0) circle (0.07cm);
\node at (0.15,-0.2) {$0$};

\draw[dashed,->-=1,black] (0,-1.5) to [out=90, in=-90] (0,2.5);
\draw[dashed,->-=1,black] (-5,0) to [out=0, in=-180] (6.5,0);

\draw[->-=0.6,black] (0,0)--(3,1);
\fill (3,1) circle (0.07cm);
\node at (3,1.4) {$b_{2}$};
\draw[dashed,black] (3,1)--(6,2);
\draw[->-=0.6,black] (3,1)--($(3,1)+(28.43:3.2)$);
\draw[->-=0.6,black] (3,1)--($(3,1)+(-10:2.9)$);
\node at (4.5,2.2) {$\Sigma_{2}$};
\node at (4.5,0.4) {$\Sigma_{4}$};

\draw[-<-=0.5,black] (0,0)--(-3,1);
\fill (-3,1) circle (0.07cm);
\node at (-3,1.4) {$b_{1}$};
\draw[dashed,black] (-3,1)--(-6,2);
\draw[-<-=0.6,black] (-3,1)--($(-3,1)+(180-28.43:3.2)$);
\draw[-<-=0.6,black] (-3,1)--($(-3,1)+(-180+10:2.9)$);
\node at (-4.5,2.2) {$\Sigma_{1}$};
\node at (-4.5,0.4) {$\Sigma_{3}$};

 \node at (-0.65,0.45) {$\Sigma_{5}$};

\draw ([shift=(0:1.5cm)]0,0) arc (0:18.43494:1.5cm);
\node at (1.7,0.25) {$\phi$};
\draw ([shift=(161.565:1.5cm)]0,0) arc (161.565:180:1.5cm);
\node at (-1.7,0.25) {$\phi$};

\draw ([shift=(18.43494:2.5cm)]3,1) arc (18.43494:28.43494:2.5cm);
\node at (5.6,2.1) {$\epsilon$};
\end{tikzpicture}
\end{center}
\caption{\label{fig: Sigma 1,2,3,4}The jump contour $\cup_{i=1}^{5} \Sigma_{i}$ for the RH problem for $T$.}
\end{figure}
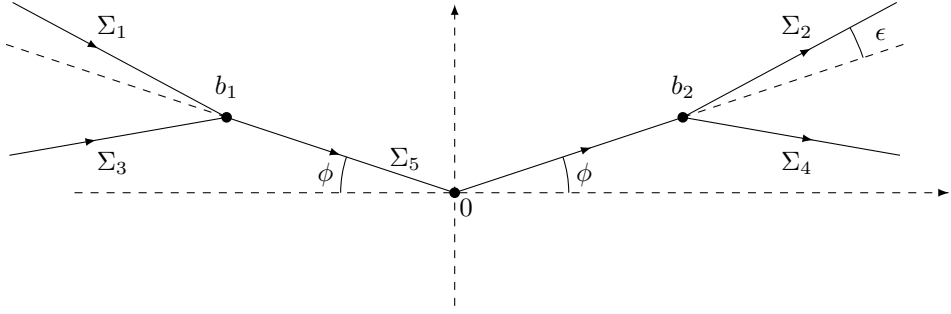
The third transformation $T \mapsto S$ uses the $g$-function and is defined by
\begin{equation}
S(\zeta) = e^{\frac{\ell}{2}s^{\rho}\sigma_{3}}T(\zeta) e^{-s^{\rho}g(\zeta) \sigma_{3}}e^{-\frac{\ell}{2}s^{\rho}\sigma_{3}}.
\end{equation}

The remainder of the steepest descent analysis of \cite{ClaeysGirSti} consists of finding good approximations of $S$ in different regions of the complex plane. 
Define the function $\gamma(\zeta)$ by
\begin{align*}
\gamma(\zeta) = \left( \frac{\zeta-b_{1}}{\zeta -b_{2}} \right)^{1/4},
\end{align*}
where the branch is such that $\gamma(\zeta)$ is an analytic function of $\zeta \in \mathbb{C} \setminus \Sigma_{5}$ and $\gamma(\zeta) \sim 1$ as $\zeta \to \infty$. 
Define also the function $p:\mathbb{C}\setminus \Sigma_{5} \to \C$ by
\begin{align}\label{def of p}
& p(\zeta) = - \frac{r(\zeta)}{2\pi i} \int_{\Sigma_{5}} \frac{\ln \mathcal{G}(\xi)}{r_{+}(\xi)}\frac{d\xi}{\xi-\zeta},
\end{align}
where the branch for $\ln \mathcal{G}$ is such that
\begin{equation}\label{def of ln mathcal G}
\ln \mathcal{G}(\zeta) = \ln F(i s^{\rho} \zeta + \tfrac{1}{2}) - i s^{\rho} \left( \ln(s) \zeta - h(\zeta) \right)
\end{equation}
with $\ln F$ defined as in \eqref{def of ln F}. 
Outside small neighborhoods of $b_{1}$ and $b_{2}$, $S$ is well approximated by the global parametrix $P^{\infty}$ defined by
\begin{align}\label{def of Pinf}
&P^{\infty}(\zeta) = e^{-p_{0} \sigma_{3}}Q^{\infty}(\zeta)
e^{p(\zeta)\sigma_{3}} \quad \text{with} \quad Q^{\infty}(\zeta) = \begin{pmatrix}
\frac{\gamma(\zeta)+\gamma(\zeta)^{-1}}{2} & \frac{\gamma(\zeta)-\gamma(\zeta)^{-1}}{2i} \\
\frac{\gamma(\zeta)-\gamma(\zeta)^{-1}}{-2i} & \frac{\gamma(\zeta)+\gamma(\zeta)^{-1}}{2}
\end{pmatrix}.
\end{align}
The function $p$ satisfies $p(\zeta) = \overline{p(-\bar{\zeta})}$ and
\begin{align}
& p_{+}(\zeta) + p_{-}(\zeta) = - \ln \mathcal{G}(\zeta), & & \zeta \in \Sigma_{5}, \\
& p(\zeta) = p_{0} + \frac{p_{1}}{\zeta} + \bigO(\zeta^{-2}), & & \zeta \to \infty, \label{def of expansion of p with p0 p1}
\end{align}
where the constants $p_{0} \in \R$ and $p_{1} \in i\R$ are given by
\begin{align*}
p_{0} & = \frac{1}{2\pi i} \int_{\Sigma_{5}} \frac{\ln \mathcal{G}(\xi)}{r_{+}(\xi)}d\xi, \\
p_{1} & = - \frac{b_{1}+b_{2}}{4\pi i} \int_{\Sigma_{5}} \frac{\ln \mathcal{G}(\xi)}{r_{+}(\xi)}d\xi + \frac{1}{2\pi i} \int_{\Sigma_{5}} \frac{\xi \ln \mathcal{G}(\xi)}{r_{+}(\xi)}d\xi
 = \frac{1}{2\pi i} \int_{\Sigma_{5}} \frac{(\xi - i \im (b_{2})) \ln \mathcal{G}(\xi)}{r_{+}(\xi)}d\xi.
\end{align*}

\subsubsection{The local parametrix $P$}
Near the points $b_{1}$ and $b_{2}$, $S$ is no longer well approximated by $P^{\infty}$, and we need to construct local approximations to $S$ (also called local parametrices and denoted by $P$). Following \cite{ClaeysGirSti}, these local parametrices are built out of Airy functions and are defined in small open disks $\mathbb{D}_{\delta}(b_1)$ and $\mathbb{D}_{\delta}(b_2)$ centered at $b_{1}$ and $b_{2}$, respectively:
\begin{align*}
\mathbb{D}_{\delta}(b_{j}) = \{z \in \mathbb{C} : |z-b_{j}|< \delta\}, \qquad j = 1,2,
\end{align*}
for some sufficiently small radius $\delta > 0$ which is independent of $s$. Furthermore, $P$ satisfies the following matching condition with $P^{\infty}$ on the boundary $\partial \mathbb{D}_{\delta}(b_{1}) \cup \partial \mathbb{D}_{\delta}(b_{2})$:
$$e^{p_{0}\sigma_{3}}P(\zeta) = \big(I + 
\bigO(s^{-\rho})\big)  e^{p_{0}\sigma_{3}}P^{\infty}(\zeta), \qquad s \to + \infty,$$
uniformly for $\zeta \in \partial \mathbb{D}_{\delta}(b_{1}) \cup \partial \mathbb{D}_{\delta}(b_{2})$. 
The local parametrix $P$ obeys the symmetry 
\begin{equation}\label{Psymm}
P(\zeta) = \overline{P(- \overline{\zeta})}, \qquad \zeta \in \mathbb{D}_{\delta}(b_{1}) \cup \mathbb{D}_{\delta}(b_{2}),
\end{equation}
and therefore we can restrict attention to the construction of $P$ in $\mathbb{D}_{\delta}(b_{1})$. There are a few minor typos in \cite{ClaeysGirSti}: the factors $\sqrt{2\pi}$ in \cite[equations (3.57)--(3.59)]{ClaeysGirSti} should be $2 \sqrt{\pi}$ and the signs of the exponential factors in \cite[equations (3.63), (3.65), (3.67)]{ClaeysGirSti} should be modified. These typos have no repercussion on the results of \cite{ClaeysGirSti}, but will play a role for us. In what follows, we therefore give the definition of $P$ in detail. First, define the complex-valued functions $\{y_{j}(\zeta)\}_1^3$ by
\begin{align*}
y_{j}(\zeta) = e^{\frac{2\pi i j}{3}} \mbox{Ai}(e^{\frac{2\pi i j}{3}}\zeta), \qquad j = 0,1,2,
\end{align*}
and let the $2\times 2$-matrix valued functions $\{A_j(\zeta)\}_1^3$ be given by
\begin{align}
& A_{1}(\zeta) = -2 i \sqrt{\pi} \begin{pmatrix}
-y_{2}(\zeta) & -y_{0}(\zeta) \\
-y_{2}'(\zeta) & -y_{0}'(\zeta)
\end{pmatrix}, \label{def of A1} \\
& A_{2}(\zeta) = -2 i \sqrt{\pi} \begin{pmatrix}
-y_{2}(\zeta) & y_{1}(\zeta) \\
-y_{2}'(\zeta) & y_{1}'(\zeta)
\end{pmatrix}, \label{def of A2} \\
& A_{3}(\zeta) = -2 i \sqrt{\pi} \begin{pmatrix}
y_{0}(\zeta) & y_{1}(\zeta) \\
y_{0}'(\zeta) & y_{1}'(\zeta)
\end{pmatrix}. \label{def of A3}
\end{align}
These functions satisfy
\begin{align*}
& A_{1}(\zeta) = A_{2}(\zeta) \begin{pmatrix}
1 & -1 \\
0 & 1 
\end{pmatrix}, \qquad A_{2}(\zeta) = A_{3}(\zeta) \begin{pmatrix}
1 & 0 \\
1 & 1
\end{pmatrix}, \qquad 
A_{1}(\zeta) = A_{3}(\zeta) \begin{pmatrix}
1 & -1 \\
1 & 0
\end{pmatrix}.
\end{align*}
Moreover, 
\begin{align}\label{weak asymp for Ak}
A_{k}(\zeta) = \zeta^{-\frac{\sigma_{3}}{4}} \begin{pmatrix}
1 & i \\
1 & -i
\end{pmatrix} \Big[I + \bigO(\zeta^{-3/2}) \Big]e^{\frac{2}{3}\zeta^{3/2}\sigma_{3}}
\end{align}
as $\zeta \to \infty$ in the sector $S_{k}$ for $k = 1,2,3$, with
\begin{align}\label{Skdef}
S_{k} = \left\{ \zeta \in \mathbb{C}: \frac{2k - 3}{3}\pi + \delta \leq \arg \zeta \leq  \frac{2k+1}{3}\pi - \delta \right\}, \qquad k = 1,2,3,
\end{align}
and the branches of the complex powers in (\ref{weak asymp for Ak}) are such that $\zeta^u = e^{u\ln|\zeta| + iu \arg \zeta}$ where $\arg \zeta$ belongs to $(-\pi/3,\pi)$, $(\pi/3,5\pi/3)$, and $(\pi,7\pi/3)$ for $\zeta$ in $S_1, S_2, S_3$, respectively.
The local parametrix $P$ is defined for $\zeta \in \mathbb{D}_{\delta}(b_{1}) \setminus \cup_{i=1}^{5} \Sigma_{i}$ by
\begin{align}\label{def of P}
P(\zeta) = E(\zeta)A_{k} \big(s^{\frac{2}{3}\rho}f(\zeta)\big) e^{-s^{\rho}q(\zeta) \sigma_{3}} \mathcal{G}(\zeta)^{-\frac{\sigma_{3}}{2}}, \qquad \zeta \in [k], \;\; k = 1,2,3,
\end{align}
where $[k]$, $k = 1,2,3$, denote the three components of $\mathbb{D}_{\delta}(b_{1}) \setminus \cup_{i=1}^{5} \Sigma_{i}$ as shown in \cite[Figure 4]{ClaeysGirSti}, $q$ is the analytic function on $\mathbb{D}_\delta(b_{1})\setminus \Sigma_{5}$ given by
\begin{align} \label{def of q}
q(\zeta) = g(\zeta) - \frac{i}{2}h(\zeta) + \frac{\ell}{2},
\end{align}
the function $f$ is defined by
\begin{align}\label{relation f and q}
f(\zeta) = \left( \frac{3}{2}q(\zeta) \right)^{\frac{2}{3}},
\end{align}
and $E$ denotes the $2 \times 2$-matrix valued function analytic on $\mathbb{D}_\delta(b_{1})$ defined by
\begin{align}\label{def of E local param}
E(\zeta) = P^{\infty}(\zeta)\mathcal{G}(\zeta)^{\frac{\sigma_{3}}{2}}\begin{pmatrix}
1 & i \\ 1 & -i
\end{pmatrix}^{-1} \Big( s^{\frac{2}{3}\rho}f(\zeta) \Big)^{\frac{\sigma_{3}}{4}}.
\end{align}
It is shown in \cite[equation (3.71)]{ClaeysGirSti} that, as $\zeta \to b_1$,
\begin{align}\label{asymptotics q}
q(\zeta) = -\frac{2}{3} \frac{(c_1+c_2)}{\sqrt{2}} \frac{\sqrt{|\re b_1|}}{b_1} (\zeta -b_1)^{\frac{3}{2}} +\bigO\Big( (\zeta -b_1)^{\frac{5}{2}}\Big),
\end{align}
where the branch cut for $(\zeta-b_1)^{\frac{3}{2}}$ runs along $\Sigma_5$ and $(\zeta-b_1)^{\frac{3}{2}} > 0$ for $\zeta-b_1 > 0$. Hence $f$ is a conformal map from $\mathbb{D}_\delta(b_{1})$ to a neighborhood of $0$ such that $\arg f'(b_1) = 2\phi/3 \in [0, \pi/3)$.

\subsubsection{The solution $R$}
In view of \eqref{asymp for ln F} and (\ref{def of ln mathcal G}), the function $\mathcal{G}(\zeta)$ is not bounded as $s^{\rho}\zeta \to \infty$. Therefore, in the definition of the last transformation $S \mapsto R$, we need to multiply by a conjugation matrix $e^{p_{0}\sigma_{3}}$ in order for $R$ to be uniformly bounded on $\mathbb{C}$.\footnote{Note that the conjugation by $e^{p_{0}\sigma_{3}}$ only affects the off-diagonal elements of $R$. Thus, even though this conjugation was not present in \cite{ClaeysGirSti}, this does not affect the results of that paper as they only depend on the $(2,2)$ element of $R$.} More precisely, we define $R$ by
\begin{align} \label{def of R}
R(\zeta) = e^{p_0 \sigma_3} S(\zeta) \times \begin{cases}
P(\zeta)^{-1} e^{- p_0 \sigma_3}, &\text{ if $\zeta \in \Done \cup  \Dtwo$},
\\
P^{\infty}(\zeta)^{-1} e^{-p_0 \sigma_3}, &\text{ elsewhere}.
\end{cases}
\end{align}
Then $R(\zeta)$ is analytic for $\zeta \in \C \setminus \Gamma_{R}$ where $\Gamma_R$ consists of the parts of $\cup_{i=1}^5 \Sigma_i$ lying outside the disks $\mathbb{D}_\delta(b_j)$, $j = 1,2$, as well as the two clockwise circles $\partial \mathbb{D}_\delta(b_j)$, $j = 1,2$, see Figure \ref{fig: Gamma R}.
We will show in Section \ref{Rsec} that $R$ satisfies a small norm RH problem and that 
\begin{align}
& R(\zeta) = I + \frac{R_{1}}{\zeta} + \bigO(\zeta^{-2}) \qquad \mbox{as } \zeta \to \infty, \label{def of R1} 
\end{align}
where the matrix $R_{1}$ possesses the asymptotics
\begin{align}\label{def of R1^1}
  R_{1} = \frac{R_{1}^{(1)}}{s^{\rho}} + \bigO(s^{-2\rho}) \qquad \mbox{as } s \to + \infty, 
\end{align}
for a certain matrix $R_{1}^{(1)}$ independent of $s$ and $\zeta$.
\begin{figure}
\begin{center}
\begin{tikzpicture}
\node at (0,0) {};
\fill (0,0) circle (0.07cm);
\node at (0.15,-0.2) {$0$};

\draw[dashed,->-=1,black] (0,-1.5) to [out=90, in=-90] (0,2.5);
\draw[dashed,->-=1,black] (-5,0) to [out=0, in=-180] (6.5,0);

\draw[->-=0.6,black] (0,0)--($(3,1)+(18.435+180:0.6)$);
\fill (3,1) circle (0.05cm);
\node at (3,1.3) {$b_{2}$};
\draw[->-=0.6,black] ($(3,1)+(18.435:0.6)$)--($(3,1)+(28.435:3.2)$);
\draw[->-=0.6,black] ($(3,1)+(-10:0.6)$)--($(3,1)+(-10:2.9)$);

\draw[-<-=0.5,black] (0,0)--($(-3,1)+(-18.435:0.6)$);
\fill (-3,1) circle (0.05cm);
\node at (-3,1.3) {$b_{1}$};
\draw[-<-=0.6,black] ($(-3,1)+(180-28.435:0.6)$)--($(-3,1)+(180-28.43:3.2)$);
\draw[-<-=0.6,black] ($(-3,1)+(-180+10:0.6)$)--($(-3,1)+(-180+10:2.9)$);

\draw[-<-=0.24] (3,1) circle (0.6cm);
\draw[-<-=0.24] (-3,1) circle (0.6cm);
\end{tikzpicture}
\end{center}
\caption{\label{fig: Gamma R}The contour $\Gamma_{R}$ for the RH problem $R$.}
\end{figure}
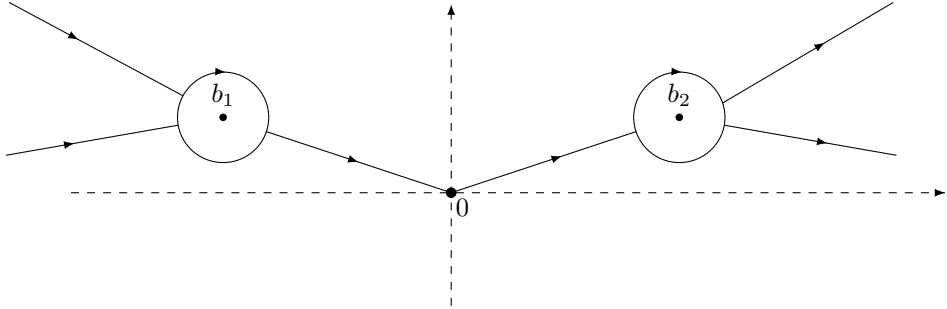

\subsection{Differential identity in $s$}
It was proved in \cite{ClaeysGirSti} that, for all $s > 0$, 
\begin{align}\nonumber
\partial_{s} \ln \det \big( 1-\mathbb{K}\big|_{[0,s]} \big) & \; = - \frac{1}{s}(Y_{1})_{2,2} 
	\\ \label{diff identity in s}
&\; = -\frac{(\re b_2)^2 (c_1+c_2)}{8} s^{2\rho-1}-is^{\rho-1}\big( -p_1(s)+(R_1(s))_{2,2} \big).
\end{align}
Furthermore, it was shown in \cite[Section 4.3]{ClaeysGirSti} that
\begin{align}
\partial_{s} \ln \det \big( 1-\mathbb{K}\big|_{[0,s]} \big) = & - \frac{(\re (b_{2}))^{2}(c_{1}+c_{2})}{8}s^{2 \rho -1}- \left( -c_{5} |b_{2}|  + \frac{c_{5}+c_{6}}{2}(|b_{2}|-\im(b_{2})) \right) s^{\rho-1} \nonumber \\
& + \frac{-\mathcal{K} + (R_{1}^{(1)})_{2,2}}{is} + \bigO(s^{-\rho-1}), \qquad \mbox{as } s \to + \infty, \label{a b but not c}
\end{align} 
where $\mathcal{K}$ is defined via the expansion (see \cite[Eq. (4.15)]{ClaeysGirSti})
\begin{equation}\label{p1 asymp from Claeys}
p_1=-ic_{5}|b_{2}| + i \frac{c_{5}+c_{6}}{2}(|b_{2}|-\im{b_2}) +\frac{\mathcal{K}}{s^\rho} + \bigO\bigg( \frac{1}{s^{2\rho}} \bigg), \qquad  s \to + \infty.
\end{equation}
Integration of (\ref{a b but not c}) yields the expressions in (\ref{coeff rho a b}) for the first two coefficients $a$ and $b$. Moreover, comparing (\ref{a b but not c}) with \eqref{a b c in the diff identity for s lol}, we infer that the coefficient $c$ can be expressed as
\begin{align}\label{expression for c in terms of mathcal K and RIp1p intro}
c = \frac{-\mathcal{K} + (R_{1}^{(1)})_{2,2}}{i}.
\end{align}
Thus, to compute $c$ it is enough to compute $\mathcal{K}$ and the $(2,2)$ entry of $R_{1}^{(1)}$.

\section{Asymptotics of $\mathcal{G} (\zeta)$ and $p(\zeta)$}\label{psec}
In this section, we establish asymptotic formulas for the functions $\mathcal{G} (\zeta)$ and $p(\zeta)$ defined in (\ref{calGdef}) and (\ref{def of p}) as $s \to + \infty$ with $\zeta$ such that $s^{\rho}\zeta \to \infty$. More precisely, we will prove that $\ln \mathcal{G} (\zeta)$ and $p(\zeta)$ admit expansions to all orders in inverse powers of $s^{\rho}\zeta$ and we will compute the coefficients of the expansion for $p(\zeta)$ explicitly up to and including the term of order $s^{-\rho}\zeta^{-1}$ (this term plays a role in the derivation of the expressions for both $c$ and $C$). The results are summarized in the followings two propositions whose proofs are presented in Sections \ref{Gpropproofsubsec} and \ref{ppropproofsubsec}, respectively. We let $\{c_j = c_j(\theta, \alpha)\}_1^8$ and $\{b_j = b_j(\theta, \alpha)\}_1^2$ be the constants defined in (\ref{def of c1...c8}) and (\ref{b1b2def}), respectively.

\begin{proposition}[Asymptotics of $\ln \mathcal{G} (\zeta)$]\label{Gprop}
Let $N\ge1$ be an integer. Let $\alpha > -1$ and $\theta \in (0, 1]$. 
There exist coefficients $\{\mathcal{G}_n = \mathcal{G}_n(\theta, \alpha)\}_1^N \subset \C$ such that the function $\mathcal{G}$ defined in (\ref{calGdef}) satisfies the asymptotic expansion
\begin{align}\label{full expansion G}
 \ln \mathcal{G} (\zeta) = c_4 \ln s + c_5\ln(i\zeta)+c_6\ln(-i\zeta)+c_7 + \frac{\mathcal{G}_{1}}{s^{\rho}\zeta} + \sum_{n=2}^{N} \frac{\mathcal{G}_n}{(s^\rho \zeta)^n} 
 + \bigO \bigg( \frac{1}{(s^\rho \zeta)^{N+1}}\bigg) 
 \end{align}
as $s^\rho \zeta \to  \infty$ uniformly for $\theta$ in compact subsets of $(0,1]$ and $\zeta \in \C$ such that $|\arg (\zeta) - \tfrac{\pi}{2}| > \epsilon$ and $|\arg (\zeta) + \tfrac{\pi}{2}| > \epsilon$ for any fixed $\epsilon  > 0$. The first coefficient is given by $\mathcal{G}_1= -ic_8$. 
\end{proposition}

\begin{proposition}[Asymptotics of $p(\zeta)$]\label{pprop}
Let $N\ge1$ be an integer. Let $\alpha > -1$ and $\theta \in (0, 1]$. 
There exist holomorphic functions $\mathcal{A}_n:\C \setminus \Sigma_5 \to \C$, $n=1,\ldots N$, with 
$\mathcal{A}_n(\zeta) = \bigO(\zeta^n)$ as $\zeta \to \infty$, such that
\begin{align} \nonumber
p(\zeta)= & -\frac{c_4}{2} \ln(s) -\frac{c_5}{2} \ln(i\zeta)-\frac{c_6}{2}\ln(-i\zeta)-\frac{c_7}{2} + \frac{\mathcal{R}(\zeta)}{2}
 	\\ \label{full expansion of p}
& + \frac{\mathcal{A}_{1}(\zeta)}{s^{\rho} \zeta} 
+ \sum_{n=2}^{N} \frac{\mathcal{A}_n(\zeta)}{(s^\rho \zeta)^n} 
+ \bigO \bigg( \frac{1}{(s^\rho \zeta)^{N+1}}\bigg) + \bigO \bigg( \frac{1}{s^{(N+1)\rho}}\bigg) \qquad \mbox{as } s \to +\infty,
\end{align}
uniformly with respect to $\zeta \in \C \setminus \Sigma_5$ such that $s^{\rho}\zeta \to \infty$ and $\theta$ in compact subsets of $(0,1]$, where the functions $\mathcal{R}(\zeta)$ and $\mathcal{A}_1(\zeta)$ are given by
\begin{align}\label{def of mathcalR}
\mathcal{R}(\zeta)&= -c_5  \ln\bigg(\frac{|b_2|^2 + i\zeta \im{b_2} - i|b_2| r(\zeta)}{(r(\zeta) + \zeta - i\im{b_2})\zeta}\bigg)-c_6\ln\bigg(\frac{|b_2|^2 + i\zeta \im{b_2} +i|b_2| r(\zeta)}{(r(\zeta) + \zeta - i\im{b_2})\zeta}\bigg)
\end{align}
and
\begin{align}\label{def of mathcalA1}
\mathcal{A}_1(\zeta) = \frac{ic_8}{2}+\frac{c_8-\frac{3\alpha^2-1}{12}}{2|b_2|}r(\zeta).
\end{align}
\end{proposition}
\begin{remark}\upshape
The expansion in \eqref{full expansion of p} is well-defined also for $\zeta \in i\R \setminus \{ 0 \}$ even though several of the coefficients have jumps across the imaginary axis. Indeed, it can be seen from (\ref{def of mathcalR}) (and more easily from the integral representation \eqref{calRexpression} of $\mathcal{R}$) that $\mathcal{R}$ has the following jump across the imaginary axis:
$$\mathcal{R}_+(\zeta) - \mathcal{R}_-(\zeta) = \begin{cases} 2\pi i c_5, & \zeta \in (i\infty, 0), \\
2\pi i c_6, & \zeta \in (-i\infty, 0),
\end{cases}$$
where $(i\infty, 0)$ and $(-i\infty, 0)$ are oriented towards the origin. It follows that the function
\[
-\frac{c_5}{2} \ln(i\zeta)-\frac{c_6}{2}\ln(-i\zeta)-\frac{c_7}{2} + \frac{\mathcal{R}(\zeta)}{2}
\]
has no jump across the imaginary axis and hence extends to an analytic function on $\C \setminus \Sigma_5$. 
\end{remark}

\begin{remark}\upshape
The expansion of $\ln \mathcal{G}(\zeta)$ as $s^{\rho} \zeta \to \infty$ up to and including the term of order $s^{-\rho}\zeta^{-1}$ is easily obtained from \eqref{asymp for ln F} and \eqref{def of ln mathcal G}, see \cite[Eq. (3.15)]{ClaeysGirSti}. The extension of this expansion to all orders is not straightforward and is the content of Proposition \ref{Gprop}. 
\end{remark}

\begin{remark}\upshape
The assumption that $0 < \theta \leq 1$ implies that $\phi = \arg b_{2}$ satisfies $0\leq \phi<\frac{\pi}{2}$, see Figure \ref{fig: support b1b2}. 
\end{remark}

\subsection{Proof of Proposition \ref{Gprop}}\label{Gpropproofsubsec}

We will employ the following exact representation for $\ln \Gamma(z)$ (see \cite[Eq. (6.34) with $r=N$ and Eq. (6.38)]{V2016}): 
\begin{align}\label{stirlingwithremainder}
\ln \Gamma(z) & = z\ln z-z-\frac{1}{2} \ln(z) + \ln\sqrt{2\pi} +\frac{1}{12z} +\sum_{n=2}^{N}\frac{B_{2n}}{(2n-1)2n} \frac{1}{z^{2n-1}} + \mathcal{D}_{N}(z), 
	\\ \label{calDNdef}
\mathcal{D}_{N}(z) & = -\frac{1}{2N+1}\int_{0}^{\infty} \frac{B_{2N+1}(\{t\})}{(z+t)^{2N+1}} dt,
\end{align}
which is valid for $|\arg z|<\pi$, with $N$ an arbitrary (but fixed) positive integer and where $\{ t \}$ denotes the fractional part of $t$, i.e., $\{t\} = t - \lfloor t \rfloor$ where $\lfloor t \rfloor$ is the largest integer smaller than or equal to $t$.
Here $B_{n}$ is the $n$th Bernoulli number and $B_{N}(x)$ the $N$th Bernoulli polynomial given by (see e.g. \cite[p. 804]{AS83})
\begin{equation*}
B_{N}(x) = \sum_{n=0}^{N} \binom{N}{n} B_{N-n}x^n.
\end{equation*}
The first terms on the right-hand side of \eqref{stirlingwithremainder} are the same as in Stirling's approximation formula; however \eqref{stirlingwithremainder} is an exact identity which is valid for all $z \in \C$ such that $|\arg z| < \pi$. It is straightforward to verify that
(see \cite[last equation on page 78]{V2016} for details)
\begin{align}\label{remainderO}
\mathcal{D}_N(z)= \bigO(z^{-2N-1}), \qquad  z \to \infty, \; |\arg z| < \pi - \epsilon,
\end{align}
for any fixed $\epsilon > 0$. 
Using the short-hand notation
\begin{equation}\label{def of x and y}
x(\xi)=x(\xi,s,\theta, \alpha)=is^\rho \xi + \frac{1+\alpha}{2}, \qquad y(\xi)=y(\xi,s,\theta,\alpha)=\frac{\frac{1+\alpha}{2}-is^\rho\xi}{\theta},
\end{equation}
we have
\begin{align}
& x(\xi) \leq 0 \qquad  \mbox{if and only if} \qquad \xi \in \big[ \tfrac{1+\alpha}{2}is^{-\rho},i\infty \big), \label{x negat} \\
& y(\xi) \leq 0 \qquad  \mbox{if and only if} \qquad \xi \in \big[ -\tfrac{1+\alpha}{2}is^{-\rho},-i\infty \big). \label{y negat}
\end{align}
Therefore, for all 
\begin{align*}
\xi \in \C \setminus \big(  \big[ \tfrac{1+\alpha}{2}is^{-\rho},i\infty \big)  \cup   \big[ -\tfrac{1+\alpha}{2}is^{-\rho},-i\infty \big)\big),
\end{align*}
we can use \eqref{stirlingwithremainder} together with \eqref{def of ln F} to write
\begin{align*}
\ln F&(is^\rho\xi+1/2)= \ln \Gamma(x(\xi)) - \ln \Gamma(y(\xi)) \\
= &\; x(\xi)\ln(x(\xi))-x(\xi)-\frac{1}{2}\ln(x(\xi))+\ln \sqrt{2\pi}+\frac{1}{12x(\xi)} +\sum_{n=2}^{N}\frac{B_{2n}}{(2n-1)2n} \frac{1}{x(\xi)^{2n-1}}+\mathcal{D}_N(x(\xi)) \\
& -y(\xi)\ln(y(\xi))+y(\xi)+\frac{1}{2}\ln(y(\xi))-\ln \sqrt{2\pi}-\frac{1}{12y(\xi)} -\sum_{n=2}^{N}\frac{B_{2n}}{(2n-1)2n} \frac{1}{y(\xi)^{2n-1}}-\mathcal{D}_N(y(\xi)).
\end{align*}
Hence, by \eqref{def of ln mathcal G} and (\ref{hdef}) we have, for any fixed $N\ge 1$,
\begin{align}\label{exactG}
\ln \mathcal{G}(\xi)= \hat{f}(\xi)+\tilde{f}(\xi) +\mathcal{D}_N(x(\xi))-\mathcal{D}_N(y(\xi)),
\end{align}
where the functions $\hat{f}(\xi)$ and $\tilde{f}(\xi)$ are defined by
\begin{subequations}\label{fhatftildedef}
\begin{align}\nonumber
\hat{f}(\xi) = &\; x(\xi)\ln(x(\xi))-x(\xi)-\frac{1}{2}\ln(x(\xi))+\frac{1}{12x(\xi)}-is^{\rho}\xi\big(a_1\ln(s)+c_1\ln(i\xi)+a_2\big)
\\
& +\sum_{n=2}^{N}\frac{B_{2n}}{(2n-1)2n} \frac{1}{x(\xi)^{2n-1}}, 
\\ \nonumber
\tilde{f}(\xi)= & -y(\xi)\ln(y(\xi))+y(\xi)+\frac{1}{2}\ln(y(\xi))-\frac{1}{12y(\xi)}-is^{\rho}\xi\big(\tilde{a}_1\ln(s)+c_2\ln(-i\xi)+\tilde{a}_2\big)
\\
& -\sum_{n=2}^{N}\frac{B_{2n}}{(2n-1)2n} \frac{1}{y(\xi)^{2n-1}} , 
\end{align}
\end{subequations}
with the real constants $a_1, a_2, \tilde{a}_1, \tilde{a}_2$ defined by
\begin{align*}
a_1=\frac{\theta}{\theta+1}, \quad a_2=-1,  \quad \tilde{a}_1=1-a_1=\frac{1}{\theta+1}, \quad \tilde{a}_2=c_3-a_2=-\frac{1+\ln \theta}{\theta}.
\end{align*}

The functions $\hat{f}(\xi)$, $\mathcal{D}_N(x(\xi))$ and $\tilde{f}(\xi)$, $\mathcal{D}_N(y(\xi))$ are analytic for
\begin{align}\label{DN f hat f tilde domain of analycity}
\xi \in \C \setminus  \big[ \tfrac{1+\alpha}{2}is^{-\rho},i\infty \big) \quad \text{ and } \quad \xi \in \C \setminus   \big[- \tfrac{1+\alpha}{2}is^{-\rho},-i\infty \big),
\end{align}
respectively. The asymptotics of $\hat{f}$ and $\tilde{f}$ as $s^{\rho}\xi \to \infty$ are easily obtained from \eqref{fhatftildedef}:
\begin{subequations}\label{fhatftildeexpansions}
 \begin{align}
 \hat{f}(\xi) &= a_3\ln(s)+c_5\ln(i\xi)+a_4+\sum_{n=1}^{N} \frac{\hat{f}_n}{(s^{\rho}\xi)^{n}}+\bigO\bigg( \frac{1}{(s^\rho \xi)^{N+1}} \bigg),\label{expansion fhat}
 \\
 \tilde{f}(\xi)&=\tilde{a}_3\ln(s)+c_6\ln(-i\xi)+\tilde{a}_4+\sum_{n=1}^{N} \frac{\tilde{f}_n}{(s^{\rho}\xi)^{n}}+\bigO\bigg( \frac{1}{(s^\rho \xi)^{N+1}}\bigg),\label{expansion ftilde}
 \end{align}
\end{subequations} 
 as $s^\rho \xi \to \infty$ uniformly for $\theta$ in compact subsets of $(0,1]$, where the constants $\{a_j, \tilde{a}_j\}_3^4 \subset \R$ and $\hat{f}_1, \tilde{f}_1 \in i\R$ are defined by
 \begin{align}
 a_3&= \frac{\alpha \theta}{2(\theta+1)}, &a_4&=0, & \hat{f}_1= -i\frac{3\alpha^2-1}{24}, \label{def of a3a4a5} 
 \\
 \tilde{a}_3&=c_4-a_3, & \tilde{a}_4&=c_7, &   \tilde{f}_1= -ic_8-\hat{f}_1, \label{def of a3a4a5 tilde} 
 \end{align}
 and $\{\hat{f}_n,\tilde{f}_n\}_{n=2}^N \subset \C$ are constants whose exact expressions are unimportant for us. However, we note that $\hat{f}_n(\theta,\alpha)$ and $\tilde{f}_n(\theta,\alpha)$ are continuous functions of $\alpha$ and $\theta$. From \eqref{remainderO} and \eqref{def of x and y}, we infer that
\begin{align}
& \mathcal{D}_{N}(x(\xi)) = \bigO\big( (s^\rho \xi)^{-2N-1}\big) & & \mbox{as } s^{\rho}\xi \to \infty, \quad |\arg (\xi) - \tfrac{\pi}{2}| > \epsilon, \label{lol23} \\
& \mathcal{D}_{N}(y(\xi)) = \bigO\big( (s^\rho \xi)^{-2N-1}\big) & & \mbox{as } s^{\rho}\xi \to \infty, \quad |\arg (\xi) + \tfrac{\pi}{2}| > \epsilon, \label{lol24}
\end{align}
for any $\epsilon > 0$ uniformly for $\theta$ in compact subsets of $(0,1]$.
Substituting \eqref{fhatftildeexpansions}--\eqref{lol24} into \eqref{exactG}, we obtain (\ref{full expansion G}) where the coefficients $\mathcal{G}_{n}$ are given by $\mathcal{G}_n =\hat{f}_n+\tilde{f}_n$; in particular, $\mathcal{G}_1= -ic_8$.
This completes the proof of Proposition \ref{Gprop}.

\subsection{Proof of Proposition \ref{pprop}}\label{ppropproofsubsec}
Recall that $p(\zeta)$ is defined by
\begin{align}\label{def of p recall}
p(\zeta)= - \frac{r(\zeta)}{2\pi i} \int_{\Sigma_{5}} \frac{\ln \mathcal{G}(\xi)}{r_{+}(\xi)}\frac{d\xi}{\xi-\zeta}.
\end{align}
Since $\Sigma_{5}$ passes through the origin, the large $s$ asymptotics for $p$ cannot be straightforwardly obtained from the asymptotics (\ref{full expansion G}) of $\ln \mathcal{G}(\zeta)$. We instead use formula \eqref{exactG} to be able to deform the contour $\Sigma_5$. Substituting \eqref{exactG} into the definition \eqref{def of p recall} of $p(\zeta)$ yields
\begin{align}\nonumber
p(\zeta)= & - \frac{r(\zeta)}{2\pi i} \int_{\Sigma_{5}} \frac{\hat{f}(\xi)}{r_{+}(\xi)}\frac{d\xi}{\xi-\zeta}- \frac{r(\zeta)}{2\pi i} \int_{\Sigma_{5}} \frac{\tilde{f}(\xi)}{r_{+}(\xi)}\frac{d\xi}{\xi-\zeta}
	\\ \label{lol27}
&-\frac{r(\zeta)}{2\pi i} \int_{\Sigma_{5}} \frac{\mathcal{D}_N(x(\xi))}{r_{+}(\xi)(\xi-\zeta)} d\xi +\frac{r(\zeta)}{2\pi i} \int_{\Sigma_{5}} \frac{\mathcal{D}_N(y(\xi))}{r_{+}(\xi)(\xi-\zeta)} d\xi.
\end{align}

The remainder of the proof is divided into three lemmas. The first lemma shows that the two integrals in (\ref{lol27}) involving $\mathcal{D}_N$ are small whenever $s^\rho$ and $s^{\rho}\zeta$ are large. 

\begin{lemma}\label{lemma: remainderp} 
For any integer $N\ge 1$, it holds that
\begin{subequations}
\begin{align}\label{lol26}
& \frac{r(\zeta)}{2\pi i} \int_{\Sigma_{5}}  \frac{\mathcal{D}_N(x(\xi))}{r_{+}(\xi)(\xi-\zeta)} d\xi 
=  \bigO\big((s^{\rho}\zeta)^{-2N-1}\big)
 + \bigO\big(s^{-\rho(2N+1)}\big) \qquad \mbox{as } s \to +\infty,
	\\\label{lol26b}
& \frac{r(\zeta)}{2\pi i} \int_{\Sigma_{5}} \frac{\mathcal{D}_N(y(\xi))}{r_{+}(\xi)(\xi-\zeta)} d\xi =  \bigO\big((s^{\rho}\zeta)^{-2N-1}\big)
 + \bigO\big(s^{-\rho(2N+1)}\big) \qquad \mbox{as } s \to +\infty,
\end{align}
\end{subequations}
uniformly for $\zeta \in \C \setminus \Sigma_5$ such that $s^{\rho}\zeta\to \infty$ and $\theta$ in compact subsets of $(0,1]$.
\end{lemma}
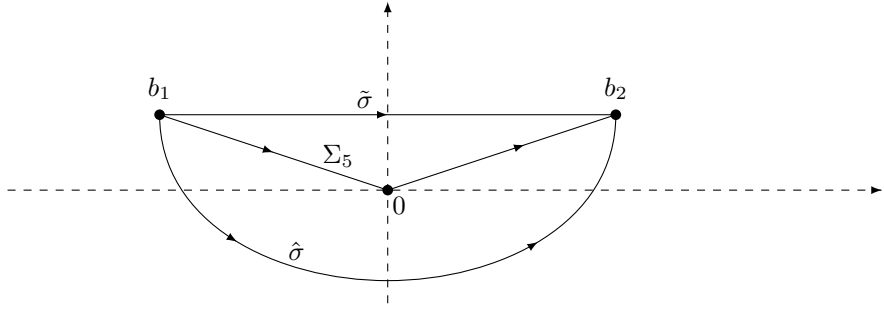
\begin{figure}
	 	\begin{center}
	 		\begin{tikzpicture}
	 		\node at (0,0) {};
	 		\fill (0,0) circle (0.07cm);
	 		\node at (0.15,-0.2) {$0$};
	 		
	 		\draw[dashed,->-=1,black] (0,-1.5) to [out=90, in=-90] (0,2.5);
	 		\draw[dashed,->-=1,black] (-5,0) to [out=0, in=-180] (6.5,0);
	 		
	 		\draw[->-=0.6,black] (0,0)--(3,1);
	 		\fill (3,1) circle (0.07cm);
	 		\node at (3,1.35) {$b_{2}$};
	 		\draw[-<-=0.5,black] (0,0)--(-3,1);
	 		\fill (-3,1) circle (0.07cm);
	 		\node at (-3,1.35) {$b_{1}$};
	 		
	 		\draw[->-=0.5,black] (-3,1)--(3,1);
	 		\draw[->-=0.5,black] (-3,1) to [out=-90, in=180] (0,-1.2);
	 		\draw[->-=0.5,black] (0,-1.2) to [out=0, in=-90] (3,1);
	 		
			 \node at (-0.65,0.45) {$\Sigma_{5}$};
	 		\node at (-1.2,-.8) {$\hat{\sigma}$};
	 		\node at (-0.3,1.2) {$\tilde{\sigma}$};
	 		\end{tikzpicture}
	 	\end{center}
	 	\caption{\label{fig: Sigma5 and sigma sigma tilde}The contours $\hat{\sigma}$ and $\tilde{\sigma}$ for $\zeta \notin \mathbb{D}_{\delta/2}(b_1)\cup\mathbb{D}_{\delta/2}(b_2)$.}
\end{figure}
\begin{proof}
Given $\zeta \in \C \setminus \Sigma_5$, the integrand in (\ref{lol26}) is an analytic function of 
\begin{align*}
\xi \in \C \setminus \Big(\Sigma_5 \cup \big[ \tfrac{1+\alpha}{2}is^{-\rho},i\infty \big)\cup \{\zeta \}\Big),
\end{align*}
see \eqref{DN f hat f tilde domain of analycity}. Using that $r_{+}(\xi) + r_{-}(\xi) = 0$ for $\xi  \in \Sigma_{5}$, we can deform the contour $\Sigma_{5}$ into another contour $\hat{\sigma}$ which crosses the imaginary axis below the origin such that
\begin{align}
& |\xi| > \epsilon \mbox{ and } |\arg (\xi) - \tfrac{\pi}{2}| > \epsilon \quad \mbox{ for all } \xi \in \hat{\sigma}, \quad \mbox{ for a certain } \epsilon > 0, \label{lol25}
\end{align}
and such that $\dist(\zeta, \hat{\sigma}) \geq \delta/2$. If $\zeta \notin \mathbb{D}_{\delta/2}(b_1)\cup\mathbb{D}_{\delta/2}(b_2)$, a representative choice of $\hat{\sigma}$ is shown in Figure \ref{fig: Sigma5 and sigma sigma tilde}, and we obtain
\begin{align}\label{first remainder}
 \frac{r(\zeta)}{2\pi i} \int_{\Sigma_{5}}  \frac{\mathcal{D}_N(x(\xi))}{r_{+}(\xi)(\xi-\zeta)} d\xi = \left\{ \begin{array}{l l}
 \ds -\frac{r(\zeta)}{2\pi i} \int_{\hat{\sigma}}  \frac{\mathcal{D}_N(x(\xi))}{r(\xi)(\xi-\zeta)} d\xi, & \zeta \in \mbox{ext}(\hat{\sigma} \cup \Sigma_{5}), \\
 \ds -\frac{r(\zeta)}{2\pi i} \int_{\hat{\sigma}}  \frac{\mathcal{D}_N(x(\xi))}{r(\xi)(\xi-\zeta)} d\xi + \mathcal{D}_{N}(x(\zeta)), & \zeta \in \mbox{int}(\hat{\sigma} \cup \Sigma_{5}),
 \end{array} \right.
 \end{align}
where, for a simple closed curve $\gamma \subset \C$, we write $\mbox{int}(\gamma)$ and $\mbox{ext}(\gamma)$ for the open subsets of $\C$ interior and exterior to $\gamma$, respectively. If $\zeta \in \mathbb{D}_{\delta/2}(b_1)\cup\mathbb{D}_{\delta/2}(b_2)$, then we use the jump relation of $r(\zeta)$ to open up the parts of $\hat{\sigma}$ close to the points $b_1$ and $b_2$ to two circles in such a way that $\partial\mathbb{D}_{\delta}(b_1)\cup\partial \mathbb{D}_{\delta}(b_2) \subset \hat{\sigma}$, see Figure \ref{fig: Sigma5 and sigma modified}, and instead of \eqref{first remainder} we obtain
\begin{align*}
 \frac{r(\zeta)}{2\pi i} \int_{\Sigma_{5}}  \frac{\mathcal{D}_N(x(\xi))}{r_{+}(\xi)(\xi-\zeta)} d\xi =-\frac{r(\zeta)}{4\pi i} \int_{\partial \mathbb{D}_{\delta}}  \frac{\mathcal{D}_N(x(\xi))}{r(\xi)(\xi-\zeta)} d\xi -\frac{r(\zeta)}{2\pi i} \int_{\hat{\sigma}\setminus \partial \mathbb{D}_{\delta}}  \frac{\mathcal{D}_N(x(\xi))}{r(\xi)(\xi-\zeta)} d\xi + \frac{\mathcal{D}_{N}(x(\zeta))}{2},
\end{align*}
where $\partial\mathbb{D}_{\delta}:=\partial\mathbb{D}_{\delta}(b_1)\cup\partial \mathbb{D}_{\delta}(b_2)$. 
The cases $\zeta \notin \mathbb{D}_{\delta/2}(b_1)\cup\mathbb{D}_{\delta/2}(b_2)$ and $\zeta \in \mathbb{D}_{\delta/2}(b_1)\cup\mathbb{D}_{\delta/2}(b_2)$ can be treated similarly. Since \eqref{lol25} holds, we can apply \eqref{lol23}, which implies the estimate $\mathcal{D}_N(x(\xi)) = \bigO(s^{-\rho(2N+1)})$ as $s \to +\infty$ uniformly for $\xi \in \hat{\sigma}$. Since $r(\zeta) \sim \zeta$ as $\zeta \to \infty$ and $\text{dist}(\zeta,\hat{\sigma}) \geq \delta/2$, we find 
 \begin{align*}
\begin{cases}
\ds -\frac{r(\zeta)}{2\pi i} \int_{\hat{\sigma}}  \frac{\mathcal{D}_N(x(\xi))}{r(\xi)(\xi-\zeta)} d\xi = \bigO(s^{-\rho(2N+1)}), & \ds \mbox{if } \zeta \notin \mathbb{D}_{\delta/2}(b_1)\cup\mathbb{D}_{\delta/2}(b_2) \\[0.3cm]
\ds -\frac{r(\zeta)}{4\pi i} \int_{\partial \mathbb{D}_{\delta}}  \frac{\mathcal{D}_N(x(\xi))}{r(\xi)(\xi-\zeta)} d\xi -\frac{r(\zeta)}{2\pi i} \int_{\hat{\sigma}\setminus \partial \mathbb{D}_{\delta}}  \frac{\mathcal{D}_N(x(\xi))}{r(\xi)(\xi-\zeta)} d\xi = \bigO(s^{-\rho(2N+1)}) & \ds \mbox{if } \zeta \in \mathbb{D}_{\delta/2}(b_1)\cup\mathbb{D}_{\delta/2}(b_2)
\end{cases}
\end{align*}
as $s\to +\infty$ uniformly for $\zeta \in \mathbb{C}\setminus \Sigma_{5}$ and $\theta$ in compact subsets of $(0,1]$. The term $\mathcal{D}_N(x(\zeta))$ is present in the case $\zeta \in \mathbb{D}_{\delta/2}(b_1)\cup\mathbb{D}_{\delta/2}(b_2)$, and also in \eqref{first remainder} if $\zeta \in \mbox{int}(\hat{\sigma} \cup \Sigma_{5})$. Since $|\arg(\zeta)- \tfrac{\pi}{2}|>\epsilon$ for a certain $\epsilon > 0$, and since $s^{\rho}\zeta \to \infty$ by assumption, we can apply \eqref{lol23} to obtain
\begin{align*}
\mathcal{D}_N(x(\zeta)) = \bigO((s^{\rho}\zeta)^{-2N-1}).
\end{align*}
This proves \eqref{lol26}.
\begin{figure}
 \begin{center}
 \begin{tikzpicture}
 \node at (0,0) {};
 \fill (0,0) circle (0.07cm);
 \node at (0.15,-0.2) {$0$};
 
 \draw[dashed,->-=1,black] (0,-1.5) to [out=90, in=-90] (0,2.5);
 \draw[dashed,->-=1,black] (-5,0) to [out=0, in=-180] (6.5,0);
 
 \draw[->-=0.6,black] (0,0)--(3,1);
 \fill (3,1) circle (0.07cm);
 \node at (4.75,0.9) {$\partial\mathbb{D}_{ \delta}(b_{2})$};
 \draw[-<-=0.5,black] (0,0)--(-3,1);
 \fill (-3,1) circle (0.07cm);
 \node at (-4.75,1.05) {$\partial\mathbb{D}_{ \delta}(b_{1})$};
 
 \draw[->-=0.6,black] (-2.05,0.69) to [out=-75, in=180] (0,-1.2);
 \draw[->-=0.5,black] (0,-1.2) to [out=0, in=-105] (2.05,0.69);
 
  \draw[<->,black] (-3,1.5)--(-3,2);
  \draw[<->,black] (3,1.5)--(3,2);
  
  \node at (-2.7,1.72) {\small $\delta/2$};
  \node at (3.3,1.73) {\small $\delta/2$};
  
  \fill (-3.2,0.7) circle (0.04cm);
  \node at (-3.05,0.7) {\small $\zeta$};
 
 \node at (-0.7,0.45) {$\Sigma_{5}$};
 \node at (-1.35,-.95) {$\hat{\sigma}$};
 
 \draw[dashed,black] (3,1) circle (0.5cm);
 \draw[dashed,black] (-3,1) circle (0.5cm);
  \draw[->-=1,black] (3,1) circle (1cm);
  \draw[->-=0.5,black] (-3,1) circle (1cm);
 \end{tikzpicture}
 \end{center}
 \caption{\label{fig: Sigma5 and sigma modified}The contour $\hat{\sigma}$ for $\zeta \in \mathbb{D}_{\delta/2}(b_1)\cup\mathbb{D}_{\delta/2}(b_2)$.}
 \end{figure}
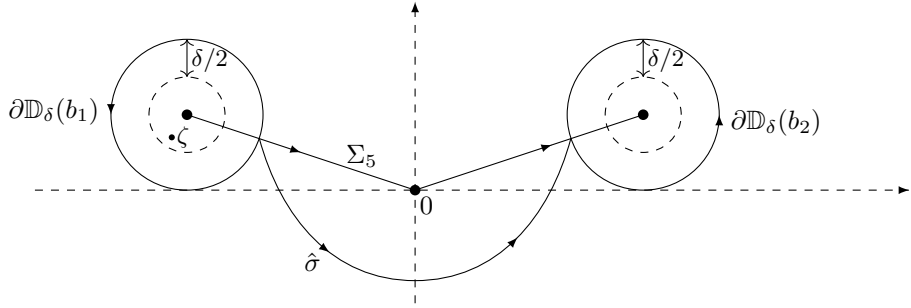
A similar argument based on deforming $\Sigma_{5}$ into a contour $\tilde{\sigma}$ which crosses the imaginary axis above the origin (see Figure \ref{fig: Sigma5 and sigma sigma tilde} in the case when $\zeta \notin \mathbb{D}_{\delta/2}(b_1)\cup\mathbb{D}_{\delta/2}(b_2)$) yields (\ref{lol26b}).
\end{proof}

It remains to compute the asymptotics of the two integrals in (\ref{lol27}) involving $\hat{f}$ and $\tilde{f}$. Since $\Sigma_{5}$ passes through the origin, we cannot immediately use the asymptotic formulas \eqref{fhatftildeexpansions} for $\hat{f}$ and $\tilde{f}$. However, since $\hat{f}$ and $\tilde{f}$ are analytic in the regions \eqref{DN f hat f tilde domain of analycity}, we can deform the contours in the same way as in the proof of Lemma \ref{lemma: remainderp} and then use \eqref{fhatftildeexpansions}.
\begin{definition}\label{def of contour Sigma+- hat}
Let $\epsilon > 0$ be sufficiently small but fixed. We define contours $\widehat{\Sigma}_{5,\pm} = \widehat{\Sigma}_{5,\pm}(\epsilon)$ as follows:
\begin{align*}
& \widehat{\Sigma}_{5,+}(\epsilon) = \big( \Sigma_{5} \cap \{|\xi| > \epsilon \} \big) \cup \{ \xi : |\xi| = \epsilon \mbox{ and } \phi < \arg \xi < \pi - \phi \}, \\
& \widehat{\Sigma}_{5,-}(\epsilon) = \big( \Sigma_{5} \cap \{|\xi| > \epsilon \} \big) \cup \{ \xi : |\xi| = \epsilon \mbox{ and } -\pi -\phi < \arg \xi < \phi \},
\end{align*}
with an orientation chosen from $b_{1}$ to $b_{2}$. Thus $\widehat{\Sigma}_{5,\pm}$ coincide with $\Sigma_{5}$ outside the disks $\{|\xi| \leq \epsilon \}$, and inside this disk, they differ from $\Sigma_{5}$ and instead coincide with the part of the circle $\{|\xi| = \epsilon\}$ lying above (resp. below) $\Sigma_{5}$.
\end{definition}
\begin{lemma}\label{lemma: pintermediateexp}

For each integer $N\ge1$, it holds that
\begin{align}\nonumber
p(\zeta) = & -(c_{4} \ln(s) + c_{7}) \frac{r(\zeta)}{2\pi i} \int_{\Sigma_{5}} \frac{1}{r_{+}(\xi)}\frac{d\xi}{\xi-\zeta} 
-r(\zeta) \bigg\{\frac{c_{5}}{2\pi i} \int_{\Sigma_{5}} \frac{\ln(i\xi)}{r_{+}(\xi)}\frac{d\xi}{\xi-\zeta} + \frac{c_6}{2\pi i} \int_{\Sigma_{5}} \frac{\ln(-i\xi)}{r_{+}(\xi)}\frac{d\xi}{\xi-\zeta} \bigg\}
	\\\nonumber
& +\sum_{n=1}^{N} \frac{1}{s^{n\rho}} 
\bigg\{\hat{f}_n \frac{r(\zeta)}{2\pi i} \int_{\Sigma_{5,-}} \frac{ 1}{\xi^{n}r_{-}(\xi)}\frac{d\xi}{\xi-\zeta} - \tilde{f}_n \frac{r(\zeta)}{2\pi i} \int_{\Sigma_{5,+}} \frac{ 1}{\xi^{n}r_{+}(\xi)}\frac{d\xi}{\xi-\zeta} \bigg\}
	\\\label{lemmapexpansion}
& + \bigO\big(( s^{\rho}\zeta)^{-N-1}\big)  + \bigO\big(s^{-\rho(N+1)}\big), \qquad s \to +\infty,
\end{align}
uniformly for $\zeta \in \C \setminus \Sigma_5$ such that $s^{\rho}\zeta\to \infty$ and $\theta$ in compact subsets of $(0,1]$, and where the contours $\Sigma_{5,\pm}$ depend on $\epsilon$ and $\zeta$ and are given by
\begin{align*}
\Sigma_{5,\pm}:= \widehat{\Sigma}_{5,\pm}(\min \{\epsilon, \tfrac{\zeta}{2}\})
\end{align*}
with $\widehat{\Sigma}_{5,\pm}(\min \{\epsilon, \frac{\zeta}{2}\})$ as in Definition \ref{def of contour Sigma+- hat}.
\end{lemma}
\begin{proof}
Let us assume that $\zeta \notin \mathbb{D}_{\delta/2}(b_1)\cup\mathbb{D}_{\delta/2}(b_2)$. For the integral involving $\hat{f}$ (resp. $\tilde{f}$), we deform $\Sigma_{5}$ into $\hat{\sigma}$ (resp. $\tilde{\sigma}$), and we pick up a residue if $\zeta \in \mbox{int}( \hat{\sigma} \cup \Sigma_{5})$ (resp. if $\zeta \in \mbox{int}( \tilde{\sigma} \cup \Sigma_{5})$). This gives
\begin{subequations}\label{fhatftildeintegrals}
\begin{align}
- \frac{r(\zeta)}{2\pi i} \int_{\Sigma_{5}} \frac{\hat{f}(\xi)}{r_{+}(\xi)}\frac{d\xi}{\xi-\zeta} & = \frac{r(\zeta)}{2\pi i} \int_{\hat{\sigma}} \frac{\hat{f}(\xi)}{r(\xi)}\frac{d\xi}{\xi-\zeta} - \chi_{\mbox{int}(\hat{\sigma} \cup \Sigma_5)}(\zeta) \hat{f}(\zeta), 
	\\
 - \frac{r(\zeta)}{2\pi i} \int_{\Sigma_{5}} \frac{\tilde{f}(\xi)}{r_{+}(\xi)}\frac{d\xi}{\xi-\zeta} & = - \frac{r(\zeta)}{2\pi i} \int_{\tilde{\sigma}} \frac{\tilde{f}(\xi)}{r(\xi)}\frac{d\xi}{\xi-\zeta} - \chi_{\mbox{int}(\tilde{\sigma} \cup \Sigma_5)}(\zeta) \tilde{f}(\zeta).
\end{align}
\end{subequations}
Since $|\xi| > \epsilon$ for $\xi \in \hat{\sigma} \cup \tilde{\sigma}$, the expansions \eqref{fhatftildeexpansions} imply that the integrals on the right-hand side of (\ref{fhatftildeintegrals}) satisfy
\begin{subequations}\label{fhatftildeintegrals2}
\begin{align}\nonumber
\frac{r(\zeta)}{2\pi i} \int_{\hat{\sigma}} \frac{\hat{f}(\xi)}{r(\xi)}\frac{d\xi}{\xi-\zeta} = & \;  (a_{3} \ln(s)+a_{4}) \frac{r(\zeta)}{2\pi i} \int_{\hat{\sigma}} \frac{1}{r(\xi)}\frac{d\xi}{\xi-\zeta}+ c_{5}\frac{r(\zeta)}{2\pi i} \int_{\hat{\sigma}} \frac{\ln(i\xi)}{r(\xi)}\frac{d\xi}{\xi-\zeta} 
	\\
& \; +\sum_{n=1}^{N} \frac{\hat{f}_n}{s^{n\rho}} \frac{r(\zeta)}{2\pi i} \int_{\hat{\sigma}} \frac{ 1}{\xi^{n}r(\xi)}\frac{d\xi}{\xi-\zeta} + \bigO(s^{-\rho(N+1)}),
	\\\nonumber
- \frac{r(\zeta)}{2\pi i} \int_{\tilde{\sigma}} \frac{\tilde{f}(\xi)}{r(\xi)}\frac{d\xi}{\xi-\zeta} = & -(\tilde{a}_{3} \ln(s)+\tilde{a}_{4}) \frac{r(\zeta)}{2\pi i} \int_{\tilde{\sigma}} \frac{1}{r(\xi)}\frac{d\xi}{\xi-\zeta}- c_{6}\frac{r(\zeta)}{2\pi i} \int_{\tilde{\sigma}} \frac{\ln(-i\xi)}{r(\xi)}\frac{d\xi}{\xi-\zeta} 
	\\
& -\sum_{n=1}^{N} \frac{\tilde{f}_n}{s^{n\rho}} \frac{r(\zeta)}{2\pi i} \int_{\tilde{\sigma}} \frac{ 1}{\xi^{n}r(\xi)}\frac{d\xi}{\xi-\zeta} +  \bigO(s^{-\rho(N+1)}),
\end{align}
\end{subequations}
as $s \to + \infty$ uniformly for $\zeta \in \C \setminus \Sigma_5$ and $\theta$ in compact subsets of $(0,1]$. 
Substituting (\ref{fhatftildeintegrals}) into (\ref{lol27}) and utilizing (\ref{fhatftildeintegrals2}) and Lemma \ref{lemma: remainderp} in the resulting expression for $p(\zeta)$, we conclude that
\begin{align}\nonumber
p(\zeta)= &\; \mathfrak{p}_1(\zeta) +\mathfrak{p}_2(\zeta) + \mathfrak{p}_3(\zeta) 
 - \chi_{\mbox{int}(\hat{\sigma} \cup \Sigma_5)}(\zeta) \hat{f}(\zeta)
 - \chi_{\mbox{int}(\tilde{\sigma} \cup \Sigma_5)}(\zeta) \tilde{f}(\zeta)
	\\ \label{lol272}
& + \bigO\big((s^{\rho}\zeta)^{-2N-1}\big)
+ \bigO\big(s^{-\rho(N+1)}\big).
 \end{align}
where the functions $\{\mathfrak{p}_j(\zeta)\}_1^3$ are given by
\begin{align}
\mathfrak{p}_1(\zeta) 
= &\; (a_{3} \ln(s)+a_{4}) \frac{r(\zeta)}{2\pi i} \int_{\hat{\sigma}} \frac{1}{r(\xi)}\frac{d\xi}{\xi-\zeta}
-(\tilde{a}_{3} \ln(s)+\tilde{a}_{4}) \frac{r(\zeta)}{2\pi i} \int_{\tilde{\sigma}} \frac{1}{r(\xi)}\frac{d\xi}{\xi-\zeta}, \label{frak1 1}
	\\
\mathfrak{p}_2(\zeta) = &\; c_{5}\frac{r(\zeta)}{2\pi i} \int_{\hat{\sigma}} \frac{\ln(i\xi)}{r(\xi)}\frac{d\xi}{\xi-\zeta} 
 - c_{6}\frac{r(\zeta)}{2\pi i} \int_{\tilde{\sigma}} \frac{\ln(-i\xi)}{r(\xi)}\frac{d\xi}{\xi-\zeta}, \label{frak2 1}
 	\\
\mathfrak{p}_3(\zeta)  = &\; \sum_{n=1}^{N} \frac{\hat{f}_n}{s^{n\rho}} \frac{r(\zeta)}{2\pi i} \int_{\hat{\sigma}} \frac{ 1}{\xi^{n}r(\xi)}\frac{d\xi}{\xi-\zeta} 
 -\sum_{n=1}^{N} \frac{\tilde{f}_n}{s^{n\rho}} \frac{r(\zeta)}{2\pi i} \int_{\tilde{\sigma}} \frac{ 1}{\xi^{n}r(\xi)}\frac{d\xi}{\xi-\zeta}. \label{frak3 1}
 \end{align}
Now, we deform the contours $\hat{\sigma}$ and $\tilde{\sigma}$ appearing in \eqref{frak1 1}-\eqref{frak2 1} back to $\Sigma_{5}$. The integrands in the right-hand side of \eqref{frak3 1} have a non-integrable singularity at $0$, and therefore for these integrals we instead deform $\hat{\sigma}$ into $\Sigma_{5,-}$ and $\tilde{\sigma}$ into $\Sigma_{5,+}$, and we find that
\begin{subequations}\label{frakpexpressions}
\begin{align}\nonumber
\mathfrak{p}_1(\zeta) = & -(c_{4} \ln(s) + c_{7}) \frac{r(\zeta)}{2\pi i} \int_{\Sigma_{5}} \frac{1}{r_{+}(\xi)}\frac{d\xi}{\xi-\zeta} 
	\\
& + (a_{3} \ln(s) + a_{4}) \chi_{\mbox{int}(\hat{\sigma} \cup \Sigma_5)}(\zeta)  
+ (\tilde{a}_{3} \ln(s) + \tilde{a}_{4}) \chi_{\mbox{int}(\tilde{\sigma} \cup \Sigma_5)}(\zeta), 
  	  \\\nonumber
\mathfrak{p}_2(\zeta) = & -c_{5}\frac{r(\zeta)}{2\pi i} \int_{\Sigma_{5}} \frac{\ln(i\xi)}{r_{+}(\xi)}\frac{d\xi}{\xi-\zeta} 
 - c_{6}\frac{r(\zeta)}{2\pi i} \int_{\Sigma_{5}} \frac{\ln(-i\xi)}{r_{+}(\xi)}\frac{d\xi}{\xi-\zeta}
	\\
& + c_{5} \ln(i\zeta) \chi_{\mbox{int}(\hat{\sigma} \cup \Sigma_5)}(\zeta) + c_{6} \ln(-i\zeta) \chi_{\mbox{int}(\tilde{\sigma} \cup \Sigma_5)}(\zeta),
	\\\nonumber
\mathfrak{p}_3(\zeta) = & \sum_{n=1}^{N} \frac{\hat{f}_n}{s^{n\rho}} \frac{r(\zeta)}{2\pi i} \int_{\Sigma_{5,-}} \frac{ 1}{\xi^{n}r_{-}(\xi)}\frac{d\xi}{\xi-\zeta} 	
-\sum_{n=1}^{N} \frac{\tilde{f}_n}{s^{n\rho}} \frac{r(\zeta)}{2\pi i} \int_{\Sigma_{5,+}} \frac{ 1}{\xi^{n}r_{+}(\xi)}\frac{d\xi}{\xi-\zeta} 
	\\
& + \sum_{n=1}^{N} \frac{\hat{f}_n}{{(s^{\rho}\zeta)^n}}\chi_{\mbox{int}(\hat{\sigma} \cup \Sigma_5)}(\zeta) + \sum_{n=1}^{N} \frac{\tilde{f}_n}{{(s^{\rho}\zeta)^n}} \chi_{\mbox{int}(\tilde{\sigma} \cup \Sigma_5)}(\zeta).
\end{align}
\end{subequations}
Substituting (\ref{frakpexpressions}) into \eqref{lol27}, the terms proportional to $\chi_{\mbox{int}(\hat{\sigma} \cup \Sigma_5)}$ and $\chi_{\mbox{int}(\tilde{\sigma} \cup \Sigma_5)}$ in the resulting expression for $p(\zeta)$ are given by
\begin{subequations}\label{chiterms}
\begin{align}
\chi_{\mbox{int}(\hat{\sigma} \cup \Sigma_5)}(\zeta) \bigg\{ - \hat{f}(\zeta) + a_{3} \ln(s) + c_{5} \ln(i\zeta) + a_{4} + \sum_{n=1}^{N} \frac{\hat{f}_{n}}{(s^{\rho}\zeta)^n} \bigg\} 
\end{align}
and
\begin{align}
\chi_{\mbox{int}(\tilde{\sigma} \cup \Sigma_5)}(\zeta) \bigg\{ - \tilde{f}(\zeta) + \tilde{a}_{3} \ln(s) + c_{6} \ln(-i\zeta) + \tilde{a}_{4} + \sum_{n=1}^{N} \frac{\tilde{f}_{n}}{(s^{\rho}\zeta)^n} \bigg\},
\end{align}
\end{subequations}
respectively. Recalling \eqref{fhatftildeexpansions}, we see that the expressions in (\ref{chiterms}) are $\bigO\big((s^{\rho} \zeta)^{-N-1}\big)$ as $s \to +\infty$ uniformly for $\zeta \in \C \setminus \Sigma_5$ such that $s^{\rho}\zeta \to \infty$ and $\theta$ in compact subsets of $(0,1]$. The expansion (\ref{lemmapexpansion}) then follows from (\ref{lol272}).

The case $\zeta \in \mathbb{D}_{\delta/2}(b_1)\cup\mathbb{D}_{\delta/2}(b_2)$ only requires minor adaptations of the above arguments, which are similar to those done in the proof of Lemma \ref{lemma: remainderp}, and we omit them here. 
\end{proof} 

It remains to compute the coefficients in the expansion (\ref{lemmapexpansion}) of Lemma \ref{lemma: pintermediateexp} more explicitly.
\begin{lemma}\label{lemma: identitiespcoeff} 
Let $\mathcal{R}$ be defined by \eqref{def of mathcalR} and let $\mathcal{G}_n$ be the $n$th coefficient in the expansion of $\mathcal{G}$ given in Proposition \ref{Gprop}. Then the following identities hold:
\begin{align}\label{identitiespcoeff1}
& \frac{r(\zeta)}{2\pi i} \int_{\Sigma_5} \frac{d\xi}{r_+(\xi)(\xi-\zeta)} =  \frac{1}{2}, 
	\\ \label{identitiespcoeff2}
&r(\zeta) \bigg\{\frac{c_{5}}{2\pi i} \int_{\Sigma_{5}} \frac{\ln(i\xi)}{r_{+}(\xi)}\frac{d\xi}{\xi-\zeta} + \frac{c_6}{2\pi i} \int_{\Sigma_{5}} \frac{\ln(-i\xi)}{r_{+}(\xi)}\frac{d\xi}{\xi-\zeta} \bigg\} 
 =  \frac{c_5\ln(i\zeta)}{2} + \frac{c_6\ln(i\zeta)}{2} - \mathcal{R}(\zeta),
	\\\label{identitiespcoeff3}
& \hat{f}_n \frac{r(\zeta)}{2\pi i} \int_{\Sigma_{5,-}} \frac{ 1}{\xi^{n}r_{-}(\xi)}\frac{d\xi}{\xi-\zeta}-\tilde{f}_n \frac{r(\zeta)}{2\pi i} \int_{\Sigma_{5,+}} \frac{ 1}{\xi^{n}r_{+}(\xi)}\frac{d\xi}{\xi-\zeta}
= \frac{\mathcal{A}_n(\zeta)}{\zeta^n}, \qquad n = 1, \dots, N,
\end{align}
where the functions $\mathcal{A}_n(\zeta)$ are defined by
\begin{align}\label{calAndef}
\mathcal{A}_n(\zeta) = -\frac{\mathcal{G}_n}{2}+\frac{ \mathcal{G}_n-2\hat{f}_n}{2 } \sum_{k=0}^{n-1} \frac{r(\zeta)\zeta^{k}}{k!} \frac{d^{k}}{d\xi^{k}}\big(r(\xi)^{-1}\big)\bigg|_{\xi=0-}, \qquad n = 1, \dots, N.
\end{align}
In particular, $\mathcal{A}_1(\zeta)$ is given explicitly by (\ref{def of mathcalA1}), the functions $\mathcal{A}_n(\zeta)$ are holomorphic on $\C \setminus \Sigma_5$, satisfy $\mathcal{A}_n = \bigO(\zeta^n)$ as $\zeta \to \infty$ uniformly for $\theta$ in compact subsets of $(0,1]$, and depend continuously on $\alpha$ and $\theta$.
\end{lemma}
\begin{proof}
Using that $r_{+}(\zeta) + r_{-}(\zeta) = 0$ for $\zeta \in \Sigma_{5}$, we obtain
\begin{align*}
\int_{\Sigma_5} \frac{d\xi}{r_+(\xi)(\xi-\zeta)} & = \frac{1}{2}\int_{\Sigma_5} \frac{d\xi}{r_+(\xi)(\xi-\zeta)} -\frac{1}{2}\int_{\Sigma_5} \frac{d\xi}{r_-(\xi)(\xi-\zeta)} = \frac{1}{2}\int_{\mathcal{L}}\frac{d\xi}{r(\xi)(\xi-\zeta)},
\end{align*}
where $\mathcal{L}$ is a clockwise loop which encircles $\Sigma_5$ but which does not encircle $\zeta$. Deforming $\mathcal{L}$ to infinity, picking up a residue at $\xi = \zeta$, and using that $\frac{1}{r(\xi)(\xi-\zeta)} = \bigO(\xi^{-2})$ as $\xi \to \infty$, we get
\[
\int_{\Sigma_5} \frac{d\xi}{r_+(\xi)(\xi-\zeta)} = \frac{\pi i}{r(\zeta)},
\]
which proves \eqref{identitiespcoeff1}. 

In order to prove (\ref{identitiespcoeff2}), we first establish the identities
\begin{subequations}\label{intlnixi}
\begin{align}\label{intlnixia}
& \int_{\Sigma_5} \frac{\ln(i\xi) d\xi}{r_+(\xi)(\xi-\zeta)} = \pi i \int_{0}^{i\infty} \frac{d\xi}{r(\xi)(\xi-\zeta)} + \pi i \frac{\ln(i\zeta)}{r(\zeta)},
	\\\label{intlnixib}
& \int_{\Sigma_5} \frac{\ln(-i\xi) d\xi}{r_+(\xi)(\xi-\zeta)} = \pi i\int_{0}^{-i\infty}\frac{d\xi}{r(\xi)(\xi-\zeta)}  +\pi i \frac{\ln(-i\zeta)}{r(\zeta)}.
\end{align}
\end{subequations}
The function $\ln(i\xi)$ is not analytic on $(0,i\infty)$. Therefore, to prove (\ref{intlnixia}), we first open up the contour $\Sigma_5$ and deform it into a loop $\mathcal{L}_1$ which encircles $\zeta$ but which avoids the positive imaginary axis as shown in Figure \ref{fig: L1}. This gives
\[
\int_{\Sigma_5} \frac{\ln(i\xi) d\xi}{r_+(\xi)(\xi-\zeta)} = \frac{1}{2}\int_{\Sigma_5} \frac{\ln(i\xi) d\xi}{r_+(\xi)(\xi-\zeta)} - \frac{1}{2}\int_{\Sigma_5} \frac{\ln(i\xi) d\xi}{r_-(\xi)(\xi-\zeta)}
= \frac{1}{2}\int_{\mathcal{L}_1} \frac{\ln(i\xi) d\xi}{r(\xi)(\xi-\zeta)} + \pi i \frac{\ln(i\zeta)}{r(\zeta)}.
\]
Deforming the circular part of $\mathcal{L}_1$ to infinity and using that $\ln(i\xi)$ jumps by $2\pi i$ across the positive imaginary axis, the identity (\ref{intlnixia}) follows. The identity (\ref{intlnixib}) follows in a similar way by deforming the contour to a loop which encircles $\zeta$ but which does not encircle the negative imaginary axis. 

\begin{figure}
	\begin{center}
		\begin{tikzpicture}
		\node at (0,0) {};
		\fill (0,0) circle (0.07cm);
		\node at (0.15,-0.2) {$0$};
		
		\draw[dashed,->-=1,black] (0,-1.5) to [out=90, in=-90] (0,2.5);
		\draw[dashed,->-=1,black] (-5,0) to [out=0, in=-180] (6.5,0);
		
		\draw[->-=0.6,black] (0,0)--(3,1);
		\fill (3,1) circle (0.07cm);
		\node at (3,1.3) {$b_{2}$};
		\draw[-<-=0.5,black] (0,0)--(-3,1);
		\fill (-3,1) circle (0.07cm);
		\node at (-3,1.3) {$b_{1}$};
		
		\draw[-<-=0.5,black] (-3.5,1) to [out=-90, in=180] (0,-1.2);
		\draw[-<-=0.5,black] (0,-1.2) to [out=0, in=-90] (3.5,1);
		\draw[->-=0.5,black] (-3.5,1) to [out=90, in=180] (-0.15,2);
		\draw[->-=0.5,black] (0.15,2) to [out=0, in=90] (3.5,1);
		\draw[-<-=0.5,black] (0,0)--(-0.15,2);
		\draw[->-=0.6,black] (0,0)--(0.15,2);
		
		\node at (-2.2,-1.1) {$\mathcal{L}_{1}$};
		
		\node at (1.25,-0.4) {$\zeta$};
		\fill (1,-0.4) circle (0.07cm);
		\end{tikzpicture}
	\end{center}
	\caption{\label{fig: L1}The contour $\mathcal{L}_{1}$.}
\end{figure}
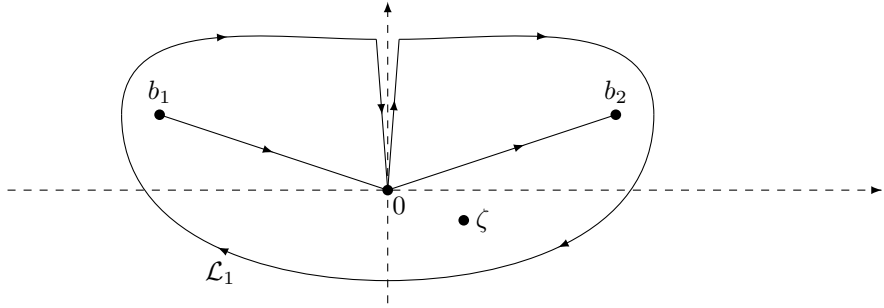

Using that
\begin{align*}
\frac{\partial}{\partial\xi} \Bigg( -\frac{1}{r(\zeta)}\ln\bigg(\frac{b_2(\zeta+\xi)+b_1(-2b_2+\zeta+\xi)-2(\zeta \xi +r(\zeta) r(\xi))}{\xi-\zeta} \bigg) \Bigg)= \frac{1}{r(\xi)(\xi-\zeta)}
\end{align*}
and $r_+(0)=i|b_2|=-r_-(0)$, we can write
\begin{align*}
-\ln\bigg(\frac{b_2(\zeta+\xi)+b_1(-2b_2+\zeta+\xi)-2(\zeta \xi +r(\zeta) r(\xi))}{\xi-\zeta} \bigg)\Bigg|_{\xi = 0^+}^{i\infty} &= \ln\bigg(\frac{|b_2|^2 + i\zeta \im{b_2} - i|b_2| r(\zeta)}{(r(\zeta) + \zeta - i\im{b_2})\zeta}\bigg),
\\
-\ln\bigg(\frac{b_2(\zeta+\xi)+b_1(-2b_2+\zeta+\xi)-2(\zeta \xi +r(\zeta) r(\xi))}{\xi-\zeta} \bigg)\Bigg|_{\xi = 0^-}^{-i\infty} &= \ln\bigg(\frac{|b_2|^2 + i\zeta \im{b_2} +i|b_2| r(\zeta)}{(r(\zeta) + \zeta - i\im{b_2})\zeta}\bigg),
\end{align*}
which shows that
\begin{align}\label{calRexpression}
\mathcal{R}(\zeta) = -r(\zeta)\bigg( c_5 \int_{0}^{i \infty} \frac{d\xi}{r(\xi)(\xi-\zeta)} + c_6 \int_{0}^{-i\infty} \frac{d\xi}{r(\xi)(\xi-\zeta)} \bigg).
\end{align}
The identity (\ref{identitiespcoeff2}) follows from (\ref{intlnixi}) and (\ref{calRexpression}).

To prove (\ref{identitiespcoeff3}), we write $\Sigma_{5,-}=(\Sigma_{5,-} \cup -\Sigma_{5,+}) \cup \Sigma_{5,+}$ and deform $\Sigma_{5,-} \cup -\Sigma_{5,+}$ into a clockwise loop $\mathcal{L}$ which encircles $\Sigma_5$ but which does not encircle $\zeta$. Taking into account the fact that $\mathcal{L}$ and $\Sigma_{5,-} \cup -\Sigma_{5,+}$ have opposite orientations, this shows that the left-hand side of (\ref{identitiespcoeff3}) equals
\begin{align}\label{ftwointegrals}
-\frac{r(\zeta)}{2\pi i} \hat{f}_n \int_{\mathcal{L}}  \frac{d\xi}{\xi^n r(\xi)(\xi-\zeta)}-\frac{r(\zeta)}{2\pi i}  (\tilde{f}_n-\hat{f}_n) \int_{\Sigma_{5,+}}  \frac{d\xi}{\xi^n r_+(\xi)(\xi-\zeta)}.
\end{align}
Deforming $\mathcal{L}$ to infinity (picking up a residue contribution from $\zeta$ but no contribution from infinity), we can write the first term in (\ref{ftwointegrals}) as
\begin{align}\label{identity3firstterm}
-\frac{r(\zeta)}{2\pi i} \hat{f}_n \int_{\mathcal{L}}  \frac{d\xi}{\xi^n r(\xi)(\xi-\zeta)} = -\frac{\hat{f}_n }{\zeta^{n}}.
\end{align}
On the other hand, using the jump relation of $r$ on $\Sigma_5$ to open up the contour $\Sigma_{5,+}$ and picking up a residue contribution from $\xi = 0$, we can write the second term in (\ref{ftwointegrals}) as
\begin{align}\nonumber
-\frac{r(\zeta)}{2\pi i}  (\tilde{f}_n-\hat{f}_n) \int_{\Sigma_{5,+}}  \frac{d\xi}{\xi^n r_+(\xi)(\xi-\zeta)} =& -\frac{r(\zeta)}{4\pi i}  (\tilde{f}_n-\hat{f}_n) \int_{\mathcal{L}}  \frac{d\xi}{\xi^n r(\xi)(\xi-\zeta)}
	\\\label{identity3secondterm}
&- (\tilde{f}_n-\hat{f}_n)\frac{r(\zeta)}{2(n-1)!} \frac{d^{n-1}}{d\xi^{n-1}} \bigg(\frac{1}{r(\xi)(\xi-\zeta)}\bigg)\bigg|_{\xi = 0^-}.
\end{align}
Deforming $\mathcal{L}$ to infinity (picking up a residue contribution from $\zeta$ but no contribution from infinity), the first term on the right-hand side of (\ref{identity3secondterm}) can be written as
\begin{align}\label{identity3secondtermb}
-\frac{r(\zeta)}{4\pi i}  (\tilde{f}_n-\hat{f}_n) \int_{\mathcal{L}}  \frac{d\xi}{\xi^n r(\xi)(\xi-\zeta)}=-\frac{(\tilde{f}_n-\hat{f}_n)}{2 \zeta^{n}}.
\end{align}
Substituting (\ref{identity3firstterm})--(\ref{identity3secondtermb}) into (\ref{ftwointegrals}) and recalling that $\mathcal{G}_n =\hat{f}_n+\tilde{f}_n$, it follows that the left-hand side of (\ref{identitiespcoeff3}) equals
\begin{align*}
& -\frac{\hat{f}_n }{\zeta^{n}}-\frac{(\tilde{f}_n-\hat{f}_n)}{2 \zeta^{n}}-(\tilde{f}_n-\hat{f}_n)\frac{r(\zeta)}{2(n-1)!} \frac{d^{n-1}}{d\xi^{n-1}} \bigg(\frac{1}{r(\xi)(\xi-\zeta)}\bigg)\bigg|_{\xi = 0^-}
\\
&= -\frac{\mathcal{G}_n}{2\zeta^{n}}+\frac{ \tilde{f}_n-\hat{f}_n}{2\zeta^n } \sum_{k=0}^{n-1} \frac{r(\zeta)\zeta^{k}}{k!}\frac{d^{k}}{d\xi^{k}}\big(r^{-1}(\xi)\big)\bigg|_{\xi=0^-}
= \frac{\mathcal{A}_n(\zeta)}{\zeta^n},
\end{align*}
which proves (\ref{identitiespcoeff3}).
Recalling that $\mathcal{G}_1= -ic_8$, $\hat{f}_1=-i\frac{3\alpha^2-1}{24}$, and $\tilde{f_1}=-ic_8-\hat{f}_1$ and using that $r_-(0)=-i|b_2|$, we find the explicit expression (\ref{def of mathcalA1}) for the coefficient $\mathcal{A}_1(\zeta)$.
Since $r(\zeta) \sim \zeta$ as $\zeta \to \infty$, we see from (\ref{calAndef}) that $\mathcal{A}_n(\zeta)$ is analytic for $\zeta \in \C \setminus \Sigma_5$ and of order $\bigO(\zeta^n)$ as $\zeta \to \infty$. Furthermore, since $\hat{f}_n$ and $\mathcal{G}_n$ depend continuously on $\alpha$ and $\theta$, so does $\mathcal{A}_n(\zeta)$.
\end{proof}

The asymptotic formula (\ref{full expansion of p}) for $p(\zeta)$ follows by substituting the identities of Lemma \ref{lemma: identitiespcoeff} into the expansion (\ref{lemmapexpansion}) of Lemma \ref{lemma: pintermediateexp}. This completes the proof of Proposition \ref{pprop}.

\section{Asymptotics of $R(\zeta)$} \label{Rsec}
In this section, we establish the existence of an expansion to all orders of $R(\zeta)$ as $s \to +\infty$ and derive an explicit expression for the first coefficient $R^{(1)}(\zeta)$ of this expansion. 
By expanding $R^{(1)}(\zeta)$ as $\zeta \to \infty$, we can compute the matrix $R_1^{(1)}$ defined by \eqref{def of R1^1}. Even though only the $(2,2)$ entry of $R_1^{(1)}$ is needed to compute $c$, we compute the full matrix $R^{(1)}(\zeta)$, because it will be needed later for the evaluation of $C$. 
The results are summarized in the following proposition.

\begin{proposition}[Asymptotics of $R$]\label{Rprop}
 Let $N \geq 1$ be an integer. Suppose $\alpha > -1$ and $0 < \theta \leq 1$. Let $\{c_j\}_1^8$ and $b_1, b_2$ be the complex constants expressed in terms of the parameters $\alpha$ and $\theta$ by (\ref{def of c1...c8}) and (\ref{b1b2def}), respectively.

There exist holomorphic functions $R^{(n)}:\C \setminus (\partial\Done \cup \partial \Dtwo) \to \C$, $n=1,\ldots N$, such that the matrix valued function $R(\zeta)$ defined in (\ref{def of R}) admits the expansion
  \begin{align}\label{full expansion R}
R(\zeta) = I+\sum_{n=1}^{N} \frac{R^{(n)}(\zeta)}{s^{n\rho}} + \bigO\bigg(\frac{1}{s^{(N+1)\rho}(1+ |\zeta|)}\bigg), \qquad s\to +\infty,
\end{align}
uniformly for $\zeta \in \mathbb{C}\setminus \Gamma_R$ and $\theta$ in compact subsets of $(0,1]$. 
As $\zeta \to \infty$, $R^{(n)}(\zeta) = O(\zeta^{-1})$ for each $n = 1, \dots, N$.
The expansion (\ref{full expansion R}) can be differentiated with respect to $\zeta$ in the sense that
\begin{align}\label{dRdzetaexpansion}
R'(\zeta) = \sum_{n=1}^{N} \frac{R^{(n)\prime}(\zeta)}{s^{n\rho}} + \bigO\bigg(\frac{1}{s^{(N+1)\rho}(1+ |\zeta|)^2}\bigg), \qquad s\to +\infty,
\end{align}
uniformly for $\zeta \in \mathbb{C}\setminus \Gamma_R$ and $\theta$ in compact subsets of $(0,1]$. 
For any $N \geq 1$,
\begin{align}\label{RplusRminusestimate}
R_{+}^{-1}(\zeta)R_{+}'(\zeta)-R_{-}^{-1}(\zeta)R_{-}'(\zeta) = \bigO\bigg(\frac{1}{s^{N\rho}(1+ |\zeta|)^N}\bigg) \qquad \mbox{ as } s \to +\infty,
\end{align}
uniformly for $\zeta \in \cup_{i=1}^4 \Sigma_i \setminus (\Done \cup \Dtwo)$ and $\theta$ in compact subsets of $(0,1]$.
Moreover, the first coefficient $R^{(1)}(\zeta)$ is given explicitly by
\begin{align}\label{explicit R1(zeta)}
R^{(1)}(\zeta)=\frac{A}{\zeta-b_1}+\frac{B}{(\zeta-b_1)^2} - \frac{\bar A}{\zeta-b_2} + \frac{\bar B}{(\zeta-b_2)^2}, \qquad \zeta \in \mathbb{C}\setminus (\overline{\Done} \cup \overline{\Dtwo}),
\end{align}
where the constant matrices $A$ and $B$ are defined by
\begin{align}\label{ABexplicit}
A = \begin{pmatrix}
A_{1,1} & A_{1,2} \\ A_{2,1} & A_{2,2}
\end{pmatrix},
\qquad 
B = - \frac{5b_1}{48(c_1+c_2)} \begin{pmatrix}
i & 1 \\ 1 & -i
\end{pmatrix},
\end{align}
with
\begin{align*}
A_{1,1}&=\frac{3 \im b_2 +2i\re b_2-12(|b_2|(c_5-c_6)(c_5+c_6)+(c_5^2+c_6^2)\im b_2+2ic_5c_6 \re b_2)}{48 (c_1+c_2) \re b_2}, 
\\
A_{1,2}&=\frac{4i(3|b_2|(c_5-c_6)(1+c_5+c_6)+\im b_2 + 3(c_5+c_5^2+c_6+c_6^2)\im b_2)}{48(c_1+c_2)\re b_2}
\\
&\quad -\frac{(5+12c_6+12c_5(1+2c_6))}{48(c_1+c_2)}
\\
A_{2,1}&= \frac{12i|b_2|(c_5-c_6)(-1+c_5+c_6)+4i(1+3(c_5-1)c_5+3(c_6-1)c_6) \im b_2}{48(c_1+c_2)\re b_2}
\\
& \quad +\frac{-5+12(c_5+c_6-2c_5c_6)}{48(c_1+c_2)}
\\
A_{2,2}&=-A_{1,1}.
\end{align*}
In particular, the matrix $R_1^{(1)}$ in (\ref{def of R1^1}) is given by $R_1^{(1)} = A-\bar{A}$ and has $(2,2)$ element
\begin{align}\label{R1122}
(R_1^{(1)})_{2,2} = -i\frac{1- 12c_5c_6}{12(c_1+c_2)}.
\end{align}
\end{proposition}

The remainder of this section is devoted to the proof of Proposition \ref{Rprop}. We start by obtaining an asymptotic expansion of the jump matrix $J_R$ for the RH problem satisfied by $R$.

\subsection{Asymptotics of $J_R$}
We recall from (\ref{def of R}) that $R$ is given by
\[
R(\zeta) = e^{p_0 \sigma_3} S(\zeta) \times \begin{cases}
P(\zeta)^{-1} e^{-p_0 \sigma_3}, &\text{ if $\zeta \in \Done \cup  \Dtwo$},
\\
P^{\infty}(\zeta)^{-1} e^{-p_0 \sigma_3}, &\text{ elsewhere},
\end{cases}
\]
where $S$, $P$, and $P^{\infty}$ have been defined in Section \ref{Section: ClaeysGirSti}. 
For $\zeta \in \Gamma_R$, $R$ satisfies the jump condition $R_+ = R_- J_R$ where 
\begin{align}\label{jumps for R}
J_R(\zeta)= \begin{cases}
e^{p_0 \sigma_3} P_-^\infty(\zeta)J_S(\zeta)P_+^\infty(\zeta)^{-1} e^{-p_0 \sigma_3}, &\text{ if $\zeta \in \Gamma_R \setminus (\partial \Done \cup \partial \Dtwo)$}, \\
e^{p_0 \sigma_3}P(\zeta)P^\infty(\zeta)^{-1} e^{-p_0 \sigma_3}, &\text{ if $\zeta \in \partial \Done \cup \partial \Dtwo$},
\end{cases}
\end{align}
and $J_S$ denotes the jump matrix for $S$ (see \cite[Eq. (3.21)]{ClaeysGirSti}):
\begin{align}\label{JSdef}
J_S(\zeta) = \begin{cases}
\begin{pmatrix} 1 & -\mathcal{G}(\zeta) e^{s^\rho(2g(\zeta) - i h(\zeta) + \ell)} \\ 0 & 1 \end{pmatrix}, & \zeta \in \Sigma_1 \cup \Sigma_2, 
	\\
\begin{pmatrix} 1 & 0 \\ \mathcal{G}(\zeta)^{-1} e^{-s^\rho(2g(\zeta) - i h(\zeta) + \ell)}  & 1 \end{pmatrix}, & \zeta \in \Sigma_3 \cup \Sigma_4, 
	\\
\begin{pmatrix} e^{-s^\rho(g_+(\zeta) - g_-(\zeta))} & -\mathcal{G}(\zeta) \\ \mathcal{G}(\zeta)^{-1} & 0 \end{pmatrix}, & \zeta \in \Sigma_5. 
\end{cases}
\end{align}

The symmetries 
$$\mathcal{G}(\zeta)=\overline{\mathcal{G}(-\bar{\zeta})}, \quad g(\zeta)=\overline{g(-\bar{\zeta})}, \quad h(\zeta) = -\overline{h(-\bar{\zeta})}, \quad P^\infty(\zeta)=\overline{P^\infty(-\bar{\zeta})}, \quad P(\zeta)=\overline{P(-\bar{\zeta})},$$
together with the fact that $p_0, \ell \in \R$ imply that the jump matrix $J_R$ obeys the symmetry
\begin{align}\label{JRsymm}
J_R(\zeta)=\overline{J_R(-\bar{\zeta})}, \qquad \zeta \in \Gamma_R.
\end{align}
From (\ref{JRsymm}) and the symmetry of the behavior of $R$ near the points of self-intersection of $\Gamma_{R}$ and infinity, as well as the uniqueness of the solution of the RH problem for $R$, we conclude that $R$ obeys the symmetry
\begin{align*}
R(\zeta) = \overline{R(-\overline{\zeta})}, \qquad \zeta \in \mathbb{C}\setminus \Gamma_{R}.
\end{align*}
Note that $|b_2| = (1+\theta)\theta^{\frac{1-\theta}{1+\theta}}$ by (\ref{b1b2def}) and (\ref{def of Re b2}), so that $b_1$ and $b_2$ approach the origin only as $\theta \downarrow 0$. 

The next lemma establishes the existence of an asymptotic expansion to all orders of the jump matrix $J_R(\zeta)$ as $s \to + \infty$.

\begin{lemma}[Asymptotics of $J_R$]\label{lemma: full expansion JR}
Let $N\ge 1$ be an integer and let $\alpha > -1$. There exists an asymptotic expansion
\begin{align} \label{full expansion JR}
	J_R(\zeta) = I + \sum_{n=1}^N \frac{J_R^{(n)}(\zeta)}{s^{n \rho}} + \bigO\bigg( \frac{1}{s^{(N+1)\rho}}\bigg) \qquad \mbox{as } s \to +\infty,
	\end{align}
where the error term is uniform for $\zeta \in \partial \Done\cup \partial \Dtwo$ and for $\theta$ in compact subsets of $(0,1]$, and 
\[
J_R^{(n)}: (\overline{\Done} \cup \overline{\Dtwo} ) \setminus \{b_{1},b_{2}\} \to \C^{2\times2}, \qquad n = 1, \dots, N,
\]
are holomorphic functions which satisfy the symmetry
\begin{equation}\label{symmetry of JR n for all n}
J_R^{(n)}(\zeta)=\overline{J_R^{(n)}(-\bar{\zeta})}, \qquad \zeta \in \overline{\Done} \cup \overline{\Dtwo}, \ n = 1, \dots, N.
\end{equation}
For $\zeta \in \overline{\Done}$, $J_{R}^{(1)}(\zeta)$ is explicitly given by
\begin{align} \label{explicit JR1}
J_R^{(1)}(\zeta) & \; = \frac{1}{72q(\zeta)}Q^\infty(\zeta)e^{\frac{\mathcal{R}(\zeta)}{2}\sigma_3}\begin{pmatrix}
-1 & -6i \\ -6i & 1
\end{pmatrix} e^{-\frac{\mathcal{R}(\zeta)}{2}\sigma_3} \big( Q^\infty(\zeta) \big)^{-1},
\end{align}
i.e.,
\begin{align}\nonumber
J_R^{(1)}(\zeta)_{1,1}&=-J_R^{(1)}(\zeta)_{2,2}, 
	\\ \nonumber
J_R^{(1)}(\zeta)_{2,1}&= \frac{1}{72 q(\zeta)}\bigg( i\frac{-\re b_2+6(i\im b_2-\zeta)\cosh(\mathcal{R}(\zeta))}{r(\zeta)}  +6i\sinh (\mathcal{R}(\zeta)) \bigg),
	\\ \nonumber
J_R^{(1)}(\zeta)_{1,2}&=\frac{1}{72 q(\zeta)}\bigg(i \frac{-\re b_2+6(i\im b_2-\zeta)\cosh(\mathcal{R}(\zeta))}{r(\zeta)}  -6i\sinh (\mathcal{R}(\zeta)) \bigg),
	\\ \label{JR1entries}
J_R^{(1)}(\zeta)_{2,2}&=\frac{1}{72 q(\zeta)}\bigg( \frac{-(i\im b_2-\zeta)+6 \re (b_2)\cosh(\mathcal{R}(\zeta))}{r(\zeta)} \bigg).
\end{align}
\end{lemma}
\begin{proof}
Substituting the expressions \eqref{def of Pinf}, \eqref{def of P} and \eqref{def of E local param} for $P^\infty$, $P$, and $E$ into the expression \eqref{jumps for R} for $J_{R}$ on $\partial \Done$, we find
\begin{align}\nonumber
J_{R}(\zeta) = & \; Q^{\infty}(\zeta) e^{p(\zeta)\sigma_{3}} \mathcal{G}(\zeta)^{\frac{\sigma_{3}}{2}} \begin{pmatrix}
1 & i \\ 1 & -i
\end{pmatrix}^{-1} \Big( s^{\frac{2}{3}\rho}f(\zeta) \Big)^{\frac{\sigma_{3}}{4}} 
	\\\label{JRexpression1}
& \times A_{k}\big( s^{\frac{2}{3}\rho}f(\zeta) \big)e^{-s^{\rho}q(\zeta)\sigma_{3}}\mathcal{G}(\zeta)^{-\frac{\sigma_{3}}{2}}e^{-p(\zeta)\sigma_{3}}Q^{\infty}(\zeta)^{-1}, \qquad \zeta \in \partial \Done.
\end{align}

We can extend the asymptotic formula \eqref{weak asymp for Ak} for $A_{k}(\zeta)$ to all orders as follows. The Airy function admits the following well-known uniform asymptotic expansions to all orders (see \cite[Eqs. 9.7.5 and 9.7.6]{NIST}):
\begin{align}\label{Airyexpansions}
\Ai(\zeta)&\sim \frac{e^{-\frac{2}{3}\zeta^\frac{3}{2}}}{2\sqrt{\pi} \zeta^{1/4}} \sum_{l=0}^\infty \frac{(-1)^lu_l}{(\frac{2}{3}\zeta^{3/2})^l}, & &
\Ai'(\zeta)\sim-\frac{e^{-\frac{2}{3}\zeta^\frac{3}{2}}\zeta^{1/4}}{2\sqrt{\pi} } \sum_{l=0}^\infty \frac{(-1)^lv_l}{(\frac{2}{3}\zeta^{3/2})^l},
\end{align}
as $\zeta \to \infty$, $|\arg \zeta|<\pi - \delta'$ for any $\delta' >0$, where the coefficients $\{u_l, v_l\}_{l=0}^{\infty}$ are given by
\begin{align*}
& u_{l} = \frac{(6l-5)(6l-3)(6l-1)}{(2l-1)216l}u_{l-1}, & &  l \geq 1, \\
& v_{l} = \frac{6l+1}{1-6l}u_{l}, & & l \geq 1,
\end{align*}
and $u_{0} = v_{0} = 1$. Substituting the asymptotic expansions (\ref{Airyexpansions}) into \eqref{def of A1}--\eqref{def of A3}, it follows that, for $k = 1,2,3$,
\begin{align}\nonumber
A_{k}(\zeta) & \sim  \zeta^{-\frac{\sigma_3}{4}}\sum_{l=0}^{\infty}\frac{1}{(\frac{2}{3}\zeta^{3/2})^l}  \begin{pmatrix}
u_l & i(-1)^lu_l \\ v_l & -i(-1)^lv_l
\end{pmatrix}e^{\frac{2}{3}\zeta^{3/2}\sigma_3} 
	\\ \label{ExpansionAk}
& \sim \zeta^{-\frac{\sigma_3}{4}} \begin{pmatrix}
1 & i \\
1 & -i
\end{pmatrix} \left( I + \sum_{l=1}^{\infty} \frac{1}{(\frac{2}{3}\zeta^{3/2})^l} \frac{u_{l}}{1-6l}\begin{pmatrix}
1 & (-1)^{l+1}6 l i \\
6 l i & (-1)^{l}
\end{pmatrix} \right) e^{\frac{2}{3}\zeta^{3/2}\sigma_3}
\end{align}
uniformly in the sector $S_k$ defined in (\ref{Skdef}), where the branches of complex powers are as in (\ref{weak asymp for Ak}).

Next note that by combining the expansions \eqref{full expansion G} and \eqref{full expansion of p}, we find
\begin{align}
e^{p(\zeta)\sigma_3}\mathcal{G}(\zeta)^{\frac{\sigma_3}{2}} &= \exp \left( \left[ \frac{\mathcal{R}(\zeta)}{2} + \sum_{n=1}^{N} \frac{\tilde{\mathcal{A}}_n(\zeta)}{(s^\rho \zeta)^n}+\bigO\bigg( \frac{1}{(s^\rho \zeta)^{N+1}}\bigg) + \bigO \bigg( \frac{1}{s^{(N+1)\rho}}\bigg)\right] \sigma_{3} \right) \label{expansion e^p mathcalG}
\end{align}
as $s \to +\infty$ uniformly for $\theta$ in compact subsets of $(0,1]$ and uniformly for $\zeta \in \C \setminus \Sigma_{5}$ such that $s^{\rho}\zeta \to \infty$, $|\arg (\zeta) - \tfrac{\pi}{2}| > \epsilon$ and $|\arg (\zeta) + \tfrac{\pi}{2}| > \epsilon$ for any fixed $\epsilon  > 0$, where $\mathcal{R}(\zeta)$ is defined by \eqref{def of mathcalR} and $\tilde{\mathcal{A}}_n(\zeta)$ are holomorphic functions of $\zeta \in \mathbb{C}\setminus \Sigma_{5}$ defined by
\begin{align*}
& \tilde{\mathcal{A}}_n(\zeta) = \mathcal{A}_n(\zeta) +\frac{\mathcal{G}_n}{2}, \qquad n \geq 1.
\end{align*}

Utilizing the large $s$ expansions \eqref{ExpansionAk} and \eqref{expansion e^p mathcalG} in the expression (\ref{JRexpression1}) for $J_R(\zeta)$, we obtain
	\begin{align}\nonumber
	J_R(\zeta)= & \; I + Q^\infty(\zeta)\exp \left( \left[ \frac{\mathcal{R}(\zeta)}{2} + \sum_{n=1}^{N} \frac{\tilde{\mathcal{A}}_n(\zeta)}{(s^\rho \zeta)^n} +\bigO\big( s^{-(N+1)\rho }\big) \right] \sigma_{3} \right)
		\\\nonumber
	& \times \left( \sum_{l=1}^{N} \frac{1}{(s^{\rho} q(\zeta))^l} \frac{u_{l}}{1-6l}\begin{pmatrix}
	1 & (-1)^{l+1}6 l i \\
	6 l i & (-1)^{l}
	\end{pmatrix} + \bigO\big((s^\rho q(\zeta))^{-(N+1)}\big)  \right)
		\\ \label{formula full expansion JR}
	& \times \exp \left(- \left[ \frac{\mathcal{R}(\zeta)}{2} + \sum_{n=1}^{N} \frac{\tilde{\mathcal{A}}_n(\zeta)}{(s^\rho \zeta)^n}+\bigO\big( s^{-(N+1)\rho }\big) \right] \sigma_{3} \right) Q^\infty(\zeta)^{-1},
\end{align}
as $s \to +\infty$ uniformly for $\zeta \in \partial \Done$ and $\theta$ in compact subsets of $(0,1]$. 
The error term $\bigO((s^\rho q(\zeta))^{-(N+1)})$ can be replaced by $\bigO(s^{-(N+1)\rho})$, because $|q(\zeta)|$ is uniformly bounded away from zero on $\partial \Done$ by (\ref{asymptotics q}).
It follows that $J_R$ admits an expansion of the form \eqref{full expansion JR} with coefficients $J_R^{(n)}(\zeta)$, $n = 1, \dots, N$, which can be computed explicitly from (\ref{formula full expansion JR}) by straightforward algebra. In particular, this gives the explicit expression (\ref{explicit JR1}) for the first coefficient $J_{R}^{(1)}(\zeta)$; using the definition (\ref{def of Pinf}) of $Q^\infty(\zeta)$ and the fact that $r(\zeta) = (\zeta -b_2)\gamma(\zeta)^2$, the relations in (\ref{JR1entries}) follow. 

We finally show that $J_R^{(n)}(\zeta)$, $n = 1, \dots, N$, are analytic functions of $\zeta \in \overline{\Done} \setminus \{b_1\}$. This will complete the proof of the lemma because the expansion \eqref{full expansion JR} for $\zeta \in \partial \Dtwo$ and the symmetry \eqref{symmetry of JR n for all n} then follow from (\ref{JRsymm}).
Clearly, the coefficients $J_R^{(n)}$ are analytic on $\overline{\Done} \setminus \Sigma_{5}$. 
In fact, it follows from (\ref{formula full expansion JR}) that they have no jump across $\Sigma_{5}$, because for $\zeta \in \Sigma_{5}$ we have
\begin{align*}
& \mathcal{R}_{+}(\zeta) + \mathcal{R}_{-}(\zeta) = 0, \qquad  q_{+}(\zeta) + q_{-}(\zeta) = 0, \qquad Q_{+}^{\infty}(\zeta) = Q_{-}^{\infty}(\zeta)\begin{pmatrix}
0 & -1 \\ 1 & 0
\end{pmatrix}, \\
& \tilde{\mathcal{A}}_{n,+}(\zeta) + \tilde{\mathcal{A}}_{n,-}(\zeta) = 0 \qquad \mbox{for all }n \geq 1, \\
& \frac{q_{-}(\zeta)^{l}}{q_{+}(\zeta)^{l}} \begin{pmatrix}
1 & (-1)^{l+1}6 l i \\
6 l i & (-1)^{l}
\end{pmatrix}^{-1} \begin{pmatrix}
	0 & -1 \\ 1 & 0
\end{pmatrix} \begin{pmatrix}
1 & (-1)^{l+1}6 l i \\
6 l i & (-1)^{l}
\end{pmatrix} = \begin{pmatrix}
0 & -1 \\ 1 & 0
\end{pmatrix} \quad \mbox{for all }l \geq 1.
\end{align*}
This shows that the coefficients $J_R^{(n)}(\zeta)$ are analytic on $\overline{\Done} \setminus \{b_1\}$ (note however that the $J_R^{(n)}(\zeta)$ may have poles at $b_1$ because $q(\zeta)\to 0$ as $\zeta \to b_1$).
\end{proof}

\subsection{Existence of an expansion to all orders}
In the following lemma, we show that the $L^{p}$ norm of $w_R := J_R-I$ on $\Gamma_{R}$ is small for any $1 \leq p \leq \infty$ uniformly for $\theta$ in compact subsets of $(0,1]$, whenever $s$ is large enough.

\begin{lemma}[Estimates of $w_R$]\label{lemma: small norm jumps} Let $N \geq 1$ be an integer and let $K$ be a compact subset of $(0,1]$. For each $1 \le p \le \infty$ and each $M \geq 0$, there exist positive constants $C'$ and $c'$ such that the following estimates hold:
	\begin{subequations} \label{Smalljumpestimates}
		\begin{align}
	\sup_{\theta \in K}	\left\lVert w_R - \sum_{n=1}^{N}\frac{J_R^{(n)}}{s^{n\rho}} \right\rVert_{L^p(\partial \Done \cup \partial \Dtwo)} &\le \frac{C'}{s^{(N+1)\rho}},\label{Smalljumpestimatesa}
		\\
	\sup_{\theta \in K}	\lVert (1+ |\zeta|)^M w_R \rVert_{L^p(\Gamma_R \setminus (\partial \Done \cup \partial \Dtwo))} &\le C'e^{-c's^\rho}, \label{Smalljumpestimatesb}
		\\
	\sup_{\theta \in K}	\lVert (1+ |\zeta|)^M \partial_\zeta w_R \rVert_{L^\infty(\cup_{i=1}^4 \Sigma_i \setminus ( \Done \cup \Dtwo))} &\le C'e^{-c's^\rho}.\label{Smalljumpestimatesc}
		\end{align}
	\end{subequations}
\end{lemma}
\begin{proof}
In this proof, $c'$ and $C'$ denote generic positive constant which may change within a computation. Since $\partial \Done \cup \partial \Dtwo$ is compact, the estimate \eqref{Smalljumpestimatesa} follows from Lemma \ref{lemma: full expansion JR}. 

Assume $\zeta \in (\Sigma_1 \cup \Sigma_2)\setminus (\Done \cup \Dtwo)$. By (\ref{def of Pinf}), (\ref{jumps for R}), and (\ref{JSdef}), we have 
\begin{align}\label{wRexpression}
	w_R(\zeta)=Q^{\infty}(\zeta) \begin{pmatrix}
	0 & -e^{2p(\zeta)}\mathcal{G}(\zeta) e^{s^{\rho}(2g(\zeta) - i h(\zeta) + \ell)} \\
	0 & 0
	\end{pmatrix}Q^{\infty}(\zeta)^{-1}.
\end{align}
We see from the expression (\ref{def of mathcalR}) for $\mathcal{R}(\zeta)$ that $|\re{\mathcal{R}(\zeta)}| = \bigO(\ln |\zeta|)$ as $\zeta \to \infty$
and hence 
$$|e^{\mathcal{R}(\zeta)/2}| = \bigO\big((1 + |\zeta|)^{C'}\big)$$
uniformly for $\zeta \in \Sigma_1\cup \Sigma_2$ and $\theta \in K$. It then follows from Propositions \ref{Gprop} and \ref{pprop} (see (\ref{expansion e^p mathcalG})) that 
\begin{align}\label{e2pcalGestimate}
\big|e^{2p(\zeta)}\mathcal{G}(\zeta)\big| = \bigO\big((1 + |\zeta|)^{C'}\big)
\end{align}
uniformly for $\zeta \in \Sigma_1\cup \Sigma_2$, $\theta \in K$, and $s \geq 1$.
Furthermore, a minor modification of the proof of \cite[Lemma 3.1]{ClaeysGirSti}\footnote{Note that there is a typo in the lemma: it should be $2g(\zeta) - i h(\zeta) + \ell$ instead of $2g(\zeta) - i h(\zeta) - \ell$. Existence of a constant $c'$ is clear from the proof of the lemma.} together with the fact that $h(\zeta) = \bigO(\zeta \ln \zeta)$ as $\zeta \to \infty$ yields
	\begin{align} \label{re2g-ih+l bounded}
	\re (2g(\zeta) - i h(\zeta) + \ell)< -c' |\zeta| <0, \qquad \zeta \in \Sigma_{1} \cup \Sigma_{2},
	\end{align}
	for some $c' >0$ for all $\theta \in K$. Equations (\ref{wRexpression}), (\ref{e2pcalGestimate}), and (\ref{re2g-ih+l bounded}) imply that, for any $M \geq 0$, 
\begin{align}\label{wRSigma12}
\sup_{\theta \in K}	\lVert (1+ |\zeta|)^M w_R \rVert_{L^p((\Sigma_{1}\cup\Sigma_{2}) \setminus (\Done \cup \Dtwo))} \le C'e^{-c's^\rho},
\end{align}
and a similar argument shows that
\begin{align}\label{wRSigma34}
\sup_{\theta \in K}	\lVert (1+ |\zeta|)^M w_R \rVert_{L^p((\Sigma_{3}\cup\Sigma_{4}) \setminus (\Done \cup \Dtwo))} \le C'e^{-c's^\rho}. 
\end{align}

Let now $\zeta \in \Sigma_5 \setminus (\Done \cup \Dtwo)$. Then, from (\ref{jumps for R}) and (\ref{JSdef}), we obtain
	\begin{align}
	J_{R}(\zeta) = e^{p_0 \sigma_3} P_-^\infty(\zeta)\begin{pmatrix}
	e^{-s^\rho(g_+(\zeta)-g_-(\zeta))}& -\mathcal{G}(\zeta) \\
	\mathcal{G}(\zeta)^{-1} & 0
	\end{pmatrix}P_+^\infty(\zeta)^{-1}e^{-p_0 \sigma_3}.
	\end{align}
	Using the jump relation of $P^\infty$, given in \cite[Eq. (3.47)]{ClaeysGirSti}, and (\ref{def of Pinf}) this becomes
	\begin{align*}
	J_{R}(\zeta) &=e^{p_0 \sigma_3} P_-^\infty(\zeta)\begin{pmatrix}
	e^{-s^\rho(g_+(\zeta)-g_-(\zeta))}&- \mathcal{G}(\zeta) \\
	\mathcal{G}(\zeta)^{-1} & 0
	\end{pmatrix}\begin{pmatrix}
	0& \mathcal{G}(\zeta) \\
	-\mathcal{G}(\zeta)^{-1} & 0
	\end{pmatrix}P_-^\infty(\zeta)^{-1}e^{-p_0 \sigma_3}
	\\
	&=e^{p_0 \sigma_3} P_-^\infty(\zeta)\begin{pmatrix}
	1&\mathcal{G}(\zeta)e^{-s^\rho(g_+(\zeta)-g_-(\zeta))}  \\
	0 & 1
	\end{pmatrix}P_-^\infty(\zeta)^{-1}e^{-p_0 \sigma_3}
	\\
	&=Q_{-}^{\infty}(\zeta) \begin{pmatrix}
	1& -e^{2p_-(\zeta)}\mathcal{G}(\zeta) e^{-s^\rho(g_+(\zeta)-g_-(\zeta))} \\
	0 & 1
	\end{pmatrix}Q_{-}^{\infty}(\zeta)^{-1}.
	\end{align*}
Note that $Q^{\infty}_{-}(\zeta)$ and $Q^{\infty}_{-}(\zeta)^{-1}$ are independent of $s$ and bounded from above and from below for $\zeta \in \Sigma_{5}\setminus (\Done \cup \Dtwo)$. Combining Proposition \ref{pprop} with \cite[Lemma 3.1]{ClaeysGirSti}, we have
\begin{align*}
|e^{2p_-(\zeta)}\mathcal{G}(\zeta) e^{-s^\rho(g_+(\zeta)-g_-(\zeta))}| \leq C' e^{-c' s^{\rho}} \qquad \mbox{for }\zeta \in \Sigma_{5} \mbox{ such that } s^{\rho}\zeta \geq M
\end{align*}
for a certain large constant $M$, uniformly for $\theta \in K$. For $\zeta \in \Sigma_{5}$ such that $s^{\rho}\zeta \leq M$, the same estimate still holds; this follows from \cite[Lemma 3.1]{ClaeysGirSti} together with the fact that 
\begin{align*}
e^{2p_-(\zeta)}\mathcal{G}(\zeta) = \bigO(1) \qquad \mbox{for }\zeta \in \Sigma_{5} \mbox{ such that } s^{\rho}\zeta \leq M.
\end{align*}
Therefore, we have
\begin{align*}
\sup_{\theta \in K}	\lVert w_R \rVert_{L^p(\Sigma_{5}\setminus (\Done \cup \Dtwo))} \le C'e^{-c's^\rho},
\end{align*}
which together with (\ref{wRSigma12}) and (\ref{wRSigma34}) finishes the proof of (\ref{Smalljumpestimatesb}).

The estimates (\ref{wRSigma12}) and (\ref{wRSigma34}) can clearly be extended to narrow open sectors containing the rays $\cup_{i=1}^4 \Sigma_i \setminus ( \Done \cup \Dtwo)$. The estimate (\ref{Smalljumpestimatesc}) then follows from the analyticity of the jump matrix $J_R$ and Cauchy's estimate.
\end{proof}

For the reader's convenience, we recall some well-known facts from the theory of singular integral operators. For a function $u \in L^2(\Gamma_R)$ we define the Cauchy integral $\mathcal{C}u$ by
\begin{align*}
\mathcal{C}u(\zeta) = \frac{1}{2\pi i} \int_{\Gamma_R} \frac{u(\xi)}{\xi-\zeta}d\xi,\qquad \zeta \in \C \backslash \Gamma_R,
\end{align*}
and we denote the non-tangential limits of $\mathcal{C}u$ from the left- and right-hand side of $\Gamma_R$ by $\mathcal{C}_{+} u$ and $\mathcal{C}_{-}u$, respectively. The Cauchy operator $\mathcal{C}_{w_R}:L^2(\Gamma_R) \to L^2(\Gamma_R)$ is defined by
\begin{align}
\mathcal{C}_{w_R}u=\mathcal{C}_-(w_R u).
\end{align}
This operator is bounded and linear and, assuming that $I-\mathcal{C}_{w_R}:L^2(\Gamma_R) \to L^2(\Gamma_R)$ is invertible, the solution of the RH problem for $R$ is given by (see e.g. \cite[Section 7]{Deiftetal})
\begin{align} \label{explicit solution of RH problem R}
R= I + \mathcal{C}(\mu_R w_R),
\end{align}
where
\begin{align}\label{muR in terms of CwR}
\mu_R=I+(I-\mathcal{C}_{w_R})^{-1}\mathcal{C}_{w_R}(I).
\end{align}
In particular, if $\mathcal{C}_{w_R}$ has sufficiently small $L^2$-operator norm, $I-\mathcal{C}_{w_R}$ can be inverted in terms of a Neumann series, that is,
\begin{align}\label{Neumann series}
(I-\mathcal{C}_{w_R})^{-1}= \sum_{n=0}^{\infty} \mathcal{C}_{w_R}^n.
\end{align}
Hence it follows from Lemma \ref{lemma: small norm jumps} and the estimate
\begin{align} 
\lVert \mathcal{C}_{w_R} \rVert_{L^2(\Gamma_R)\to L^2(\Gamma_R)} \le \lVert \mathcal{C_-} \rVert_{L^2(\Gamma_R)\to L^2(\Gamma_R)} \lVert w_R\rVert_{L^\infty(\Gamma_R)},
\end{align}
that $I-\mathcal{C}_{w_R}$ is invertible for all sufficiently large $s$. Here $\lVert \cdot \rVert_{L^2(\Gamma_R)\to L^2(\Gamma_R)}$ denotes the operator norm of bounded linear operators $L^2(\Gamma_R)\to L^2(\Gamma_R)$. 

The standard theory for asymptotics of small norm RH problems (see e.g. \cite{Deiftetal}) together with Lemma \ref{lemma: small norm jumps} implies
that $R$ satisfies (\ref{full expansion R}) and that this expansion can be differentiated with respect to $\zeta$. The basic idea here is to combine \eqref{explicit solution of RH problem R}--\eqref{Neumann series} and the expansion \eqref{full expansion JR} of the jump matrix. This immediately gives \eqref{full expansion R} uniformly for $\zeta$ bounded away from the contour $\Gamma_R$. For $\zeta$ close to $\Gamma_R$, one uses analyticity of the jump matrix in a neighborhood of $\Gamma_R$ to deform the contour in such a way that $\zeta$ is bounded away from the deformed contour.

Using the jump relation $R_+ = R_- J_R$, the left-hand side of (\ref{RplusRminusestimate}) can be written for $\zeta \in \cup_{i=1}^4 \Sigma_i \setminus (\Done \cup \Dtwo)$ as
$$J_R^{-1}(\zeta) R_{-}^{-1}(\zeta)R_{-}'(\zeta)J_R(\zeta) + J_R^{-1}(\zeta)J_R'(\zeta) 
- R_{-}^{-1}(\zeta)R_{-}'(\zeta).
$$
The estimate (\ref{RplusRminusestimate}) is then a consequence of the estimates (\ref{Smalljumpestimatesb}) and (\ref{Smalljumpestimatesc}) of $J_R(\zeta)$ and $J_R'(\zeta)$, as well as the expansions (\ref{full expansion R}) and
(\ref{dRdzetaexpansion}) of $R(\zeta)$ and $R'(\zeta)$.

\subsection{Explicit expression for $R^{(1)}(\zeta)$ }
We next derive the explicit expression (\ref{explicit R1(zeta)}) for the coefficient $R^{(1)}(\zeta)$. We have $R=I+\mathcal{C}(\mu_R w_R)$ and, by Lemma \ref{lemma: small norm jumps} and \eqref{muR in terms of CwR},
\begin{align*}
w_R(\zeta) = \frac{J_R^{(1)}(\zeta)}{s^\rho}+ \bigO(s^{-2\rho} (1+|\zeta|)^{-2}), \qquad  \mu_R(\zeta) = I+\bigO(s^{-\rho}),
\end{align*}
as $s\to +\infty$, where the error terms are uniform with respect to $\zeta \in \Gamma_R$ and $\theta$ in compact subsets of $(0,1]$. This implies
\begin{align}\label{R1JR1}
R^{(1)}(\zeta)= \mathcal{C}J_R^{(1)}(\zeta) 
= \frac{1}{2\pi i} \int_{\partial \Done \cup \partial \Dtwo} \frac{J_R^{(1)}(\xi)}{\xi-\zeta} d\xi,
\end{align}
where $\partial \Done$ and $\partial \Dtwo$ are oriented clockwise. From Lemma \ref{lemma: full expansion JR} and (\ref{asymptotics q}), $J_R^{(1)}$ is analytic on $(\overline{\Done} \cup \overline{\Dtwo})\setminus \{ b_1,b_2 \}$ with a double pole at each of the points $b_1$ and $b_2$. Furthermore, by (\ref{symmetry of JR n for all n}) we have $J_R^{(1)}(\zeta) = \overline{J_R^{(1)}(-\bar\zeta)}$ and hence
\begin{align}
\frac{1}{2\pi i} \int_{\partial \Done \cup \partial \Dtwo} \frac{J_R^{(1)}(\xi)}{\xi-\zeta} d\xi=\frac{1}{2\pi i} \int_{\partial \Done} \frac{J_R^{(1)}(\xi)}{\xi-\zeta} d\xi+\overline{\frac{1}{2\pi i}\int_{\partial \Done} \frac{J_R^{(1)}(\xi)}{\xi+\bar\zeta} d\xi}.
\end{align}
By Cauchy's formula, if $\zeta \notin \overline{\Done}$, we have
\begin{align}\label{intDoneJR1}
\frac{1}{2\pi i} \int_{\partial \Done} \frac{J_R^{(1)}(\xi)}{\xi-\zeta} d\xi &= \frac{A}{\zeta-b_1} + \frac{B}{(\zeta-b_1)^2},
\end{align}
where the matrices $A$ and $B$ are defined by
\begin{align}\label{ABdef}
A = \frac{d}{d\xi}\big( (\xi-b_1)^2J_R^{(1)}(\xi)\big)\big|_{\xi=b_1}, \qquad B = \lim_{\xi\to b_1}\big( (\xi-b_1)^2 J_R^{(1)}(\xi) \big),
\end{align}
so that
$$J_R^{(1)}(\xi) = \frac{B}{(\xi-b_1)^2} + \frac{A}{\xi - b_1} + \bigO(1), \qquad \xi \to b_1.$$
It follows from equations (\ref{R1JR1})--(\ref{intDoneJR1}) that $R^{(1)}(\zeta)$ satisfies (\ref{explicit R1(zeta)}) with $A$ and $B$ given by (\ref{ABdef}). 

We next show that the matrices $A$ and $B$ can be written as in (\ref{ABexplicit}).
Expanding \eqref{def of g''} in powers of $\sqrt{\zeta-b_1}$ and recalling the definition \eqref{def of q} of $q$, we obtain
\begin{align}\nonumber
q(\zeta) = & -\frac{2}{3} \frac{c_1+c_2}{\sqrt{2}} \frac{\sqrt{\re b_2}}{b_1} (\zeta-b_1)^{\frac{3}{2}}
	\\ \label{strong asymptotics of q}
& +\frac{(c_1 + c_2)(3i\im b_2 +\re b_2) }{30 \sqrt{2}\, b_1^2 \sqrt{\re b_2}}(\zeta-b_1)^{\frac{5}{2}}   +\bigO\big((\zeta-b_1)^{3}\big), \qquad \zeta \to b_1. 
\end{align}
Expansion of \eqref{def of mathcalR} gives
\begin{align}\label{asymptotics of R as zeta to b1}
\mathcal{R}(\zeta) =-\frac{\sqrt{2}}{\sqrt{\re b_2} } \bigg( i(c_5+c_6)+(c_5-c_6)\frac{b_2}{|b_2|} \bigg) \sqrt{\zeta -b_1} + \bigO\big( (\zeta -b_1)^{\frac{3}{2}} \big).
\end{align}
Substituting \eqref{strong asymptotics of q} and \eqref{asymptotics of R as zeta to b1} into \eqref{explicit JR1} a straightforward calculation shows that $A$ and $B$ can be written as in (\ref{ABexplicit}).

Finally, it follows from (\ref{explicit solution of RH problem R}) and Lemma \ref{lemma: small norm jumps} that the order in which the expansions in $s$ and $\zeta$ are computed is irrelevant for the evaluation of the coefficient $R_1^{(1)}$ defined in \eqref{def of R1^1}. Thus, from \eqref{full expansion R} and \eqref{explicit R1(zeta)}, we have 
$R_1^{(1)} = A-\bar{A}$ and a straightforward computation then gives the expression (\ref{R1122}) for $(R_1^{(1)})_{2,2}$.
This completes the proof of Proposition \ref{Rprop}.

\section{Proof of Theorem \ref{thm: all order expansion} and of the expression (\ref{little c in thm}) for $c$}\label{Section: little c}
In this section, we use the expansions of $p$ and $R$ derived in Sections \ref{psec} and \ref{Rsec} to prove Theorem \ref{thm: all order expansion} and to provide a first proof of the expression (\ref{little c in thm}) for the constant $c$.

\subsection{Proof of Theorem \ref{thm: all order expansion}.} 
Propositions \ref{pprop} and \ref{Rprop} yield expansions for $p_{1}(s)$ and $R_1(s)$ in negative powers of $s^\rho$ to all orders uniformly for $\theta$ in compact subsets of $(0,1]$. 
Indeed, since $p$ is analytic at $\zeta = \infty$, (\ref{def of expansion of p with p0 p1}) implies
$$p_1(s) = \frac{1}{2\pi i} \int_{|\zeta|=r} p(\zeta) d\zeta$$ 
where $r$ is any fixed large radius; substituting in \eqref{full expansion of p}, this gives the following extension of (\ref{p1 asymp from Claeys}) to all orders as $s \to + \infty$:
\begin{equation}\label{p1asymptoticsallorders}
p_1(s) = -ic_{5}|b_{2}| + i \frac{c_{5}+c_{6}}{2}(|b_{2}|-\im{b_{2}}) 
+ \sum_{n=1}^{N} \frac{\frac{1}{2\pi i} \int_{|\zeta|=r} \zeta^{-n}  \mathcal{A}_n(\zeta) d\zeta}{s^{n\rho}}
+ \bigO\bigg( \frac{1}{s^{(N+1)\rho}} \bigg).
\end{equation}
Similarly, by the definition (\ref{def of R1}) of $R_1(s)$ and the expansion (\ref{full expansion R}) of $R(\zeta)$,
$$R_1(s) = \lim_{r \to +\infty} \frac{1}{2\pi i} \int_{|\zeta|=r} R(\zeta) d\zeta
= \lim_{r \to +\infty} \bigg\{\frac{1}{2\pi i} \int_{|\zeta|=r} \sum_{n=1}^{N} \frac{R^{(n)}(\zeta)}{s^{n\rho}}d\zeta +  \int_{|\zeta|=r} \mathfrak{g}(s, \zeta) d\zeta\bigg\}$$ 
where the function $\mathfrak{g}$ obeys the bound $|\mathfrak{g}(\zeta, s)| \leq C' s^{-(N+1)\rho}(1+ |\zeta|)^{-1}$. 
The coefficients $R^{(n)}(\zeta)$ are analytic at $\zeta = \infty$ by Proposition \ref{Rprop}.
Hence
$$R_1(s) = \sum_{n=1}^{N} \frac{R_1^{(n)}}{s^{n\rho}} 
+ \bigO(s^{-(N+1)\rho}), \qquad s \to +\infty,$$ 
where $R_1^{(n)}$ denotes the coefficient of $\zeta^{-1}$ in the large $\zeta$ expansion of $R^{(n)}(\zeta)$, and we have used that
$$\bigg|\lim_{r \to +\infty} \int_{|\zeta|=r} \mathfrak{g}(s, \zeta) d\zeta\bigg|
\leq \limsup_{r \to +\infty} 2\pi r C' s^{-(N+1)\rho}(1+ r)^{-1} = 2\pi C' s^{-(N+1)\rho}.$$

Since \eqref{diff identity in s} expresses $\partial_{s} \ln \det \big(  1- \mathbb{K}\big|_{[0,s]}\big)$ identically in terms of $p_{1}(s)$ and $(R_1(s))_{2,2}$, we deduce the existence of an asymptotic expansion to all orders of $\det (1- \left.\mathbb{K}\right|_{[0,s]})$ as $s \to +\infty$ for each $\theta \in (0,1]$. This proves Theorem \ref{thm: all order expansion}.

\subsection{Proof of the expression (\ref{little c in thm}) for $c$}
Comparing (\ref{p1 asymp from Claeys}) and (\ref{p1asymptoticsallorders}), we see that 
$$\mathcal{K} = \frac{1}{2\pi i} \int_{|\zeta|=r} \frac{\mathcal{A}_{1}(\zeta)}{\zeta} d\zeta,$$
where $r > 0$ is any large radius, i.e., $\mathcal{K}$ is the term of order $1$ in the large $\zeta$ expansion of the function $\mathcal{A}_{1}(\zeta)$ defined in (\ref{def of mathcalA1}).
A direct computation shows that
\begin{align*}
\mathcal{A}_{1}(\zeta)=\frac{c_8- \frac{3\alpha^2-1}{12}}{2|b_2|}\zeta + \frac{i}{2}\bigg\{c_8- \bigg(c_8- \frac{3\alpha^2-1}{12}\bigg)\frac{\im b_2 }{|b_2|}\bigg\} + O(\zeta^{-1}), \qquad \zeta \to \infty,
\end{align*}
and therefore
\begin{align}\label{calKexpression}
\mathcal{K} = \frac{i}{2}\bigg\{ c_8- \bigg(c_8- \frac{3\alpha^2-1}{12}\bigg)\frac{\im b_2 }{|b_2|}\bigg\}.
\end{align}

Substituting the expressions (\ref{R1122}) and (\ref{calKexpression}) for $(R_1^{(1)})_{2,2}$ and $\mathcal{K}$ into \eqref{expression for c in terms of mathcal K and RIp1p intro} and recalling the definition \eqref{def of c1...c8} of the constants $\{c_j\}_1^8$, we obtain the expression \eqref{little c in thm} for $c$.

\begin{remark}\upshape
The above evaluation of the constant $c$ is based on the differential identity (\ref{diff identity in s}) in $s$. In Section \ref{section: integration in theta}, we will obtain an independent second proof of (\ref{little c in thm}) by using a differential identity in $\theta$.
\end{remark}

\begin{remark}[The constant $c$ for two other models]\upshape\label{c1c2remark}
Our approach to obtain the constant $c$ presented in Sections \ref{psec} and \ref{Rsec} is based on the differential identity in $s$ derived in \cite{ClaeysGirSti}. Hence, it also applies to two other random matrix models studied in \cite{ClaeysGirSti}. 
The first model consists of random matrices of the form
\begin{align*}
M^{(1)}= (G_r \ldots G_1)^\ast G_r\ldots G_1,
\end{align*}
where ${}^*$ denotes the complex conjugate transpose operator, and each $G_j$ is an independent $(n+\nu_j)\times(n+\nu_{j-1})$ complex Ginibre matrix, with integers $r\ge1$, $\nu_0=0$, and $\nu_j\ge0$, $j=1,\ldots,r$. The second model consists of products of the form
\begin{align*}
M^{(2)}=(T_r\ldots T_1)^\ast T_r\ldots T_1,
\end{align*}
where each $T_{j}$ is an $(n+\nu_j)\times(n+\nu_{j-1})$ upper left truncation of an $\ell_j \times \ell_j$ Haar distributed unitary matrix $U_j$. Here $U_1,\ldots,U_r$ are assumed to be independent and $\nu_0=0$, $r\ge1$, and $\nu_j\ge0$, $j=1,\ldots,r$, are integers. Furthermore, it is assumed that $\ell_j \ge n+\nu_j+1$ and $\sum_{j=1}^{r}(\ell_j-n-\nu_j)\ge n$. In the second model, a subset $J\subset \{2,\ldots,r\}$ of cardinality $q<r$ is fixed such that $\mu_j:=\ell_{k_j}-n>\nu_j$ for $k_j \in  J$ and $\ell_k -n \to +\infty$ for $k\in \{ 1,\ldots r \} \setminus J$ as $n$ and $\ell_1,\ldots,\ell_r$ go to infinity. In \cite{ClaeysGirSti}, it is shown that these two models admit large gap asymptotics for the eigenvalues of the form
\begin{align*}
\mathbb{P}^{(j)}(\mbox{gap on } [0,s]) = C^{(j)} \exp \left( -a^{(j)} s^{2\rho^{(j)}} + b^{(j)} s^{\rho} + c^{(j)} \ln s \right) (1 + o(1)) \qquad  \mbox{as }s \to + \infty,
\end{align*}
where the first and second model corresponds to $j=1$ and $j=2$, respectively. Moreover, explicit expressions are derived for the constants $\rho^{(j)}$, $a^{(j)}$, and $b^{(j)}$. 

A straightforward modification of our approach yields the existence of constants $C_1^{(j)},\ldots,C_N^{(j)} \in \mathbb{R}$ such that
\begin{equation}
\mathbb{P}^{(j)}(\mbox{gap on } [0,s])  = C^{(j)} \exp \Big(  -a^{(j)} s^{2\rho^{(j)}}+b^{(j)} s^{\rho}+c^{(j)} \ln s + \sum_{j=1}^{N}C_j^{(j)} s^{-j\rho} + \bigO\big(s^{-(N+1)\rho}\big) \Big),
\end{equation}
as $s\to+\infty$ for $j=1,2$, and shows that the constants $c^{(1)}$ and $c^{(2)}$ are given explicitly by
\begin{align}
c^{(1)}&= \frac{r-1}{12(r+1)} - \frac{1}{2(r+1)}  \sum_{j=1}^{r} \nu_j^2, \label{expr for c1 first pp}
\\
c^{(2)}&= \frac{r-q-1}{12(r-q+1)}- \frac{1}{2(r-q+1)} \bigg( \sum_{j=1}^r \nu_j^2 - \sum_{j=1}^{q} \mu_j^2 \bigg). \nonumber
\end{align}
Let $\mathbb{K}^{(1)}$ be the hard edge limiting kernel for the eigenvalues associated to the first model presented above (this is the same notation as in \cite{ClaeysGirSti}). For certain particular choices of the parameters $\nu_{1}, \ldots,\nu_{r}$ and $\theta$, the kernel $\mathbb{K}^{(1)}$ defines the same point process (up to rescaling) as the one associated to $\mathbb{K}$\footnote{$\mathbb{K}$ in the present paper is denoted by $\mathbb{K}^{(3)}$ in \cite{ClaeysGirSti}.}--this is a result of Kuijlaars and Stivigny, see \cite[Theorem 5.1]{KuijSti2014}. More precisely, if $r \geq 1$ is an integer, $\alpha > -1$ and
\begin{align}
& \theta = \frac{1}{r}, & & \nu_{j} = \alpha + \frac{j-1}{r}, \qquad j = 1,\ldots,r, \label{cond parameters for relation first and third kernel}
\end{align}
then the kernels $\mathbb{K}^{(3)}$ and $\mathbb{K}$ are related by
\begin{align*}
\left( \frac{x}{y} \right)^{\alpha}\mathbb{K}^{(1)}(x,y) = r^{r} \mathbb{K}^{(3)}(r^{r}x,r^{r}y).
\end{align*}
Therefore, if the parameters satisfy \eqref{cond parameters for relation first and third kernel}, we obtain the following relations:\footnote{The quantities $\rho$, $a$, $b$, $c$ and $C$ in the present paper are denoted by $\rho^{(3)}$, $a^{(3)}$, $b^{(3)}$,  $c^{(3)}$ and $C^{(3)}$ in \cite{ClaeysGirSti}.}
\begin{align}
& \rho^{(1)} = \rho, \qquad a^{(1)} = a r^{2r\rho}, \qquad b^{(1)} = b r^{r\rho}, \label{rho a b first and third kernel} \\
& c^{(1)} = c, \qquad C^{(1)} = r^{rc}C.
\end{align}
The three relations in \eqref{rho a b first and third kernel} can be verified from \cite{ClaeysGirSti}, and the relation $c^{(1)} = c$ can be verified directly from \eqref{little c in thm} and \eqref{expr for c1 first pp}. This provides a non-trivial consistency check of the results from \cite{ClaeysGirSti} and of our result for $c$ and $c^{(1)}$.
\end{remark}

\section{Differential identity in $\theta$}\label{Section: diff identity in theta}
In this section, we derive an identity for the derivative of $\ln \det(1-\mathbb{K}|_{[0,s]})$ with respect to $\theta$. As explained in Section \ref{outlinesubsec}, this differential identity is needed for the derivation of the expression (\ref{big C in thm}) for $C$. Our proof of Lemma \ref{lemma diff theta first 1} below is inspired by the derivation of the differential identity \eqref{diff identity in s} given in \cite{ClaeysGirSti}.

\begin{lemma}[Differential identity in $\theta$, $1$st version]\label{lemma diff theta first 1}
For every $\alpha > -1$, $\theta > 0$ and $s>0$, the following identity holds:
\begin{align}\nonumber
& \partial_{\theta} \ln \det \Big( \left. 1-\mathbb{K} \right|_{[0,s]} \Big) 
=  \frac{1}{2} \int_{\gamma \cup \widetilde{\gamma}} \partial_{\theta} \ln \Gamma \Big( \frac{\frac{\alpha}{2}+1-z}{\theta} \Big) \tr[Y_{+}^{-1}(z)Y_{+}^{\prime}(z)\sigma_{3} - Y_{-}^{-1}(z)Y_{-}^{\prime}(z)\sigma_{3}] \frac{dz}{2\pi i}  
	\\ \label{first diff identity lemma} 
& =  \frac{-1}{2\theta^{2}} \int_{\gamma \cup \widetilde{\gamma}} \Big(\frac{\alpha}{2}+1-z\Big) \psi \Big( \frac{\frac{\alpha}{2}+1-z}{\theta} \Big) \tr[Y_{+}^{-1}(z)Y_{+}^{\prime}(z)\sigma_{3} - Y_{-}^{-1}(z)Y_{-}^{\prime}(z)\sigma_{3}] \frac{dz}{2\pi i}. 
\end{align}
where $\psi = (\ln \Gamma)^{\prime}$ is the di-gamma function.
\end{lemma}
\begin{proof}
From \cite[Theorem 2.1]{BertolaCafasso} and \cite[Eq. (2.20)]{ClaeysGirSti}, letting $\theta$ play the role of the deformation parameter, we deduce that
\begin{equation}\label{first diff id}
\partial_{\theta} \ln \det \Big( \left. 1-\mathbb{K} \right|_{[0,s]} \Big) = \int_{\gamma \cup \widetilde{\gamma}} \tr[Y_{-}^{-1}(z)Y_{-}^{\prime}(z)\partial_{\theta}J(z)J^{-1}(z)] \frac{dz}{2\pi i},
\end{equation}
where $J(z) := Y_{-}^{-1}(z)Y_{+}(z)$, i.e.,
\begin{align*}
J(z) = \left\{ \begin{array}{l l}
\ds \begin{pmatrix}
1 & -s^{-z}F(z) \\
0 & 1
\end{pmatrix}, & z \in \gamma, \\[0.3cm]
\ds \begin{pmatrix}
1 & 0 \\
s^{z}F(z)^{-1} & 1
\end{pmatrix}, & z \in \tilde{\gamma}.
\end{array} \right.
\end{align*}
Therefore, we obtain
\begin{equation*}
\partial_{\theta}J(z)J(z)^{-1} = \partial_{\theta} \ln \Gamma \Big( \frac{\frac{\alpha}{2}+1-z}{\theta} \Big) (J(z)-I)\sigma_{3},
\end{equation*}
from which it follows that
\begin{equation}\label{lol2}
\partial_{\theta} \ln \det \Big( \left. 1-\mathbb{K} \right|_{[0,s]} \Big) = \int_{\gamma \cup \widetilde{\gamma}} \partial_{\theta} \ln \Gamma \Big( \frac{\frac{\alpha}{2}+1-z}{\theta} \Big) \tr[Y_{-}^{-1}Y_{-}^{\prime}(J-I)\sigma_{3}] \frac{dz}{2\pi i}.
\end{equation}
Since $J$ is triangular and $J-I$ is off-diagonal, using also the jump relations for $Y$, we infer that
\begin{equation*}
J\sigma_{3}J^{-1} = (2J-I)\sigma_{3},
\end{equation*}
from which we obtain
\begin{equation}\label{observation 1 inside lemma}
\tr[Y_{+}^{-1}Y_{+}^{\prime}\sigma_{3}] = 2 \tr[Y_{-}^{-1}Y_{-}^{\prime}J\sigma_{3}]-\tr[Y_{-}^{-1}Y_{-}^{\prime}\sigma_{3}].
\end{equation}
A similar computation yields
\begin{equation}\label{observation 2 inside lemma}
\tr[JY_{+}^{-1}Y_{+}^{\prime}\sigma_{3}] = \tr[Y_{-}^{-1}Y_{-}^{\prime}J\sigma_{3}].
\end{equation}
By substituting \eqref{observation 2 inside lemma} in \eqref{observation 1 inside lemma}, we obtain
\begin{equation}\label{observation 3 inside lemma}
\tr[(I-J)Y_{+}^{-1}Y_{+}^{\prime}\sigma_{3}] = \tr[Y_{-}^{-1}Y_{-}^{\prime}(J-I)\sigma_{3}].
\end{equation}
Using \eqref{observation 2 inside lemma} and \eqref{observation 3 inside lemma}, we arrive at
\begin{align*}
\tr[Y_{-}^{-1}Y_{-}^{\prime}(J-I)\sigma_{3}] 
& = \ds \frac{1}{2} \Big( \tr[(I-J)Y_{+}^{-1}Y_{+}^{\prime}\sigma_{3}] + \tr[Y_{-}^{-1}Y_{-}^{\prime}(J-I)\sigma_{3}] \Big) 
	\\
& = \ds \frac{1}{2} \Big( \tr[Y_{+}^{-1}Y_{+}^{\prime}\sigma_{3}] - \tr[Y_{-}^{-1}Y_{-}^{\prime}\sigma_{3}] \Big).
\end{align*}
Substitution of the above identity into \eqref{lol2} finishes the proof.
\end{proof}
In the following lemma, we rewrite the differential identity \eqref{first diff identity lemma} in a form which is more convenient for the asymptotic analysis. Let us define the sequence $\{\zeta_j\}_0^\infty \subset i\R$ by 
\begin{align}\label{poles of digamma}
\zeta_j = -i\frac{\frac{1 +\alpha}{2} + j\theta}{s^\rho}, \qquad j = 0,1,2, \dots,
\end{align}
and the meromorphic function $H(\zeta)$ by
\begin{align}\label{def of H}
H(\zeta) = \frac{1}{\theta^{2}} \left( \frac{1+\alpha}{2}- i s^{\rho}\zeta \right) \psi \bigg(\frac{\frac{1+\alpha}{2}-is^{\rho}\zeta}{\theta} \bigg).
\end{align}
Note that $H$ has a simple pole at each of the points $\zeta_j$, $j = 1,2, \dots$, and no other poles in $\C$; the point $\zeta_0$ is a simple pole of $\psi(\frac{\frac{1+\alpha}{2}-is^{\rho}\zeta}{\theta})$ but not of $H$.

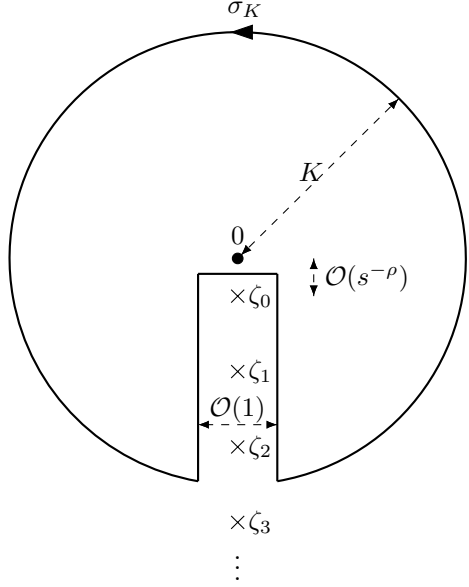
\begin{figure}
\begin{center}
\begin{tikzpicture}
\draw[fill] (0,0) circle (0.07cm);
\node at (0.1,3.3) {$\sigma_K$};

\node at (0,0.3) {$0$};

\node at (0.3,-.5) {$\zeta_0$};
\node at (0.3,-1.5) {$\zeta_1$};
\node at (0.3,-2.5) {$\zeta_2$};
\node at (0.3,-3.5) {$\zeta_3$};

\draw (0,-0.5) node[cross=0.1cm,rotate=0] {};

\draw (0,-1.5) node[cross=0.1cm,rotate=0] {};

\draw (0,-2.5) node[cross=0.1cm,rotate=0] {};

\draw (0,-3.5) node[cross=0.1cm,rotate=0] {};

\node at (0,-4) {$\vdots$};

\draw[thick,black] ([shift=(-80:3cm)]0,0) arc (-80:260:3);

\draw[thick] (-80:3)--($(-80:3)+(0,2.75)$);
\draw[thick] (-100:3)--($(-100:3)+(0,2.75)$);
\draw[thick] ($(-80:3)+(0,2.75)$)--($(-100:3)+(0,2.75)$);

\draw[dashed,-<-=0.02,->-=1] (0,0) -- (45:3);
\node at (50:1.5) {$K$};

\draw[dashed,-<-=0.0,->-=1] ($(-80:3)+(0,.75)$)--($(-100:3)+(0,.75)$);
\node at (-90:2) {$\bigO(1)$};

\draw[dashed,-<-=0.0,->-=1] ($(1,0)$)--($(1,-0.5)$);
\node at (1.7,-0.25) {$\bigO(s^{-\rho})$};


\draw[black,arrows={-Triangle[length=0.3cm,width=0.2cm]}]
($(-0.1,0)+(90:3)$) --  ++(-0.0001,0);

\end{tikzpicture}
\end{center}
\caption{The contour $\sigma_{K}$ and the poles $\{\zeta_{j}\}_0^\infty$ of $\psi(\frac{\frac{1+\alpha}{2}-is^{\rho}\zeta}{\theta})$ in the complex $\zeta$-plane. The uppermost pole $\zeta_0$ lies a distance $\bigO(s^{-\rho})$ from the origin as $s \to + \infty$. The horizontal line segment has a length of order $\bigO(1)$ as $s \to +\infty$ and crosses the imaginary axis half-way between the origin and $\zeta_0$. 
\label{fig: sigmaK}}
\end{figure}

Given $K > |b_1|$, we let $\sigma_K$ denote the closed $s$-dependent counterclockwise contour displayed in Figure \ref{fig: sigmaK}. The contour $\sigma_{K}$ surrounds $\Sigma_{5}$ once in the positive direction, but does not surround any of the poles $\zeta_{j}$ of $H$. The circular part of $\sigma_{K}$ has radius $K$ and its horizontal part has a length of order $\bigO(1)$ as $s \to +\infty$ and crosses the imaginary axis at the point $\zeta_0/2$. 
If $K = 2|b_1|$, we write $\sigma$ for $\sigma_K$, i.e., $\sigma = \sigma_{2|b_1|}$.
We also define the contour $\widetilde{\Sigma}_{K}$ as the union of the parts exterior to $\sigma_K$ of the rays $\{\Sigma_{i}\}_1^4$ defined in (\ref{Sigmaidef}), i.e., 
\begin{align*}
\widetilde{\Sigma}_{K} = \bigcup_{i=1}^{4} \Sigma_{i} \setminus \{|\zeta| \leq K\}.
\end{align*}

\begin{lemma}[Differential identity in $\theta$, $2$nd version]\label{lemma: simplified diff dientity}
Let $K$ be such that $K > 2|b_{1}|$. Then
\begin{align}
& \partial_{\theta} \ln \det \Big( \left. 1-\mathbb{K} \right|_{[0,s]} \Big) = I_{1} + I_{2} + I_{3,K} + I_{4,K}, \label{diff identity simplified}
\end{align}
where
\begin{align}\label{I1def}
I_{1} & = s^{\rho} \int_{\sigma} H(\zeta) g'(\zeta) \frac{d\zeta}{2\pi i}, 
	\\\label{I2def}
I_{2} & = \frac{1}{2} \int_{\sigma} H(\zeta) \tr \Big[ P^{\infty}(\zeta)^{-1}P^{\infty}(\zeta)' \sigma_{3} \Big] \frac{d\zeta}{2\pi i},  
	\\\label{I3Kdef}
I_{3,K} & = \frac{1}{2} \int_{\sigma_{K}} H(\zeta) \tr \Big[P^{\infty}(\zeta)^{-1}e^{-p_{0}\sigma_{3}}R^{-1}(\zeta)R'(\zeta)e^{p_{0}\sigma_{3}}P^{\infty}(\zeta)\sigma_{3}\Big] \frac{d\zeta}{2\pi i},  
	\\\label{I4Kdef}
I_{4,K} & = -\frac{1}{2} \int_{\widetilde{\Sigma}_{K}} H(\zeta) \tr \Big[P^{\infty}(\zeta)^{-1}e^{-p_{0}\sigma_{3}}\Big(R_{+}^{-1}(\zeta)R_{+}'(\zeta)-R_{-}^{-1}(\zeta)R_{-}'(\zeta)\Big)e^{p_{0}\sigma_{3}}P^{\infty}(\zeta)\sigma_{3}\Big]. 
\end{align} 
\end{lemma}
\begin{proof}
Using the change of variable $z = is^{\rho} \zeta + \frac{1}{2}$ in \eqref{first diff identity lemma}, we obtain an integral over $\gamma_{U}\cup \tilde{\gamma}_{U}$ whose integrand is expressed in terms of $U$ via \eqref{def of U}. By deforming the contour of this integral using the analytic continuations of $U_{+}$ and $U_{-}$ (i.e., using $T$), we arrive at
\begin{align*}
\partial_{\theta} \ln \det \Big( \left. 1-\mathbb{K} \right|_{[0,s]} \Big) = -\frac{1}{2} \int_{\cup_{i=1}^{5}\Sigma_{i}} H(\zeta)
\tr[T_{+}^{-1}(\zeta)T_{+}^{\prime}(\zeta)\sigma_{3} - T_{-}^{-1}(\zeta)T_{-}^{\prime}(\zeta)\sigma_{3}] \frac{d\zeta}{2\pi i}. 
\end{align*}
Another contour deformation gives
\begin{align}\nonumber
\partial_{\theta} \ln \det \Big( \left. 1-\mathbb{K} \right|_{[0,s]} \Big)  
 =&\; \frac{1}{2} \int_{\sigma_{K}} H(\zeta) \tr[T^{-1}(\zeta)T^{\prime}(\zeta)\sigma_{3}] \frac{d\zeta}{2\pi i} 
 	\\\label{lol3}
& -\frac{1}{2} \int_{\cup_{i=1}^{5}\Sigma_{i}} H(\zeta)
\tr[T_{+}^{-1}(\zeta)T_{+}^{\prime}(\zeta)\sigma_{3} - T_{-}^{-1}(\zeta)T_{-}^{\prime}(\zeta)\sigma_{3}] \frac{d\zeta}{2\pi i}.
\end{align}

For $\zeta \in \sigma_{K}$, we have $\zeta \notin  \Done \cup \Dtwo$. Therefore, inverting the transformations $T \mapsto S \mapsto R$ for $\zeta \in \sigma_{K}$, we find
\begin{align}\nonumber
\tr[T^{-1}(\zeta)T^{\prime}(\zeta)\sigma_{3}] = & \; 2s^{\rho} g'(\zeta) + \tr [P^{\infty}(\zeta)^{-1}P^{\infty}(\zeta)'\sigma_{3}] 
	\\ \label{TrTinvTprime}
& + \tr [P^{\infty}(\zeta)^{-1}e^{-p_{0}\sigma_{3}}R^{-1}(\zeta)R'(\zeta)e^{p_{0}\sigma_{3}}P^{\infty}(\zeta)\sigma_{3}].
\end{align}
The first two terms on the right-hand side of (\ref{TrTinvTprime}) are analytic in the region between $\sigma$ and $\sigma_{K}$. Therefore, substituting (\ref{TrTinvTprime}) into the first term on the right-hand side of \eqref{lol3} and deforming the contour from $\sigma_K$ to $\sigma$ in the integrals involving the first two terms on the right-hand side of (\ref{TrTinvTprime}), we find that this term equals $I_{1} + I_{2} + I_{3,K}$. 

Similarly, by inverting the transformations $T \mapsto S \mapsto R$ for $\zeta \in \widetilde{\Sigma}_{K}$, we find that the second term on the right-hand side of (\ref{lol3}) equals $I_{4,K}$.
\end{proof}

\begin{remark}\upshape
In the application of the differential identity (\ref{diff identity simplified}) to the proof of Theorem \ref{thm:main results}, we will choose $K = s^\rho$; that is, the radius $K$ will be $s$-dependent and growing to infinity as $s \to +\infty$. 
\end{remark}

The remainder of the paper is devoted to the proof of Theorem \ref{thm:main results}. The proof is divided into two steps. The first step consists of obtaining large $s$ asymptotics of the differential identity \eqref{diff identity simplified} uniformly for $\theta$ in compact subsets of $(0,1]$. This is achieved by computing the large $s$ asymptotics of each of the four terms $I_1$, $I_2$, $I_{3,K}$, and $I_{4,K}$ on the right-hand side of \eqref{diff identity simplified}. These computations are presented in Sections \ref{I1sec}-\ref{I3I4sec}.
The second step is presented in Section \ref{section: integration in theta} and consists of integrating the resulting asymptotic expansion from $\theta = 1$ to an arbitrary $\theta \in (0,1]$.

\section{Asymptotics of $I_{1}$}\label{I1sec}
In this section, we prove the following proposition which establishes the large $s$ asymptotics of $I_{1}$. 

\begin{proposition}[Large $s$ asymptotics of $I_{1}$]\label{I1prop}
Let $\alpha > -1$. As $s \to +\infty$, the function $I_1$ defined in (\ref{I1def}) satisfies
\begin{align}
I_{1} = \mathcal{I}_{1}^{(1)} s^{2\rho} \ln(s^\rho) + I_{1}^{(1)} s^{2\rho} 
+ \mathcal{I}_{1}^{(2)} s^{\rho} \ln(s^\rho) + I_{1}^{(2)} s^{\rho} 
+ \mathcal{I}_{1}^{(3)}\ln(s^\rho) + I_{1}^{(3)} + \bigO(s^{-\rho}\ln (s^{\rho})) \label{asymp for I1}
\end{align}
uniformly for $\theta$ in compact subsets of $(0,1]$, where the coefficients $\mathcal{I}_{1}^{(1)}$, $I_{1}^{(1)}$, $\mathcal{I}_{1}^{(2)}$, $I_{1}^{(2)}$, $\mathcal{I}_{1}^{(3)}$, $I_{1}^{(3)}$ are given by
\begin{subequations}\label{cal1coeffs}
\begin{align} 
\mathcal{I}_{1}^{(1)} = & - \frac{2a}{\rho(1+\theta)^2}, 
	\\ 
 I_1^{(1)} = & - \partial_\theta a, 
 	\\
\mathcal{I}_{1}^{(2)} = &\; 0, 
	\\
I_{1}^{(2)} =& -\frac{(1+\theta)(1+\alpha)}{2\theta}\theta^{-\frac{2\theta}{1+\theta}}, 
	\\
\mathcal{I}_{1}^{(3)} = &\; \frac{3(1+\alpha)^{2}-2\theta^{2}}{24\theta^{2}}, 
	\\ \nonumber
I_{1}^{(3)} = &\; \frac{1+\alpha}{4\theta} \ln(2\pi) + \frac{3(1+\alpha)^{2}-2\theta^{2}}{24\theta^{2}} \left( - \frac{2\theta}{\theta+1}\ln \theta + \frac{\theta-1}{\theta} \ln(1+\theta) \right) 
	\\ 
& + \zeta^{\prime}(-1) - \ln G \Big( \frac{1+\alpha+2\theta}{2\theta} \Big). 
\end{align}
\end{subequations}
\end{proposition}
\begin{proof}
Recall from (\ref{poles of digamma}) that $\zeta_0 = -i\frac{1 +\alpha}{2}s^{-\rho}$.
Define $\Psi(\zeta) = \Psi(\zeta, s, \theta, \alpha)$ by
\begin{align}\label{def of Psi}
\Psi(\zeta) = s^{\rho}  \int_{\zeta_{\star}}^\zeta H(\xi) d\xi,
\end{align}
where $\zeta_{\star} \in \C \setminus (-i\infty, \zeta_0]$ is some point at which $\Psi$ is normalized to vanish; we will choose this normalization point below. 
Then $\Psi$ is analytic in $\C \setminus (-i\infty, \zeta_0]$. In particular, $\Psi$ is analytic on $\sigma$. Using the explicit expression \eqref{def of g''} for $g''$, an integration by parts therefore gives
\begin{align*}
I_{1} & = - \int_{\sigma} \Psi(\zeta) g''(\zeta) \frac{d\zeta}{2\pi i}
	\\
& = i\frac{c_1 + c_2}{2}\int_{\sigma} \Psi(\zeta) \bigg(\frac{1}{\zeta} - \frac{1}{r(\zeta)} + \frac{i \im{b_1}}{\zeta r(\zeta)}\bigg) \frac{d\zeta}{2\pi i}
	\\
& = i\frac{c_1 + c_2}{2}\Psi(0)
+ i\frac{c_1 + c_2}{2}\int_{\sigma} \Psi(\zeta) \bigg(- \frac{1}{r(\zeta)} + \frac{i \im{b_1}}{\zeta r(\zeta)}\bigg) \frac{d\zeta}{2\pi i}.
\end{align*}

We assume that $\sigma$ is big enough to enclose the straight line segment $[b_1, b_2]$ and move the branch cut for $r(\zeta)$ upwards from $\Sigma_5$ to the horizontal line segment $[b_1, b_2]$; this does not change the value of the integral. We let $\tilde{r}$ denote the analytic continuation of $r$ defined by
\begin{align}\label{rtildedef}
\tilde{r}(\zeta) = [(\zeta-b_{1})(\zeta-b_{2})]^{\frac{1}{2}},
\end{align}
where the branch is such that $\tilde{r}$ is analytic in $\mathbb{C}\setminus [b_{1},b_{2}]$ and $\tilde{r}(\zeta) \sim \zeta$ as $\zeta \to \infty$. Then $\tilde{r}(\zeta)$ is equal to $r(\zeta)$ except for $\zeta$ in the region enclosed by $\Sigma_5 \cup [b_{1},b_{2}]$ where we instead have $\tilde{r}(\zeta) = -r(\zeta)$. Deforming $\sigma$ upwards through the origin, a residue contribution is generated by the simple pole of $i \im{b_1}/(\zeta \tilde{r}(\zeta))$ at $\zeta = 0$. We find
\begin{align*}
I_{1} = &\; i\frac{c_1 + c_2}{2}\Psi(0)
+ i\frac{c_1 + c_2}{2}\bigg\{\Psi(0) \frac{i \im{b_1}}{\tilde{r}(0)}
	\\
& +
\int_{[b_1, b_2]} \Psi(\zeta) \bigg[\bigg(- \frac{1}{\tilde{r}(\zeta)} + \frac{i \im{b_1}}{\zeta \tilde{r}(\zeta)}\bigg)_-
- 
\bigg(- \frac{1}{\tilde{r}(\zeta)} + \frac{i \im{b_1}}{\zeta \tilde{r}(\zeta)}\bigg)_+
\bigg] \frac{d\zeta}{2\pi i}
\bigg\}
\end{align*}
where $[b_1, b_2]$ is oriented from  $b_1$ to $b_2$ with $+$ and $-$ sides to the left and right as usual, and $\tilde{r}(0) = r_-(0) = -i|b_2|$.
Thus,
\begin{align}\nonumber
I_{1} = &\; i\frac{c_1 + c_2}{2}\Psi(0)\bigg(1 - \frac{\im{b_1}}{|b_1|}\bigg) -2i\frac{c_1 + c_2}{2} \int_{[b_1, b_2]} \Psi(\zeta) 
\bigg(- \frac{1}{\tilde{r}_+(\zeta)} + \frac{i \im{b_1}}{\zeta \tilde{r}_+(\zeta)}\bigg)
\frac{d\zeta}{2\pi i} 
	\\\label{lol4}
= &\; i\frac{c_1 + c_2}{2}\Psi(0)\bigg(1 - \frac{\im{b_1}}{|b_1|}\bigg) - 2i\frac{c_1 + c_2}{2} \int_{\gamma_{b_2b_1}} \Psi(\zeta) 
\bigg(1 - \frac{i \im{b_1}}{\zeta}\bigg)\frac{1}{r(\zeta)}
\frac{d\zeta}{2\pi i}, 
\end{align}
where $\gamma_{b_2b_1}$ denotes the part of the circle of radius $|b_2|$ centered at the origin going from $b_2$ to $b_1$ and oriented counterclockwise. 

Let us choose $\zeta_{\star} = 0$; then $\Psi(0) = 0$, so the first term on the right-hand side of (\ref{lol4}) vanishes. The choice $\zeta_{\star} = 0$ implies that the term $\frac{\frac{1+\alpha}{2}-is^{\rho}\xi}{ \theta}$ is not uniformly large for $\xi \in [0,\zeta]$ with $\zeta \in \gamma_{b_{1}b_{2}}$ as $s \to + \infty$, 
so the large $s$ behavior of $\Psi(\zeta)$ does not follow immediately from (\ref{def of H}) and \eqref{def of Psi}; however, we can determine the large $s$ asymptotics of $\Psi(\zeta)$ as follows. Using the change of variables
\begin{equation}
x = \frac{1}{\theta}\left( \frac{1+\alpha}{2}-i s^{\rho}\xi \right), \qquad dx = \frac{-is^{\rho}}{\theta}d\xi,
\end{equation}
we can write
\begin{align*}
\ds \Psi(\zeta) = &\; i \int_{0}^{\zeta} \frac{1}{\theta}\left( \frac{1+\alpha}{2}-i s^{\rho }\xi \right) \psi \left( \frac{\frac{1+\alpha}{2}-is^{\rho}\xi}{\theta} \right) \frac{-is^{\rho}}{\theta}d\xi \\[0.35cm]
=&\; i \int_{z_{\star}}^{z} x \psi(x)dx = i \int_{z_{\star}}^{z} x \partial_{x}\ln \Gamma(x)dx,
\end{align*}
where
\begin{equation*}
z = \frac{1}{\theta}\left( \frac{1+\alpha}{2}-is^{\rho }\zeta \right), \qquad z_{\star} = \frac{1+\alpha}{2\theta}.
\end{equation*}
Integrating by parts, we get
\begin{equation}\label{expression for Psi}
\Psi(\zeta) = i \left( \Big[x \ln \Gamma(x)\Big]_{z_{\star}}^{z} - \int_{z_{\star}}^{z} \ln \Gamma(x)dx \right).
\end{equation}
Using the well-known identity (see e.g. \cite[Eq. 5.17.4]{NIST})
\begin{equation}
\int_{1}^{z} \ln \Gamma(x)dx = \frac{z-1}{2}\ln(2\pi) - \frac{(z-1)z}{2}+(z-1)\ln \Gamma(z) - \ln G(z)
\end{equation}
in \eqref{expression for Psi}, we obtain
\begin{equation}\label{lol1}
\Psi(\zeta) = i \left( \frac{is^{\rho}\zeta}{2\theta}\Big[ \ln(2\pi)+1 \Big] - \frac{i s^{\rho}\zeta}{2\theta} \left( \frac{1+\alpha}{\theta}-\frac{is^{\rho}\zeta}{\theta} \right) + \ln \frac{\Gamma \left( \frac{1+\alpha}{2\theta}-\frac{is^{\rho }\zeta}{\theta} \right)}{\Gamma\left( \frac{1+\alpha}{2\theta} \right)}+ \ln \frac{G \left( \frac{1+\alpha}{2\theta}-\frac{is^{\rho }\zeta}{\theta} \right)}{G\left( \frac{1+\alpha}{2\theta} \right)} \right).
\end{equation}
The above expression is convenient since the large $z$ asymptotics of $\Gamma(z)$ and $G(z)$ are known (see e.g. \cite[Eqs. 5.11.1 and 5.17.5]{NIST}):
\begin{align}
& \ln G(z+1) = \frac{z^{2}}{4}+z \ln \Gamma(z+1)-\left(\frac{z(z+1)}{2}+\frac{1}{12}\right)\ln z - \frac{1}{12} + \zeta^{\prime}(-1) + \bigO(z^{-2}), \label{large z asymp for log Barnes G} \\
& \ln \Gamma(z) = (z-\tfrac{1}{2})\ln z - z + \tfrac{1}{2}\ln(2\pi) + \frac{1}{12z} + \bigO(z^{-3}), \label{large z asymp for log Gamma}
\end{align}
as $z \to \infty$ with $|\arg z|< \pi$, where $\zeta$ is Riemann's zeta function.\footnote{The Riemann zeta function $\zeta(\cdot)$ should not be confused with the complex variable $\zeta$ introduced in (\ref{zetadef}).}
Expanding \eqref{lol1} as $s^{\rho}\zeta \to \infty$, we get
\begin{align}\nonumber
\ds \Psi(\zeta)=& -\frac{i \zeta^{2}}{2\theta^{2}}s^{2\rho}\ln(s^{\rho}) + \frac{i \zeta^{2}}{4\theta^{2}}\Big( 1-2 \ln \Big( \frac{-i \zeta}{\theta} \Big) \Big)s^{2\rho} + \frac{1+\alpha}{2\theta^{2}}\zeta s^{\rho}\ln(s^{\rho}) 
	\\\nonumber
& + \frac{1}{2\theta^{2}}\Big( -\zeta \theta + (1+\alpha)\zeta \ln \Big( \frac{-i \zeta}{\theta} \Big) \Big) s^{\rho} + i \frac{3(1+\alpha)^{2}-2\theta^{2}}{24\theta^{2}} \ln(s^{\rho})
	\\\nonumber
& + \frac{i}{24 \theta^{2}}\Big(6(1+\alpha)\theta \ln(2\pi) + (3(1+\alpha)^{2}-2\theta^{2})\ln \Big( \frac{-i  \zeta}{\theta} \Big)\Big) 
	\\  \label{asymp for Psi until constant}
& + i \Big( \zeta^{\prime}(-1) - \ln \Gamma \Big( \frac{1+\alpha}{2\theta} \Big)- \ln G \Big( \frac{1+\alpha}{2\theta} \Big) \Big) + \bigO \Big( \frac{1}{\zeta s^{\rho}} \Big).
\end{align}
Substituting \eqref{asymp for Psi until constant} into \eqref{lol4}, we find that $I_1$ satisfies (\ref{asymp for I1}) as $s \to +\infty$ with coefficients given by
\begin{align}\nonumber
\mathcal{I}_{1}^{(1)} & = - 2i\frac{c_1 + c_2}{2} \int_{\gamma_{b_2b_1}}  -\frac{i\zeta^2}{2 \theta^2}  
\bigg(1 - \frac{i \im{b_1}}{\zeta}\bigg)\frac{1}{r(\zeta)}
\frac{d\zeta}{2\pi i} 
	\\ \nonumber
& = - \frac{c_1 + c_2}{2 \theta^2} 
\bigg(\int_{\gamma_{b_2b_1}} \frac{\zeta^2}{r(\zeta)} \frac{d\zeta}{2\pi i} 
- i \im{b_1} \int_{\gamma_{b_2b_1}} \frac{\zeta}{r(\zeta)}
\frac{d\zeta}{2\pi i} \bigg), 
	\\ \nonumber
I_{1}^{(1)} & = - 2i\frac{c_1 + c_2}{2} \int_{\gamma_{b_2b_1}}  \frac{i \zeta^2 (1- 2 \ln(-\frac{i}{\theta}) - 2\ln{\zeta})}{4 \theta^2}
\bigg(1 - \frac{i \im{b_1}}{\zeta}\bigg)\frac{1}{r(\zeta)}
\frac{d\zeta}{2\pi i} 
	\\ \nonumber
& = -\frac{1- 2 \ln(-\frac{i}{\theta})}{2} \mathcal{I}_{1}^{(1)}	
- \frac{c_1 + c_2}{2 \theta^2} \int_{\gamma_{b_2b_1}} \zeta^2 \ln(\zeta)
\bigg(1 - \frac{i \im{b_1}}{\zeta}\bigg)\frac{1}{r(\zeta)}
\frac{d\zeta}{2\pi i} 
	\\ \nonumber
& = -\frac{1- 2 \ln(-\frac{i}{\theta})}{2} \mathcal{I}_{1}^{(1)}	
- \frac{c_1 + c_2}{2 \theta^2}
\bigg\{\int_{\gamma_{b_2b_1}} \frac{\zeta^2 \ln(\zeta)}{r(\zeta)}
\frac{d\zeta}{2\pi i}
- i \im{b_1} \int_{\gamma_{b_2b_1}} \frac{\zeta \ln(\zeta)}{r(\zeta)} \frac{d\zeta}{2\pi i}\bigg\},
	\\ \nonumber
\mathcal{I}_{1}^{(2)} & = - 2i\frac{c_1 + c_2}{2} \int_{\gamma_{b_2b_1}}  \frac{(\alpha +1) \zeta }{2 \theta^2} 
\bigg(1 - \frac{i \im{b_1}}{\zeta}\bigg)\frac{1}{r(\zeta)}
\frac{d\zeta}{2\pi i} 
	 \\\nonumber
& = - 2i\frac{c_1 + c_2}{2} \frac{\alpha +1}{2 \theta^2} 
\bigg\{\int_{\gamma_{b_2b_1}}  
\frac{\zeta}{r(\zeta)} \frac{d\zeta}{2\pi i} 
- i \im{b_1} \int_{\gamma_{b_2b_1}} \frac{1}{r(\zeta)} \frac{d\zeta}{2\pi i} \bigg\}, 
	\\\nonumber
I_{1}^{(2)} & = - 2i\frac{c_1 + c_2}{2} \int_{\gamma_{b_2b_1}} \frac{\zeta  \left(-\theta +(\alpha +1) \ln(-\frac{i}{\theta}) + (\alpha +1) \ln{\zeta}\right)}{2 \theta^2}
\bigg(1 - \frac{i \im{b_1}}{\zeta}\bigg)\frac{1}{r(\zeta)}
\frac{d\zeta}{2\pi i} 
	\\\nonumber
& = \frac{-\theta +(\alpha +1) \ln(-\frac{i}{\theta})}{\alpha +1} \mathcal{I}_{1}^{(2)}
- 2i\frac{c_1 + c_2}{2} \frac{\alpha +1}{2 \theta^2}
\bigg\{\int_{\gamma_{b_2b_1}} \frac{\zeta\ln{\zeta}}{r(\zeta)} \frac{d\zeta}{2\pi i} 
- i \im{b_1} \int_{\gamma_{b_2b_1}} \frac{\ln{\zeta}}{r(\zeta)}\frac{d\zeta}{2\pi i} \bigg\}, 
	\\\nonumber
\mathcal{I}_{1}^{(3)} & = - 2i\frac{c_1 + c_2}{2} \int_{\gamma_{b_2b_1}}  
\frac{i \left(3 (1+\alpha)^2 -2 \theta^2 \right) }{24 \theta^2}
\bigg(1 - \frac{i \im{b_1}}{\zeta}\bigg)\frac{1}{r(\zeta)}
\frac{d\zeta}{2\pi i} 
	\\ \nonumber
& = (c_1 + c_2) \frac{3 (1+\alpha)^2-2 \theta^2}{24 \theta^2}
\bigg\{\int_{\gamma_{b_2b_1}} \frac{1}{r(\zeta)}
\frac{d\zeta}{2\pi i}
- i \im{b_1}  \int_{\gamma_{b_2b_1}} \frac{1}{\zeta r(\zeta)}
\frac{d\zeta}{2\pi i}\bigg\},
	\\\nonumber
I_{1}^{(3)} = &\; (c_{1}+c_{2}) \bigg\{ \frac{1}{24\theta^{2}}\Big( 6(1+\alpha)\theta \ln(2\pi) + (3(1+\alpha)^{2}-2\theta^{2})\ln(-\tfrac{i}{\theta}) \Big) 
	\\\nonumber
& +\zeta^{\prime}(-1)-\ln \Gamma \Big( \frac{1+\alpha}{2\theta} \Big) - \ln G \Big( \frac{1+\alpha}{2\theta} \Big) \bigg\} \int_{\gamma_{b_{2}b_{1}}} \left( 1 - i \frac{\im b_{2}}{\zeta}\right) \frac{1}{r(\zeta)}\frac{d\zeta}{2\pi i}  
	\\ \label{first expression for I1(3)} 
& +(c_{1}+c_{2}) \frac{3(1+\alpha)^{2}-2\theta^{2}}{24\theta^{2}}\int_{\gamma_{b_{2}b_{1}}}\left(  1 - i \frac{\im b_{2}}{\zeta} \right) \frac{\ln \zeta}{r(\zeta)} \frac{d\zeta}{2\pi i}.
\end{align} 

It only remains to show that the coefficients in (\ref{first expression for I1(3)}) can be expressed as in (\ref{cal1coeffs}). This requires the evaluation of several integrals; we have collected the necessary results in the next lemma. 

\begin{lemma}\label{lemma: integrals for I1}
Let $\alpha > -1$ and $\theta \in (0,1]$. Let $r(\zeta)$ denote the square root defined in (\ref{rdef}). Then the following identities hold:
\begin{subequations}
\begin{align}
2\int_{\gamma_{b_2b_1}} \frac{1}{r(\zeta)} \frac{d\zeta}{2\pi i}  
= & \; 1, \label{int1overr} 
	\\
2\int_{\gamma_{b_2b_1}} \frac{\zeta}{r(\zeta)} \frac{d\zeta}{2\pi i} 
= & \; \frac{b_1 + b_2}{2}, \label{intzetaoverr} 
	\\
2\int_{\gamma_{b_2b_1}} \frac{\zeta^2}{r(\zeta)} \frac{d\zeta}{2\pi i} 
= & \; \frac{3 b_1^2+2 b_1 b_2+3   b_2^2}{8}, \label{intzeta2overr} 
	\\
2 \int_{\gamma_{b_{2}b_{1}}} \frac{1}{\zeta r(\zeta)}\frac{d\zeta}{2\pi i} = & - \frac{i}{|b_{2}|}, \label{int1overzetar} 
	\\
2\int_{\gamma_{b_2b_1}} \frac{\ln{\zeta}}{r(\zeta)} \frac{d\zeta}{2\pi i} 
= & \; \ln(i(|b_2| + \im{b_2})) - \ln{2}, \label{logzetaoverr} 
	\\
2\int_{\gamma_{b_2b_1}} \frac{\zeta \ln{\zeta}}{r(\zeta)} \frac{d\zeta}{2\pi i} 
= & -i |b_2| + i\Big(1 + \ln(2i(|b_2| + \im b_2)) - \ln 4\Big)\im{b_2}, \label{zetalogzetaoverr} 
	\\\nonumber
2\int_{\gamma_{b_2b_1}} \frac{\zeta^2 \ln{\zeta}}{r(\zeta)} \frac{d\zeta}{2\pi i} 
= & \; \frac{1}{4} \bigg(2 \left((\re{b_1})^2-2 (\im{b_2})^2\right) \ln\big(2 i (|b_2|+\im{b_2})\big)+6 |b_2| \im(b_2) 
	\\\label{zeta2logzetaoverr}
& +(\im{b_2})^2 (8\ln(2)-6)+(\re{b_1})^2 (1-4 \ln{2})\bigg),  \\
 2\int_{\gamma_{b_2b_1}} \frac{\ln{\zeta}}{\zeta r(\zeta)} \frac{d\zeta}{2\pi i} 
= &\; \frac{\ln(\frac{2i|b_2|^2}{|b_2| + \im{b_2}})}{i|b_2|}. \label{lnzetaoverzetarintegral}
\end{align}
\end{subequations}
\end{lemma}
\begin{proof} 
See Appendix \ref{I1lemmaapp}.
\end{proof}

Since $\theta \in (0,1]$, we have $\arg b_{2} = \pi - \arg b_{1} \in [0, \pi/2)$, and hence
\begin{subequations}\label{lnb1b2}
\begin{align}
& \ln(ib_{1}) = \ln(b_{1}) - \frac{3\pi i}{2}, & & \ln(-ib_{1}) = \ln(b_{1}) - \frac{i \pi}{2}, \\
& \ln(ib_{2}) = \ln(b_{2}) + \frac{\pi i}{2}, & & \ln(-ib_{2}) = \ln(b_{2}) - \frac{i \pi}{2}.
\end{align}
\end{subequations}
Substituting the expressions of Lemma \ref{lemma: integrals for I1} into (\ref{first expression for I1(3)}) and using \eqref{coeff rho a b}, \eqref{def of c1...c8}, \eqref{def of Re b2}, and (\ref{lnb1b2}) to simplify, we arrive at the expressions (\ref{cal1coeffs}) for the coefficients $\mathcal{I}_{1}^{(1)}$, $I_{1}^{(1)}$, $\mathcal{I}_{1}^{(2)}$, $I_{1}^{(2)}$, $\mathcal{I}_{1}^{(3)}$, $I_{1}^{(3)}$. 
This completes the proof of Proposition \ref{I1prop}.
\end{proof}

\section{Asymptotics of $I_{2}$}\label{I2sec}
The large $s$ asymptotics of $I_{2}$ is a consequence of the following three propositions whose proofs are given in Sections \ref{I2ZZXsubsec}, \ref{Xsubsec}, and \ref{Zsubsec}, respectively.  

\begin{proposition}[Splitting of $I_2$]\label{I2ZZXprop}
The function $I_{2}$ defined in (\ref{I2def}) can be written as 
$$I_2 = X + Z,$$
where $X$ and $Z$ are defined by
\begin{align}\label{def of Z and X}
X = \int_{\Sigma_5} H'(\zeta) \ln \mathcal{G}(\zeta)\frac{d\zeta}{2\pi i}, \qquad
Z =  -2 \int_{\gamma_{b_2b_1}} H'(\zeta) p(\zeta)\frac{d\zeta}{2\pi i}.
\end{align}
\end{proposition}

\begin{proposition}[Large $s$ asymptotics of $X$]\label{Xprop}
Let $\alpha > -1$. The quantity $X$ defined in (\ref{def of Z and X}) admits the following asymptotic expansion as $s \to +\infty$:
\begin{align}
X =&\; \mathscr{X}_{1}^{(2)} s^{\rho} \big( \ln(s^{\rho})\big)^{2}  
+\big(\mathcal{X}_{1}^{(2)}+\mathcal{X}_{3}^{(2)} \big) s^{\rho}\ln(s^{\rho}) 
+ \big(X_{1}^{(2)}+X_{3}^{(2)}\big) s^{\rho} \nonumber 
	\\
&  + \big(\mathcal{X}_{1}^{(3)} + \mathcal{X}_{3}^{(3)} \big) \ln(s^{\rho})
 + X_{1}^{(3)} + X_{2}^{(3)} + X_{3}^{(3)} + \bigO\big( s^{-\rho}\ln(s^{\rho})\big)
\end{align}
uniformly for $\theta$ in compact subsets of $(0,1]$, where the coefficients are given by
\begin{align}\nonumber
\mathscr{X}_{1}^{(2)} & = \frac{(\alpha +1) (\theta -1) (b_1-b_2)}{4 \pi  \theta^3},
	 \\\nonumber
\mathcal{X}_{1}^{(2)} & = \frac{\alpha \theta \Big(  b_1  \ln(i b_1)-b_2  \ln(i b_2)\Big)+((\alpha +2) \theta -2 (\alpha +1)) \left(b_1 \ln \left(-\frac{i b_1}{\theta }\right)-b_2 \ln\left(-\frac{i b_2}{\theta }\right)\right)}{4 \pi  \theta^3},
	\\ \nonumber
X_{1}^{(2)} & = \frac{1}{4 \pi \theta^{3}}\bigg\{b_{1} \ln(-\frac{ib_{1}}{\theta})\Big( \alpha\theta \ln(ib_{1})-(1+\alpha-\theta)\ln(-\frac{ib_{1}}{\theta}) \Big) 
	\\ \nonumber
& \hspace{1.5cm} -b_{2} \ln(-\frac{ib_{2}}{\theta})\Big( \alpha\theta \ln(ib_{2})-(1+\alpha-\theta)\ln(-\frac{ib_{2}}{\theta}) \Big) \bigg\},
	\\ \nonumber
\mathcal{X}_{1}^{(3)} & = \frac{i (\alpha +1) \left(\alpha  \theta \Big(  \ln(i b_1)-\ln(i b_2)\Big)+((\alpha +2) \theta -2
   (\alpha +1)) \left(\ln \left(-\frac{i b_1}{\theta }\right)-\ln \left(-\frac{i
   b_2}{\theta }\right)\right)\right)}{8 \pi 
   \theta^3},
   	\\ \nonumber
X_{1}^{(3)} & = \frac{i}{48 \pi \theta^{3}} \bigg\{  6 \alpha \theta (1+\alpha - \theta)\big( \ln(ib_{1})-\ln(ib_{2}) \big)
 	\\ \nonumber
& \hspace{1.7cm} - \Big( 9(1+\alpha)^{2} + 8 \theta^{2} + \theta(3\alpha (\alpha-6)-19) \Big) \Big( \ln \Big( \tfrac{-ib_{1}}{\theta} \Big)-\ln \Big( \tfrac{-ib_{2}}{\theta} \Big) \Big) 
	\\ \nonumber
& \hspace{1.7cm} + 6 \alpha \theta (1+\alpha)\Big( \ln(ib_{1})\ln\Big( \tfrac{-ib_{1}}{\theta} \Big) - \ln(ib_{2})\ln\Big( \tfrac{-ib_{2}}{\theta} \Big) \Big) 
	\\ \label{X1asympcoeffs}
& \hspace{1.7cm} - 6 (1+\alpha)(1+\alpha-\theta) \Big( \Big( \ln \big( \tfrac{-ib_{1}}{\theta} \big)\Big)^{2}
-\Big(\ln \big( \tfrac{-ib_{2}}{\theta} \big)\Big)^{2} \Big) \bigg\},
\end{align}
\begin{align}\nonumber
X_{2}^{(3)} = &\; \frac{\alpha}{2\theta} + \sum_{k=1}^{\infty}  \bigg\{ -k\psi \left( 1+\alpha+k \theta \right) + k\ln \left( \frac{1+\alpha}{2}+ k \theta \right) +\frac{\alpha}{2\theta} 
	\\ \label{X2p3p}
&\hspace{1.8cm} +  \frac{1-6 \alpha - 9 \alpha^{2}}{24 \theta^{2}} \ln \left( 1+\frac{1}{k} \right) \bigg\},
\end{align}
and
\begin{align}\nonumber
\mathcal{X}_{3}^{(2)} = & -  \frac{i(\alpha +1) (\theta -1) }{2 \theta^3} \frac{b_1 - b_2}{2\pi i}, 
	\\\nonumber
X_{3}^{(2)} = & - \frac{i(\alpha +1) (\theta -1) }{2 \theta^3} 
\frac{b_2 - b_1 + b_1\ln(-\frac{i b_1}{\theta}) - b_2 \ln(-\frac{ib_2}{\theta})}{2\pi i}, 
	\\\nonumber
\mathcal{X}_{3}^{(3)} = & \; \frac{\theta\left(9 \alpha^2+6 \alpha -1\right)  -3 (\alpha +1)^2+2 \theta^2 }{24 \theta^3} \frac{\ln(b_1/b_2)}{2 \pi i}, 
	\\ \nonumber
X_{3}^{(3)} = & \; 
\frac{\theta\left(9 \alpha^2+6 \alpha -1\right)  -3 (\alpha +1)^2+2 \theta^2}{24 \theta^3}
\frac{\ln(\frac{b_1}{b_2}) \ln(-\frac{b_1b_2}{\theta^2})}{4\pi i} 
	\\ \label{X3asympcoeffs}
& \; +
\frac{6 (\alpha +1) (\theta -1) (\alpha -\theta +1)}{24 \theta^3} \frac{\ln(b_1/b_2)}{2 \pi i}.
\end{align}

\end{proposition}

\begin{proposition}[Large $s$ asymptotics of $Z$]\label{Zprop}
Let $\alpha > -1$. The quantity $Z$ defined in (\ref{def of Z and X}) admits the following asymptotic expansion as $s \to +\infty$:
\begin{align}
Z = \mathscr{Z}^{(2)} s^\rho (\ln(s^\rho))^2  + \mathcal{Z}^{(2)} s^\rho \ln(s^\rho)
  + Z^{(2)} s^\rho  + \mathcal{Z}^{(3)} \ln(s^\rho) + Z^{(3)} + \bigO(s^{-\rho} \ln(s^\rho))\label{asymp for Z}
\end{align}
uniformly for $\theta$ in compact subsets of $(0,1]$, where the coefficients $\mathscr{Z}^{(2)}$, $\mathcal{Z}^{(2)}$, $Z^{(2)}$, $\mathcal{Z}^{(3)}$, $Z^{(3)}$ are given by
\begin{align*}
\mathscr{Z}^{(2)} = & -\frac{c_4}{2\pi \theta^2 \rho} (b_1 - b_2),
	\\
\mathcal{Z}^{(2)} =&\;  \frac{1}{2 \pi  \theta^2}\bigg\{\pi  |b_2| (c_5-c_6)+\frac{b_2 c_4 \ln\left(-\frac{i b_2}{\theta }\right)-b_1 c_4 \ln \left(-\frac{ib_1}{\theta }\right)}{\rho } +(b_1-b_2)(c_5+c_6-c_7)
	\\
& +\frac{i \pi}{2} \Big(b_1 (3c_5+c_6)+b_2 (c_5-c_6)\Big) -b_1 \ln(b_1)(c_5+c_6)
	\\
& +b_2 \ln(b_2) (c_5+c_6)+\pi \im(b_2) (c_5+c_6)\bigg\},
\end{align*}
\begin{align*}
Z^{(2)} = &\; \frac{1}{4\pi \theta^{2}} \bigg\{ -2 (b_{1}-b_{2})\Big( 2(c_{5}+c_{6})-c_{7} \Big)- i \pi \Big( (-4 i |b_{2}| + 3b_{1}+b_{2})c_{5} + (b_{1}-b_{2})c_{6} \Big) \\
& + b_{1} \ln(b_{1}) \Big( 4c_{5}+4c_{6}-2c_{7}+ \pi i (3c_{5}+c_{6}) \Big) - 2 b_{1} (\ln{b_{1}})^2(c_{5}+c_{6}) \\
& +b_{2} \ln (b_{2}) \Big( -4c_{5}-4c_{6}+2c_{7} + \pi i (c_{5}-c_{6})+2(c_{5}+c_{6})\ln (b_{2}) \Big) \\
& +\pi \Big( 2 |b_{2}|(c_{5}-c_{6})-i(b_{1}+b_{2})(c_{5}+c_{6}) \Big) \ln \left( \frac{i(|b_{2}|+\im(b_{2}))}{2} \right) \\
& + 2 \left( 1+\ln(\tfrac{-i}{\theta})\right) \bigg[ (b_{1}-b_{2})(c_{5}+c_{6}-c_{7}) + |b_{2}| (c_{5}-c_{6}) \pi \\
& + \frac{\pi i}{2}\Big( b_{2}(c_{5}-c_{6})+b_{1}(3c_{5}+c_{6})  \Big) + \pi \im (b_{2}) (c_{5}+c_{6}) \\
& - b_{1} \ln (b_{1})(c_{5}+c_{6}) + b_{2} \ln(b_{2})(c_{5}+c_{6}) \bigg] \bigg\},
	\\
\mathcal{Z}^{(3)} =&\; \frac{2 \pi  \big(c_8-\frac{3\alpha^2-1}{12}\big) \rho  (\im(b_2)-|b_2|)-i |b_2| (\ln(b_1)-\ln(b_2)) (\alpha  c_4+c_4-2 c_8 \rho )}{4 \pi  |b_2| \theta^2 \rho },
\end{align*}
and
\begin{align*}
Z^{(3)} =&\; \frac{1}{8 \pi  |b_2| \theta^2}\bigg\{4 \pi  \bigg(c_8-\frac{3\alpha^2-1}{12}\bigg) \bigg[-|b_2| \ln \left(\frac{2 i |b_2|^2}{|b_2|+\im(b_2)}\right)
	\\
& +\im(b_2) \ln \left(\frac{1}{2} i (|b_2|+\im(b_2))\right)+|b_2|-\im(b_2)\bigg]
	\\
& +i (\alpha +1) |b_2| \bigg[-4 i \pi  \ln(|b_2|) (c_5-c_6)+4 i \pi  c_6 \ln \left(\frac{2}{|b_2|+\im(b_2)}\right)
	\\
& +\ln(b_1) \Big(-2 c_7+i \pi  (3 c_5+c_6)\Big)-(\ln{b_1})^2 (c_5+c_6) \\ & +\ln(b_2) \Big(\ln(b_2) (c_5+c_6)
 +i \pi  (c_5-c_6)+2 c_7\Big)+2 \pi ^2 c_5\bigg]
	\\
& +4 i \left(1+\ln\Big(-\frac{i}{\theta }\Big)\right) \bigg(|b_2| c_8 (\ln(b_1)-\ln(b_2))+i \pi  \bigg(c_8-\frac{3\alpha^2-1}{12}\bigg) (|b_2|-\im(b_2))\bigg)
	\\
& +2 i |b_2| c_8 \left((\ln{b_1})^2-(\ln{b_2})^2\right)\bigg\}.
\end{align*}
\end{proposition}

\subsection{Proof of Proposition \ref{I2ZZXprop}}\label{I2ZZXsubsec}
Recall that the integral $I_{2}$ is given by
\begin{align*}
I_{2} & = \frac{1}{2} \int_{\sigma} H(\zeta) \tr \Big[ P^{\infty}(\zeta)^{-1}P^{\infty}(\zeta)' \sigma_{3} \Big] \frac{d\zeta}{2\pi i},
\end{align*}
where $H$ is given by \eqref{def of H} and $\sigma$ is a closed curve surrounding $\Sigma_{5}$ once in the positive direction which does not surround any of the poles $\zeta_{j}$ of $H$. A straightforward computation gives
\begin{align*}
\tr\left[ P^{\infty}(\zeta)^{-1}P^{\infty \prime}(\zeta)\sigma_3 \right]
= &\; \tr\bigg[ e^{-p(\zeta) \sigma_3} Q^\infty(\zeta)^{-1} e^{p_0\sigma_3}  
\Big(e^{-p_0\sigma_3} Q^{\infty \prime}(\zeta) e^{p(\zeta) \sigma_3}
	\\
& + e^{-p_0\sigma_3} Q^\infty(\zeta) e^{p(\zeta) \sigma_3} p'(\zeta) \sigma_3\Big)\sigma_3 \bigg]
	\\
= &\; \tr\bigg[Q^\infty(\zeta)^{-1} Q^{\infty \prime}(\zeta) \sigma_3 \bigg]
+ \tr\big[p'(\zeta) I\big] = 2 p'(\zeta),
\end{align*}
where we have used that
$$Q^\infty(\zeta)^{-1} Q^{\infty \prime}(\zeta) \sigma_3 = i\frac{\gamma'(\zeta)}{\gamma(\zeta)} \sigma_1$$
has trace zero in the last step. Hence
\begin{align*}
& I_{2} = \int_\sigma H(\zeta) p'(\zeta)\frac{d\zeta}{2\pi i}.
\end{align*}
The function $p(\zeta)$ is analytic for $\zeta \in \C \setminus \Sigma_5$ and satisfies the following jump condition across $\Sigma_5$:
$$p_+(\zeta) + p_-(\zeta) = -\ln \mathcal{G}(\zeta), \qquad \zeta \in \Sigma_5.$$
Integrating by parts, deforming the contour, and using the jump condition for $p$, we find
\begin{align*}
 I_{2} 
 & = -\int_\sigma H'(\zeta) p(\zeta)\frac{d\zeta}{2\pi i}
=  \int_{\Sigma_5} H'(\zeta) (p_+(\zeta) - p_-(\zeta))\frac{d\zeta}{2\pi i} \nonumber
	\\
& =  \int_{\Sigma_5} H'(\zeta) (2p_+(\zeta) + \ln \mathcal{G}(\zeta))\frac{d\zeta}{2\pi i}
\nonumber	\\
& =  -2 \int_{\gamma_{b_2b_1}} H'(\zeta) p(\zeta)\frac{d\zeta}{2\pi i}
+ \int_{\Sigma_5} H'(\zeta) \ln \mathcal{G}(\zeta)\frac{d\zeta}{2\pi i} 
= Z + X,
\end{align*}
which completes the proof.

\subsection{Proof of Proposition \ref{Xprop}}\label{Xsubsec}
An integration by parts gives
\begin{align*}
X = &\; \frac{H(\zeta) \ln \mathcal{G}(\zeta)}{2\pi i}\bigg|_{\zeta = b_1}^{b_2}
- \int_{\Sigma_5} H(\zeta) \frac{\mathcal{G}'(\zeta)}{\mathcal{G}(\zeta)} \frac{d\zeta}{2\pi i}.
\end{align*}
From the expression \eqref{def of ln mathcal G} for $\ln\mathcal{G}$, we have
\begin{align*}
\frac{\mathcal{G}'(\zeta)}{\mathcal{G}(\zeta)} = & -is^\rho\bigg\{c_1 +c_2+c_3 + c_1\ln(i\zeta) + c_2\ln(-i\zeta) + \ln(s) 
	\\
& - \psi\bigg(\frac{1+\alpha}{2} + is^\rho \zeta\bigg) - \frac{1}{\theta}\psi\bigg(\frac{\frac{1+\alpha}{2} - is^\rho \zeta}{\theta}\bigg) \bigg\}
	\\
 = &\; i s^\rho\bigg\{\psi\bigg(\frac{1+\alpha}{2} + is^\rho \zeta\bigg) - \ln(i \zeta s^\rho)
+ \frac{\psi(\frac{\frac{1+\alpha}{2} - is^\rho \zeta}{\theta}) - \ln(-\frac{i \zeta s^\rho}{\theta})}{\theta}\bigg\},
\end{align*}
where $c_{1}$, $c_{2}$, and $c_{3}$ are given by \eqref{def of c1...c8}. Thus,
\begin{align*}
X =&\; \frac{H(\zeta) \ln \mathcal{G}(\zeta)}{2\pi i}\bigg|_{\zeta = b_1}^{b_2}
- i s^\rho \int_{\Sigma_5} \bigg\{ \frac{1}{\theta^{2}} \left( \frac{1+\alpha}{2}- i s^{\rho}\zeta \right) \psi \bigg(\frac{\frac{1+\alpha}{2}-is^{\rho}\zeta}{\theta} \bigg)\bigg\}
	\\
&\times \bigg\{\psi\bigg(\frac{1+\alpha}{2} + is^\rho \zeta\bigg) - \ln(i \zeta s^\rho)
+ \frac{\psi(\frac{\frac{1+\alpha}{2} - is^\rho \zeta}{\theta}) - \ln(-\frac{i \zeta s^\rho}{\theta})}{\theta}\bigg\}\frac{d\zeta}{2\pi i}
	\\
= &\; \frac{H(\zeta) \ln \mathcal{G}(\zeta)}{2\pi i}\bigg|_{\zeta = b_1}^{b_2}
- \int_{i s^\rho \Sigma_5} \bigg\{ \frac{1}{\theta^{2}} \left( \frac{1+\alpha}{2}- w \right) \psi \bigg(\frac{\frac{1+\alpha}{2}- w}{\theta} \bigg)\bigg\}
	\\
&\times \bigg\{\psi\bigg(\frac{1+\alpha}{2} + w\bigg) - \ln(w)
+ \frac{\psi(\frac{\frac{1+\alpha}{2} - w}{\theta}) - \ln(-\frac{w}{\theta})}{\theta}\bigg\}\frac{dw}{2\pi i},
\end{align*}
where we have changed variables to $w = is^\rho \zeta$ in the second step. The function $\psi(\frac{\frac{1+\alpha}{2} - w}{\theta})$ has poles at the points $w = \frac{1+\alpha}{2} + j\theta \in (0,\infty)$, $j = 0,1,2, \dots$, and the function $\psi(\frac{1+\alpha}{2} + w)$ has poles at the points $w = -\frac{1+\alpha}{2} - j \in (-\infty,0)$, $j = 0,1,2, \dots$. Thus the term which will cause the most difficulties in the analysis is the one involving the product $\psi(\frac{\frac{1+\alpha}{2} - w}{\theta})\psi(\frac{1+\alpha}{2} + w)$ (for the other terms we can deform the contour into the left half-plane and use the large $z$ asymptotics of $\psi(z)$).

Let $m = \frac{1+\alpha}{2}$; the exact value of $m$ is not essential as long as $m > 0$. We split $X$ as follows:
\begin{align}\label{X1X2X3}
X = &\; X_1 + X_2 + X_3,
\end{align}
where
\begin{align}\label{X1def}
X_1 = &\; \frac{H(\zeta) \ln \mathcal{G}(\zeta)}{2\pi i}\bigg|_{\zeta = b_1}^{b_2},
	\\ \label{X2def}
X_2 = & - \int_{i s^\rho \Sigma_5} \frac{1}{\theta^{2}} \left( \frac{1+\alpha}{2}- w \right) \psi \bigg(\frac{\frac{1+\alpha}{2}- w}{\theta} \bigg) \bigg\{\psi\bigg(\frac{1+\alpha}{2} + w\bigg) - \ln(w) +\hat{g}(w) \bigg\}\frac{dw}{2\pi i},
	\\ \label{X3def}
X_3 = & - \int_{i s^\rho \Sigma_5} \frac{1}{\theta^{2}} \left( \frac{1+\alpha}{2}- w \right) \psi \bigg(\frac{\frac{1+\alpha}{2}- w}{\theta} \bigg) \bigg\{ \frac{\psi(\frac{\frac{1+\alpha}{2} - w}{\theta}) - \ln(-\frac{w}{\theta})}{\theta}
-\hat{g}(w)\bigg\}\frac{dw}{2\pi i},
\end{align}
and the function $\hat{g}(w)$ is defined by
\begin{equation}\label{def of g hat}
\hat{g}(w) = - \frac{\alpha}{2} \left( \frac{1}{w-m} - \frac{m}{(w-m)^{2}} \right) - \frac{1-3\alpha^{2}}{24(w-m)^{2}}.
\end{equation}
Note that all the $s$-dependence of $X_{2}$ and $X_{3}$ is in the contour. The term $\hat{g}(w)$ has been added and subtracted so that the integrand in the definition of $X_2$  is $\bigO(w^{-2} \ln w)$ as $w \to \infty$. This can be verified by using the asymptotic expansion (see \cite[Eq. 5.11.2]{NIST})
\begin{align}\label{digammaasymmptotics}
\psi(z) \sim \ln z - \frac{1}{2z} - \sum_{k=1}^\infty \frac{B_{2k}}{2k z^{2k}}, \qquad z \to \infty, \ |\arg z| \leq \pi - \delta,
\end{align}
where $B_{2k}$ is the $2k$th Bernoulli number, which implies that
\begin{align}\nonumber
\psi\bigg(\frac{1+\alpha}{2} + w\bigg) = &\; \ln\bigg(\frac{1+\alpha}{2} + w\bigg) 
- \frac{1}{2(\frac{1+\alpha}{2} + w)} - \frac{1}{12(\frac{1+\alpha}{2} + w)^2} + \bigO(w^{-4})
	\\ \label{psipluswasymptotics}
= &\; \ln{w} + \frac{\alpha}{2w} + \frac{1 - 3\alpha^2}{24 w^2} + \bigO(w^{-3})
\end{align}
as $w \to \infty$ away from the negative real axis, as well as the expansion
\begin{align}\label{asymp for g hat}
\hat{g}(w) = - \frac{\alpha}{2w} - \frac{1-3\alpha^{2}}{24w^{2}} + \bigO(w^{-3}) \qquad \mbox{as } w \to \infty.
\end{align}
The integrals defining $X_2$ and $X_3$ converge because the function $\hat{g}(w)$ is analytic except for a double pole at $w = m > 0$.

We will show that $X_1, X_2, X_3$ satisfy the large $s$ asymptotics
\begin{align}\label{asymp for X1}
& X_{1} = \mathscr{X}_{1}^{(2)} s^\rho (\ln(s^\rho))^2 + \mathcal{X}_{1}^{(2)} s^\rho \ln(s^\rho) + X_{1}^{(2)} s^\rho + \mathcal{X}_{1}^{(3)} \ln(s^\rho) + X_{1}^{(3)} + \bigO\big(s^{-\rho}\ln(s^{\rho})\big),
	\\ \label{asymp for X2}
& X_{2} = X_{2}^{(3)} + \bigO(s^{-\rho} \ln (s^{\rho})),
	\\ \label{asymp for X3}
& X_3 = \mathcal{X}_{3}^{(2)} s^{\rho} \ln(s^{\rho}) + X_{3}^{(2)} s^{\rho} + \mathcal{X}_{3}^{(3)} \ln (s^{\rho}) + X_{3}^{(3)} + \bigO\big(s^{-\rho}\ln (s^{\rho})\big), 
\end{align}
uniformly for $\theta$ in compact subsets of $(0,1]$, where the coefficients of the three expansions are given by (\ref{X1asympcoeffs}), (\ref{X2p3p}), and (\ref{X3asympcoeffs}), respectively. This will complete the proof of Proposition \ref{Xprop}.

\subsubsection{Asymptotics of $X_{1}$}
For $\zeta \in \C$ bounded away from $i \mathbb{R}$, we have the expansions (see (\ref{full expansion G}), (\ref{def of H}), and (\ref{digammaasymmptotics}))
$$\ln \mathcal{G}(\zeta) = \frac{c_4}{\rho} \ln(s^\rho) + c_5 \ln(i\zeta) + c_6 \ln(-i\zeta) + c_7 + \frac{c_8}{i\zeta s^\rho} + \bigO(s^{-2\rho}), \qquad s \to +\infty,$$
and
\begin{align}\label{bigFasymptotics}
H(\zeta) = -\frac{i \zeta \ln \left(-\frac{i \zeta s^\rho}{\theta}\right)}{\theta^2} s^\rho
+ \frac{\alpha -\theta +(\alpha +1) \ln \left(-\frac{i \zeta  s^\rho}{\theta}\right)+1}{2 \theta^2} + \bigO(s^{-\rho}), \qquad s \to +\infty,
\end{align}
where the constants $c_j$ are given by \eqref{def of c1...c8}. Substituting these expansions into the definition (\ref{X1def}) of $X_1$, we find (\ref{asymp for X1}).

\subsubsection{Asymptotics of $X_{2}$}
Since the integrand in the definition of $X_2$  is $\bigO(w^{-2} \ln w)$ as $w \to \infty$, we see that $X_2$ satisfies (\ref{asymp for X2}) with
\begin{align}\label{X23expression}
X_{2}^{(3)} = - \int_{-i \infty}^{i \infty} \left\{ \frac{1}{\theta^{2}}\left( \frac{1+\alpha}{2} - w \right) \psi \left( \frac{\frac{1+\alpha}{2}-w}{\theta} \right) \right\} \left\{ \psi \left( \frac{1+\alpha}{2}+w \right) - \ln w +\hat{g}(w)  \right\} \frac{dw}{2\pi i},
\end{align}
where the contour crosses the real line at $0$ and $\hat{g}(w)$ is given by \eqref{def of g hat}. It remains to show that $X_{2}^{(3)}$ can be written as in (\ref{X2p3p}). 

Since $\psi(z) = \frac{-1}{z+k} + \bigO\big( (z+k)^{-2} \big)$ as $z \to -k$ for each $k = 0,1,2, \dots$, it follows that 
\begin{equation*}
\psi \left( \frac{\frac{1+\alpha}{2}-w}{\theta} \right)
\end{equation*}
has a simple pole with residue $\theta$ at each of the points $\frac{1+\alpha}{2}+ k \theta$, $k = 0,1,2, \dots$. For $k \geq 1$, the associated residue gives a contribution to $X_{2}^{(3)}$ equal to (taking into account that after the deformation, the loop is going in the clockwise orientation around $\frac{1+\alpha}{2}+ k \theta$)
\begin{align*}
-k \left\{ \psi \left( 1+\alpha+k \theta \right) - \ln \left( \frac{1+\alpha}{2}+ k \theta \right) + \hat{g}\left(\frac{1+\alpha}{2}+ k \theta\right)  \right\}.
\end{align*}
On the other hand, the residue at $m = \frac{1+\alpha}{2}$ is given by
\begin{equation*}
-\frac{(1-6\alpha - 9 \alpha^{2}) \gamma_{\mathrm{E}} - 12 \alpha \theta}{24 \theta^{2}},
\end{equation*}
where $\gamma_{\mathrm{E}}$ is Euler's gamma constant.
By (\ref{psipluswasymptotics}) and (\ref{asymp for g hat}), we have
$$\psi \left( \frac{1+\alpha}{2}+w \right) - \ln w +\hat{g}(w) = \bigO(w^{-3})$$
as $w \to \infty$ away from the negative real axis. Moreover, by (\ref{digammaasymmptotics}), 
\begin{align}\label{psiminuswestimate}
\psi\bigg(\frac{\frac{1+\alpha}{2} - w}{\theta}\bigg) = &\; \ln\bigg(-\frac{w}{\theta}\bigg) + O(w^{-1}),
\end{align}	
as $w \to \infty$ away from the positive real axis; in fact, combining (\ref{digammaasymmptotics}) with the reflection formula $\psi(1-z) = \psi(z) + \pi \cot(\pi z)$, we see that  (\ref{psiminuswestimate}) holds also as $w \to \infty$ in a sector containing the positive real axis as long as $w$ stays away from the poles $\{\frac{1+\alpha}{2} + j\theta\}_{j=0}^\infty$.
Thus, deforming the contour in (\ref{X23expression}) to infinity in the right half-plane along curves which stay away from the set $\{\frac{1+\alpha}{2} + j\theta\}_{j=0}^\infty$, the contribution from infinity vanishes and we find
\begin{align*}
X_{2}^{(3)} = &  -\frac{(1-6\alpha - 9 \alpha^{2}) \gamma_{\mathrm{E}} - 12 \alpha \theta}{24 \theta^{2}} \\
& + \sum_{k=1}^{\infty} k \left\{ -\psi \left( 1+\alpha+k \theta \right) + \ln \left( \frac{1+\alpha}{2}+ k \theta \right) -\hat{g}\left(\frac{1+\alpha}{2}+ k \theta\right)  \right\},
\end{align*}
where the series is convergent because it originates from a convergent integral. Using the series representation for the Euler gamma constant (see \cite[Eq. 5.2.3]{NIST})
\begin{equation*}
\gamma_{\mathrm{E}} = \sum_{k=1}^{\infty} \left\{ \frac{1}{k}- \ln \left( 1+\frac{1}{k} \right) \right\}
\end{equation*}
together with the fact that
\begin{align*}
-k\hat{g}\left(\frac{1+\alpha}{2}+ k \theta\right) =  \frac{1-6 \alpha - 9 \alpha^{2}}{24 k \theta^{2}} + \frac{\alpha}{2 \theta},
\end{align*}
we conclude that $X_{2}^{(3)}$ can be written as in (\ref{X2p3p}). This proves (\ref{asymp for X2}).

\subsubsection{Asymptotics of $X_{3}$}
Deforming the contour $is^{\rho}\Sigma_{5}$ in the definition (\ref{X3def}) of $X_3$ in the left half-plane to the contour $is^{\rho}\gamma_{b_{2}b_{1}}$, we get
\begin{align}\label{lol10}
X_3 = & \int_{i s^\rho \gamma_{b_2b_1}} \frac{1}{\theta^{2}} \left( \frac{1+\alpha}{2}- w \right) \psi \bigg(\frac{\frac{1+\alpha}{2}- w}{\theta} \bigg) \bigg\{ \frac{\psi(\frac{\frac{1+\alpha}{2} - w}{\theta}) - \ln(-\frac{w}{\theta})}{\theta}
-\hat{g}(w)\bigg\}\frac{dw}{2\pi i}.
\end{align}
The above representation is convenient for the asymptotic analysis of $X_3$, because the argument $(\frac{1+\alpha}{2}- w)/\theta$ of $\psi$ is large as $s \to + \infty$ uniformly for $w \in i s^\rho \gamma_{b_2b_1}$. As $w \to \infty$ away from the positive real axis, we have
\begin{align}\label{psialphawthetaasymptotics}
\psi\bigg(\frac{\frac{1+\alpha}{2} - w}{\theta}\bigg) = &\; \ln(-\frac{w}{\theta}) + \frac{\theta-1-\alpha}{2w} + \frac{6\theta - 2\theta^2 - 3- 3\alpha^2 + 6\alpha(\theta -1)}{24 w^2} + \bigO(w^{-3}).
\end{align}	
Substitution of the expansions (\ref{psialphawthetaasymptotics}) and \eqref{asymp for g hat} into \eqref{lol10} yields
\begin{align*}
X_3 = & \int_{i s^\rho \gamma_{b_2b_1}} \bigg\{-\frac{(\alpha +1) (\theta -1) \ln \left(-\frac{w}{\theta }\right)}{2 \theta^3}
	\\
& +\frac{6 (\alpha +1) (\theta -1) (\alpha -\theta +1)+\left(\left(9 \alpha^2+6 \alpha -1\right) \theta -3 (\alpha +1)^2+2 \theta^2\right) \ln \left(-\frac{w}{\theta}\right)}{24 \theta^3 w}
	\\
& + \bigO(w^{-2} \ln(-w))\bigg\}\frac{dw}{2\pi i}.
\end{align*}
Letting $w = i s^\rho \zeta$, we obtain
\begin{align}
X_3 = & \int_{\gamma_{b_2b_1}} \bigg\{-\frac{(\alpha +1) (\theta -1) \ln \left(-\frac{ i s^\rho \zeta}{\theta }\right)}{2 \theta^3} \nonumber
	\\
& +\frac{6 (\alpha +1) (\theta -1) (\alpha -\theta +1)+\left(\left(9 \alpha^2+6 \alpha -1\right) \theta -3 (\alpha +1)^2+2 \theta^2\right) \ln \left(-\frac{ i s^\rho \zeta}{\theta}\right)}{24 \theta^3  i s^\rho \zeta}
 \bigg\}\frac{i s^\rho d\zeta}{2\pi i} \nonumber
 	\\
& + \bigO(s^{-\rho} \ln(s^\rho)) \nonumber
	\\ \label{asymp for X3m}
= & \; \mathcal{X}_{3}^{(2)} s^{\rho} \ln(s^{\rho}) + X_{3}^{(2)} s^{\rho} + \mathcal{X}_{3}^{(3)} \ln (s^{\rho}) + X_{3}^{(3)} + \bigO\big(s^{-\rho}\ln (s^{\rho})\big) 
\end{align}
uniformly for $\theta$ in compact subsets of $(0,1]$, where the coefficients are given by
\begin{align*}
\mathcal{X}_{3}^{(2)} = & -  \frac{i(\alpha +1) (\theta -1) }{2 \theta^3} \int_{\gamma_{b_2b_1}}\frac{d\zeta}{2\pi i}, \\
X_{3}^{(2)} = & - \frac{i(\alpha +1) (\theta -1) }{2 \theta^3} \int_{\gamma_{b_2b_1}}  \ln(-\tfrac{i \zeta}{\theta }) \frac{d\zeta}{2\pi i}, \\
\mathcal{X}_{3}^{(3)} = & \; \frac{ \left(9 \alpha^2+6 \alpha -1\right) \theta -3 (\alpha +1)^2+2 \theta^2 }{24 \theta^3} \int_{\gamma_{b_2b_1}} \frac{1}{\zeta} \frac{d\zeta}{2\pi i}, \\
X_{3}^{(3)} = & \; 
\frac{\left(9 \alpha^2+6 \alpha -1\right) \theta -3 (\alpha +1)^2+2 \theta^2}{24 \theta^3}
\int_{\gamma_{b_2b_1}}\frac{\ln(-\frac{i \zeta}{\theta })}{\zeta}
\frac{d\zeta}{2\pi i}.
\end{align*}
Using that
\begin{align*}
& \int_{\gamma_{b_2b_1}}\frac{d\zeta}{2\pi i}  = \frac{b_1 - b_2}{2\pi i}, \qquad
\int_{\gamma_{b_2b_1}}  \ln(-\frac{i \zeta}{\theta })\frac{d\zeta}{2\pi i} 
= \frac{b_2 - b_1 + b_1\ln(-\frac{i b_1}{\theta}) - b_2 \ln(-\frac{ib_2}{\theta})}{2\pi i}
	\\
& \int_{\gamma_{b_2b_1}} \frac{1}{\zeta} \frac{d\zeta}{2\pi i} 
= \frac{\ln(b_1/b_2)}{2 \pi i} = \frac{\pi - 2\arg(b_2)}{2\pi},
\qquad
\int_{\gamma_{b_2b_1}}\frac{\ln(-\frac{i \zeta}{\theta })}{\zeta}
\frac{d\zeta}{2\pi i} = \frac{\ln(\frac{b_1}{b_2}) \ln(-\frac{b_1b_2}{\theta^2})}{4\pi i}, \qquad
\end{align*}
we see that the coefficients $\mathcal{X}_{3}^{(2)}, X_{3}^{(2)}, \mathcal{X}_{3}^{(3)}, X_{3}^{(3)}$ can be written as in (\ref{X3asympcoeffs}). This proves (\ref{asymp for X3}) and thus completes the proof of Proposition \ref{Xprop}.

\subsection{Proof of Proposition \ref{Zprop}}\label{Zsubsec}
We have
\begin{align*}
H'(\zeta) = -\frac{i}{\theta^2}s^\rho \ln(s^\rho) -\frac{i(1 + \ln(-\frac{i\zeta}{\theta}))}{\theta^2}s^\rho 
  + \frac{1+\alpha}{2\zeta \theta^2} + \bigO(s^{-\rho}), \qquad s \to +\infty,
\end{align*}
uniformly for $\zeta \in \gamma_{b_2b_1}$.
Moreover, by Proposition \ref{pprop}, 
$$p(\zeta) = -\frac{c_4}{2\rho} \ln(s^\rho) + \frac{\mathcal{B}(\zeta)}{2} + \frac{\mathcal{A}(\zeta)}{s^\rho} + \bigO(s^{-2\rho}), \qquad s \to +\infty,$$
uniformly for $\zeta \in \gamma_{b_2b_1}$, where the coefficients $\mathcal{B}(\zeta)$ and $\mathcal{A}(\zeta)$ are defined by
\begin{align}\label{def of mathcal A}
\mathcal{A}(\zeta) = & \; \frac{\mathcal{A}_1(\zeta)}{\zeta}= \frac{ic_8}{2\zeta}+\frac{c_8-\frac{3\alpha^2-1}{12}}{2|b_2|} \frac{r(\zeta)}{\zeta}, 
	  \\\label{def of mathcal B}
\mathcal{B}(\zeta) = & \; \mathcal{R}(\zeta) - c_{7} - c_{5} \ln(i\zeta) - c_{6} \ln (-i\zeta), 
\end{align}
with $\mathcal{A}_1$ and $\mathcal{R}$ given by \eqref{def of mathcalA1} and \eqref{def of mathcalR}. Substitution into the definition (\ref{def of Z and X}) of $Z$ shows that $Z$ admits an expansion of the form (\ref{asymp for Z}) as $s \to + \infty$, uniformly for $\theta$ in compact subsets of $(0,1]$, with coefficients given by
\begin{align}\nonumber
\mathscr{Z}^{(2)} = & -\frac{c_4}{2\pi \theta^2 \rho} (b_1 - b_2) = \frac{(\alpha +1) (\theta -1) \theta^{\frac{2}{\theta +1}-\frac{5}{2}}}{\pi  (\theta+1) \rho},
	\\\nonumber
\mathcal{Z}^{(2)} =&  \int_{\gamma_{b_2b_1}} \frac{-c_4(1 + \ln(-\frac{i\zeta}{\theta }))+\rho \mathcal{B}(\zeta)}{2 \pi \theta^2 \rho } d\zeta
	\\\nonumber
= & - \frac{c_4 }{2 \pi \theta^2 \rho} \bigg(b_1\ln(-\frac{i b_1}{\theta }) - b_2\ln(-\frac{ib_2}{\theta })\bigg)
+  \frac{1}{2 \pi \theta^2} \int_{\gamma_{b_2b_1}} \mathcal{B}(\zeta) d\zeta,
   	\\\nonumber
Z^{(2)} = &\; 
\int_{\gamma_{b_2b_1}} \frac{1+\ln(-\frac{i \zeta}{\theta})}{2 \pi  \theta^2} \mathcal{B}(\zeta)d\zeta
	\\\nonumber
= &\; \frac{1+\ln(-\frac{i}{\theta})}{2 \pi  \theta^2} \int_{\gamma_{b_2b_1}}  \mathcal{B}(\zeta)d\zeta
+ \frac{1}{2 \pi  \theta^2} \int_{\gamma_{b_2b_1}} \ln(\zeta) \mathcal{B}(\zeta)d\zeta,
	\\\nonumber
\mathcal{Z}^{(3)} =&\; \int_{\gamma_{b_2b_1}} \bigg\{ \frac{\mathcal{A}(\zeta)}{\pi  \theta^2}-\frac{i (\alpha +1) c_4}{4 \pi  \zeta  \theta^2 \rho}\bigg\} d\zeta
	\\\nonumber
= &\; \frac{1}{\pi  \theta^2}  \int_{\gamma_{b_2b_1}} \mathcal{A}(\zeta)d\zeta
- \frac{i (\alpha +1) c_4}{4 \pi  \theta^2 \rho} (\ln{b_1} - \ln{b_2}),
   	\\\nonumber
Z^{(3)} =&\; \int_{\gamma_{b_2b_1}} \bigg\{ \frac{1+\ln{\zeta} + \ln(-\frac{i}{\theta})}{\pi \theta^2}\mathcal{A}(\zeta)  
   + \frac{i (\alpha +1) \mathcal{B}(\zeta)}{4 \pi  \zeta  \theta^2}\bigg\}d\zeta
   	\\ \label{Pjdef}
=&\;  \frac{1+ \ln(-\frac{i}{\theta})}{\pi \theta^2} \int_{\gamma_{b_2b_1}}\mathcal{A}(\zeta)  d\zeta
+  \frac{1}{\pi \theta^2} \int_{\gamma_{b_2b_1}} \ln({\zeta}) \mathcal{A}(\zeta)  d\zeta
   + \frac{i (\alpha +1)}{4 \pi \theta^2}  \int_{\gamma_{b_2b_1}} \frac{\mathcal{B}(\zeta)}{\zeta} d\zeta.
\end{align}

It only remains to show that the coefficients in (\ref{Pjdef}) can be expressed as in the statement of Proposition \ref{Zprop}. Inspection of (\ref{Pjdef}) shows that there are five different integrals that need to evaluated:
$$\int_{\gamma_{b_2b_1}} \mathcal{B}(\zeta) d\zeta, \quad
\int_{\gamma_{b_2b_1}} \ln(\zeta) \mathcal{B}(\zeta)d\zeta, \quad
\int_{\gamma_{b_2b_1}} \mathcal{A}(\zeta) d\zeta, \quad
\int_{\gamma_{b_2b_1}} \ln (\zeta) \mathcal{A}(\zeta) d\zeta, \quad 
\int_{\gamma_{b_2b_1}} \frac{\mathcal{B}(\zeta)}{\zeta} d\zeta.$$
These integrals are evaluated in the following lemma.

\begin{lemma}\label{Zintlemma}
For $\alpha > -1$ and $\theta \in (0,1]$, it holds that
\begin{subequations}
\begin{align}\nonumber
\int_{\gamma_{b_2b_1}} \mathcal{B}(\zeta) d\zeta 
= &\; \pi  |b_2| (c_5-c_6)+(b_1-b_2) (c_5+c_6-c_7)
	\\\nonumber
&  + \frac{ \pi i}{2}   \Big(b_1 (3 c_5+c_6)+b_2 (c_5-c_6)\Big)-b_1 \ln(b_1) (c_5+c_6)
	\\ \label{calBintegral}
& +b_2 \ln(b_2) (c_5+c_6)+\pi \im(b_2) (c_5+c_6), 
	\\ \nonumber
\int_{\gamma_{b_2b_1}} \ln (\zeta) \mathcal{B}(\zeta) d\zeta 
= &\; \frac{1}{2} \bigg\{\pi  \ln\left(\frac{1}{2} i (|b_2|+\im(b_2))\right) \Big(2 |b_2| (c_5-c_6)-i (b_1+b_2) (c_5+c_6)\Big)
	\\\nonumber
& -i \pi  \Big(c_6 (b_1-b_2)+c_5 (-4 i |b_2|+3 b_1+b_2)\Big)
	\\\nonumber
& -2 (b_1-b_2) (2 (c_5+c_6)-c_7)
	\\\nonumber
& +b_1 \ln (b_1) \Big(i \pi  (3 c_5+c_6)+4 c_5+4 c_6-2 c_7\Big) -2 b_1 (\ln{b_1})^2 (c_5+c_6) 
	\\ \label{lnzetacalBintegral}
& +b_2 \ln(b_2) \Big(2 \ln (b_2) (c_5+c_6)+i \pi  (c_5-c_6)-4 c_5-4 c_6+2 c_7\Big)\bigg\},
	\\\label{calAintegral}
 \int_{\gamma_{b_2b_1}} \mathcal{A}(\zeta) d\zeta 
= & -\frac{c_8(\arg b_1 - \arg b_2)}{2} - \frac{\pi \big(c_8-\frac{3\alpha^2-1}{12}\big)}{2}\bigg(1 - \frac{\im{b_2}}{|b_2|}\bigg),
	\\ \nonumber
 \int_{\gamma_{b_2b_1}} \ln (\zeta) \mathcal{A}(\zeta) d\zeta 
= &\; \frac{ic_8}{4}((\ln{b_1})^2 - (\ln{b_2})^2) 
+ \frac{\pi \big(c_8-\frac{3\alpha^2-1}{12}\big)}{2|b_2|}\bigg\{|b_2| - \im{b_2} - |b_2| \ln\bigg(\frac{2i|b_2|^2}{|b_2| + \im{b_2}}\bigg)
 	\\ \label{lnzetacalAintegral}
& + (\im{b_2})\ln\bigg(\frac{i(|b_2| + \im{b_2})}{2}\bigg)\bigg\},
	\\\nonumber
\int_{\gamma_{b_2b_1}} \frac{\mathcal{B}(\zeta)}{\zeta} d\zeta 
= &\; \frac{1}{2} \bigg\{-4 i \pi  \ln(|b_2|) (c_5-c_6)+4 i \pi  c_6 \ln \left(\frac{2}{|b_2|+\im(b_2)}\right)
	\\\nonumber
& +\ln(b_1) (-2 c_7+i \pi  (3 c_5+c_6))-(c_5+c_6)(\ln{b_1})^2 
	\\\label{calBoverzetaintegral}
& +\ln(b_2) (2 c_7+i \pi  (c_5-c_6))+ (c_5+c_6)(\ln{b_2})^2 +2 \pi ^2 c_5\bigg\}.
\end{align}
\end{subequations}
\end{lemma}
\begin{proof}
See Appendix \ref{Zintlemmaapp}.
\end{proof}

Substituting the results of Lemma \ref{Zintlemma} into (\ref{Pjdef}), we obtain after simplification the expressions for the coefficients
$\mathscr{Z}^{(2)}$, $\mathcal{Z}^{(2)}$, $Z^{(2)}$, $\mathcal{Z}^{(3)}$, and $Z^{(3)}$ given in the statement of Proposition \ref{Zprop}. This completes the proof of Proposition \ref{Zprop}.

\section{Asymptotics of $I_{3,K}$ and $I_{4,K}$}\label{I3I4sec}
In this section, we prove two propositions (Proposition \ref{I3Kprop} and Proposition \ref{I4Kprop}) which establish the large $s$ asymptotics of $I_{3,K}$ and $I_{4,K}$, respectively, where we henceforth choose $K = s^\rho$.

\begin{proposition}[Large $s$ asymptotics of $I_{3,K}$]\label{I3Kprop}
Let $\alpha > -1$ and let $K = s^\rho$. As $s \to +\infty$, the function $I_{3,K}$ defined in (\ref{I3Kdef}) satisfies
\begin{align}\label{asymp for I3}
I_{3,K} = \mathcal{I}_{3}^{(3)} \ln (s^{\rho}) + I_{3}^{(3)} + \bigO(s^{-\rho}\ln (s^{\rho}))
\end{align}
uniformly for $\theta$ in compact subsets of $(0,1]$, where the coefficients $\mathcal{I}_{3}^{(3)}$ and $I_{3}^{(3)}$ are given by
\begin{align*}
\mathcal{I}_{3}^{(3)} = & - \frac{3\alpha(1+\alpha-\theta)+\theta}{12\theta^{2}(\theta+1)}, 
	\\
I_{3}^{(3)} = &\; \frac{3\alpha(1+\alpha-\theta)+\theta}{6\theta(\theta+1)^{2}} \ln ( \theta ) - \frac{3 + 3\alpha(4+3\alpha)-4\theta - 3\alpha \theta (4+\alpha) + \theta^{2}}{24 \theta^{2}(1+\theta)}.
\end{align*}
\end{proposition}
\begin{proof}
By the cyclicity of the trace, we can write the definition (\ref{I3Kdef}) of $I_{3,K}$ as
\begin{align}\label{I3K2}
I_{3,K} & = \frac{1}{2} \int_{\sigma_{K}} H(\zeta) \tr\Big[R^{-1}(\zeta)R'(\zeta)e^{p_{0}\sigma_{3}}P^{\infty}(\zeta)\sigma_{3}P^{\infty}(\zeta)^{-1}e^{-p_{0}\sigma_{3}}\Big] \frac{d\zeta}{2\pi i},
\end{align}
where $K = s^\rho$, and $\sigma_{K}$ and $H$ are defined in Lemma \ref{lemma: simplified diff dientity}. 
All the $s$-dependence of the trace lies in the factor $R^{-1}(\zeta)R'(\zeta)$, since by \eqref{def of Pinf}, the quantity
\begin{align}\nonumber
e^{p_{0}\sigma_{3}}P^{\infty}(\zeta)\sigma_{3}P^{\infty}(\zeta)^{-1}e^{-p_{0}\sigma_{3}} 
= &\; Q^{\infty}(\zeta) \sigma_{3} Q^{\infty}(\zeta)^{-1}
	\\\label{lol17}
=&\; \frac{1}{r(\zeta)}
\begin{pmatrix} \zeta -i\im(b_2)  & i\re{b_2} \\
i\re{b_2} & i\im(b_2) - \zeta \end{pmatrix}
\end{align}
is independent of $s$.
As $s \to + \infty$, we have by Proposition \ref{Rprop} that
\begin{align}\label{Rexpansionsec9}
& R(\zeta)= I+ \frac{R^{(1)}(\zeta)}{s^\rho} + \bigO\bigg(\frac{1}{s^{2\rho}(1+ |\zeta|)}\bigg)
\end{align}
uniformly for $\zeta \in \mathbb{C}\setminus \Gamma_{R}$ and that this expansion can be differentiated with respect to $\zeta$. 
The asymptotics in (\ref{Rexpansionsec9}) as well as all other asymptotic expansions in the rest of this section are uniform with respect to $\theta$ in compact subsets of $(0,1]$.

From the explicit expression (\ref{explicit R1(zeta)}) for $R^{(1)}$, we see that $R^{(1)}(\zeta)$ and $R^{(1)\prime}(\zeta)$ are $\bigO((1 +|\zeta|)^{-1})$ and $\bigO((1 +|\zeta|)^{-2})$, respectively, uniformly for $\zeta \in \sigma_K$ as $s \to +\infty$.
Therefore, 
\begin{align*}
R^{-1}(\zeta)R'(\zeta)
= &\; \bigg(I - \frac{R^{(1)}(\zeta)}{s^\rho} + \bigO\bigg(\frac{1}{s^{2\rho}(1+ |\zeta|)}\bigg)\bigg)\bigg(\frac{R^{(1)\prime}(\zeta)}{s^\rho} + \bigO\bigg(\frac{1}{s^{2\rho}(1+ |\zeta|)^2}\bigg)\bigg)
	\\
= &\; \frac{R^{(1)\prime}(\zeta)}{s^\rho} + \widetilde{R}_{R}(\zeta), \qquad \mbox{ as } s \to + \infty,
\end{align*}
uniformly for $\zeta \in \sigma_{K}$, where 
\begin{align*}
& \widetilde{R}_{R}(\zeta) = \bigO\bigg(\frac{1}{s^{2\rho}(1+ |\zeta|)^2}\bigg) \qquad \mbox{as } s \to + \infty,
\end{align*}
uniformly for $\zeta \in \sigma_{K}$.
Hence, for large $s$ and $\zeta \in \sigma_{K}$, we have
\begin{align}\label{TrWTr}
\tr\Big[R^{-1}(\zeta)R'(\zeta)e^{p_{0}\sigma_{3}}P^{\infty}(\zeta)\sigma_{3}P^{\infty}(\zeta)^{-1}e^{-p_{0}\sigma_{3}}\Big]
=
\frac{W(\zeta)}{s^\rho} + \tr \Big[ \widetilde{R}_{R}(\zeta) Q^{\infty}(\zeta) \sigma_{3} Q^{\infty}(\zeta)^{-1} \Big],
\end{align} 
where the function $W(\zeta)$ is defined by
\begin{align}\label{def of W}
W(\zeta) = \tr \bigg[R^{(1)\prime}(\zeta)Q^{\infty}(\zeta) \sigma_{3} Q^{\infty}(\zeta)^{-1}\bigg].
\end{align}
Substituting (\ref{TrWTr}) into \eqref{I3K2}, we see that
\begin{align}\label{I3Ksplitting}
I_{3,K} = I_{3} + \widetilde{I}_{3,K},
\end{align}
where $I_{3}$ and $\widetilde{I}_{3,K}$ are defined by
\begin{align*}
I_{3} = &\; \frac{1}{2s^{\rho}}\int_{\sigma} H(\zeta) W(\zeta) \frac{d\zeta}{2\pi i}, \\
\widetilde{I}_{3,K} = &\; \frac{1}{2}\int_{\sigma_{K}}H(\zeta) \tr \Big[ \widetilde{R}_{R}(\zeta) Q^{\infty}(\zeta) \sigma_{3} Q^{\infty}(\zeta)^{-1} \Big] \frac{d\zeta}{2\pi i}.
\end{align*}

We first estimate $\widetilde{I}_{3,K}$. Since
\begin{align}\nonumber
& H(\zeta) = \bigO\left(\frac{1}{s^{\rho}(\zeta-\zeta_{j})}\right)\qquad \mbox{as } \zeta \to \zeta_{j}, \ j = 0, 1, \dots, 
	\\\label{lol20}
& H(\zeta) = \bigO\big(\zeta s^{\rho} \ln (\zeta s^{\rho}) \big) \qquad \mbox{as } \zeta s^{\rho} \to \infty, 
\end{align}
we have (see also Figure \ref{fig: sigmaK})
\begin{align}\nonumber
|\widetilde{I}_{3,K}| = & \; \bigg| \frac{1}{2}\int_{\sigma_{K}\cap \{|\zeta| \leq c' s^{-\rho}\}}H(\zeta) \tr \Big[ \widetilde{R}_{R}(\zeta) Q^{\infty}(\zeta) \sigma_{3} Q^{\infty}(\zeta)^{-1} \Big] \frac{d\zeta}{2\pi i} 
	\\\nonumber
& + \frac{1}{2}\int_{\sigma_{K}\cap \{|\zeta| > c' s^{-\rho}\}}H(\zeta) \tr \Big[ \widetilde{R}_{R}(\zeta) Q^{\infty}(\zeta) \sigma_{3} Q^{\infty}(\zeta)^{-1} \Big] \frac{d\zeta}{2\pi i} \bigg| 
	\\ \nonumber
\leq & \; C' \int_{\sigma_{K}\cap \{|\zeta| \leq c' s^{-\rho}\}} \frac{|d\zeta|}{s^{3\rho}|\zeta-\zeta_0|} + C' \int_{\sigma_{K}\cap \{|\zeta| > c' s^{-\rho}\}} \frac{\ln|\zeta s^{\rho}|}{|\zeta| s^{\rho}}|d\zeta|
	\\\label{tildeI3Kestimate}
\leq &\; \frac{C'}{s^{3\rho}} + C' \frac{\ln s}{s^{\rho}} = \bigO\left(\frac{\ln s}{s^{\rho}}\right)	
\end{align}
where $c',C'>0$ are two sufficiently large constants. 

We next consider $I_{3}$. 
Substituting \eqref{explicit R1(zeta)} and \eqref{lol17} into \eqref{def of W}, it follows that
\begin{align}
W(\zeta) = &\; \frac{1}{r(\zeta)} \tr \bigg[
\bigg(-\frac{A}{(\zeta-b_1)^2} - \frac{2B}{(\zeta-b_1)^3} + \frac{\bar{A}}{(\zeta-b_2)^2} - \frac{2 \bar{B}}{(\zeta-b_2)^3}\bigg) \label{W explicit}
	\\
& \times  \begin{pmatrix} \zeta -i\im(b_2)  & i\re{b_2} \\
i\re{b_2} & i\im(b_2) - \zeta \end{pmatrix}\bigg]. \nonumber
\end{align}
Replacing $r$ by $\tilde{r}$ in $W$ does not change the value of $I_{3}$. Deforming $\sigma$ (which surrounds $0$) into another contour $\tilde{\sigma}$ which surrounds the cut $[b_{1},b_{2}]$ of $\tilde{r}$ once in the positive direction but which does not surround $0$, it transpires that
\begin{align*}
I_{3} = \frac{1}{2s^{\rho}}\int_{\tilde{\sigma}} H(\zeta) \widetilde{W}(\zeta) \frac{d\zeta}{2\pi i},
\end{align*}
where $\widetilde{W}$ is defined by the expression obtained by replacing $r$ by $\tilde{r}$ in the right-hand side of \eqref{W explicit}. Assuming that $\tilde{\sigma}$ is bounded away from $0$, we can replace $H(\zeta)$ by its large $s$ asymptotics (\ref{bigFasymptotics}); this gives
\begin{align*}
I_{3} = & - \frac{1}{2} \int_{\tilde{\sigma}} \frac{i \zeta \ln \left(-\frac{i \zeta s^\rho}{\theta}\right)}{\theta^2} \widetilde{W}(\zeta)\frac{d\zeta}{2\pi i} +  \bigO(s^{-\rho}\ln(s^\rho)).
\end{align*}
We split the leading term as follows:
\begin{align}\label{lol18}
I_{3} = I_{3,1} + I_{3,2} +  \bigO\big(s^{-\rho}\ln(s^\rho)\big),
\end{align}
where $I_{3,1}$ and $I_{3,2}$ are given by
$$I_{3,1} = - \frac{1}{2} \int_{\tilde{\sigma}} \frac{i \zeta \ln(-\frac{i s^\rho}{\theta})}{\theta^2} \widetilde{W}(\zeta)\frac{d\zeta}{2\pi i}, \qquad
I_{3,2} = - \frac{1}{2} \int_{\tilde{\sigma}} \frac{i \zeta \ln(\zeta)}{\theta^2} \widetilde{W}(\zeta)\frac{d\zeta}{2\pi i}.$$
From \eqref{W explicit}, we have the expansion
\begin{align*}
\widetilde{W}(\zeta) = -\frac{1}{\zeta^2} \tr[(A - \bar{A})\sigma_3] + \bigO(\zeta^{-3}) \qquad \mbox{ as } \zeta \to \infty.
\end{align*}
By deforming the contour $\tilde{\sigma}$ to infinity, we get
\begin{align}\label{I31explicit}
I_{3,1} = \frac{1}{2} \frac{i \ln(-\frac{i s^\rho}{\theta})}{\theta^2}  \tr[(A - \bar{A})\sigma_3]
= -\frac{(3 \alpha  (1+\alpha -\theta )+\theta ) \ln(-\frac{i s^\rho}{\theta })}{12 \theta^2 (\theta +1)},
\end{align}
while
\begin{align*}
I_{3,2} = & - \frac{1}{2} \lim_{R\to \infty}\bigg\{
\int_{C_{R}} \frac{i \zeta \ln(\zeta)}{\theta^2} \widetilde{W}(\zeta)\frac{d\zeta}{2\pi i}
+ \int_{-R}^0 \frac{i \zeta}{\theta^2} \widetilde{W}(\zeta) d\zeta\bigg\}
	\\
= & - \frac{1}{2} \lim_{R\to \infty}\bigg\{-\frac{i }{\theta^2} \tr[(A - \bar{A})\sigma_3]
\int_{C_{R}}  \frac{\ln(\zeta)}{\zeta} \frac{d\zeta}{2\pi i}
+ \int_{-R}^0 \frac{i \zeta}{\theta^2} W(\zeta) d\zeta\bigg\}
\end{align*}
We have
$$\int_{C_{R}}  \frac{\ln(\zeta)}{\zeta} \frac{d\zeta}{2\pi i} = \ln{R}$$
and the integral $\int_{-R}^0 \frac{i \zeta}{\theta^2} W(\zeta) d\zeta$ can be computed explicitly using \eqref{W explicit}. After simplification this gives
\begin{equation*}
I_{3,2} = - \frac{3 + 3\alpha(4+3\alpha)-4\theta - 3\alpha \theta (4+\alpha) + \theta^{2}}{24 \theta^{2}(1+\theta)} - \frac{3\alpha(1+\alpha - \theta)+\theta}{12 \theta^{2}(1+\theta)}\ln\Big( i \theta^{\frac{1-\theta}{1+\theta}} \Big).
\end{equation*}
Substituting this expression for $I_{3,2}$ into \eqref{lol18} and recalling (\ref{I3Ksplitting}), (\ref{tildeI3Kestimate}), and (\ref{I31explicit}), 
equation (\ref{asymp for I3}) follows.
\end{proof}

\begin{proposition}[Large $s$ asymptotics of $I_{4,K}$]\label{I4Kprop}
Let $\alpha > -1$ and let $K = s^\rho$. As $s \to +\infty$, the function $I_{4,K}$ defined in (\ref{I4Kdef}) satisfies, for any $N \geq 1$,
\begin{align}\label{asymp for I4K}
I_{4,K} = \bigO\big(s^{-N\rho}\big)
\end{align}
uniformly for $\theta$ in compact subsets of $(0,1]$.
\end{proposition}
\begin{proof}
In view of \eqref{lol17}, we can write
\begin{align*}
I_{4,K} = -\frac{1}{2} \int_{\widetilde{\Sigma}_{K}} H(\zeta) \tr \Big[\Big(R_{+}^{-1}(\zeta)R_{+}'(\zeta)-R_{-}^{-1}(\zeta)R_{-}'(\zeta)\Big)Q^{\infty}(\zeta) \sigma_{3} Q^{\infty}(\zeta)^{-1}\Big],
\end{align*}
where $Q^{\infty}(\zeta) \sigma_{3} Q^{\infty}(\zeta)^{-1}$ is independent of $s$ and $\bigO(1)$ as $\zeta \to \infty$.
Using (\ref{RplusRminusestimate}) and \eqref{lol20}, we conclude that, for any $N$ large enough,
$$|I_{4,K}| = \bigO\bigg(\int_{s^\rho}^\infty \frac{|\zeta| s^{\rho} \ln(|\zeta| s^{\rho})}{s^{N\rho}(1+ |\zeta|)^N} d|\zeta|\bigg) = \bigO\big( s^{-N\rho}\big)$$
uniformly for $\theta$ in compact subsets of $(0,1]$. This proves (\ref{asymp for I4K}).
\end{proof}

\section{Proof of Theorem \ref{thm:main results}}\label{section: integration in theta}
Substituting the large $s$ asymptotics of the integrals $I_{1}$, $I_{2}$, $I_{3,K}$, and $I_{4,K}$ established in Sections \ref{I1sec}-\ref{I3I4sec} (see Propositions \ref{I1prop}, \ref{I2ZZXprop}, \ref{Xprop}, \ref{Zprop}, \ref{I3Kprop}, and \ref{I4Kprop}) into the differential identity \eqref{diff identity simplified}, we obtain
\begin{align}\nonumber
\partial_{\theta} \ln \det \Big( \left. 1-\mathbb{K} \right|_{[0,s]} \Big) = & \; \mathcal{I}_{1}^{(1)} s^{2\rho}\ln(s^{\rho}) + I_{1}^{(1)}s^{2\rho} + (\mathscr{X}_{1}^{(2)}+\mathscr{Z}^{(2)})s^{\rho} \big( \ln(s^{\rho})\big)^{2} 
	 \\\nonumber
& \; +\big( \mathcal{I}_{1}^{(2)} + \mathcal{X}_{1}^{(2)}+\mathcal{X}_{3}^{(2)}+\mathcal{Z}^{(2)} \big) s^{\rho}\ln(s^{\rho}) + \big( I_{1}^{(2)} + X_{1}^{(2)}+X_{3}^{(2)}+Z^{(2)} \big) s^{\rho}  
	\\\nonumber
& \; + \big( \mathcal{I}_{1}^{(3)} + \mathcal{X}_{1}^{(3)} + \mathcal{X}_{3}^{(3)} + \mathcal{Z}^{(3)} + \mathcal{I}_{3}^{(3)} \big) \ln(s^{\rho}) 
	 \\ \label{asymp for the diff identity final1}
& \; + I_{1}^{(3)} + X_{1}^{(3)} + X_{2}^{(3)} + X_{3}^{(3)} + Z^{(3)} + I_{3}^{(3)} + \bigO\big( s^{-\rho}\ln(s^{\rho})\big)
\end{align}
as $s \to + \infty$ uniformly for $\theta$ in compact subsets of $(0,1]$.

\subsection{Integration of the differential identity}
Since the asymptotic formula \eqref{asymp for the diff identity final1} is valid uniformly for $\theta$ in compact subsets of $(0,1]$, we can integrate \eqref{asymp for the diff identity final1} with respect to $\theta$ from $\theta = 1$ to an arbitrary $\theta \in (0,1]$. Using the known result \eqref{asymp theta=1} valid for $\theta =1$, this yields the following lemma. 

\begin{lemma}\label{int1thetalemma0}
Let $\alpha > -1$. The following expansion is valid uniformly for $\theta$ in compact subsets of $(0,1]$ as $s \to + \infty$:
\begin{align}\nonumber
\ln \det \Big( \left. 1-\mathbb{K} \right|_{[0,s]} \Big) = &  -a s^{2\rho}+b s^{\rho}+c \ln s + \ln G(1+\alpha) - \frac{\alpha}{2}\ln(2\pi) - \frac{\alpha^2}{2}\ln{2}
	\\\nonumber
& -  \int_{1}^{\theta} \ln G \left( 1+ \frac{1+\alpha}{2\theta'} \right) d\theta' 
+ \int_{1}^{\theta} W(\theta', \alpha) d\theta'
	\\ \label{asymp diff identity almost final l} 
& + \int_{1}^{\theta} X_{2}^{(3)}(\theta', \alpha)d\theta'
 + \bigO\big(s^{-\rho} \ln(s^{\rho})\big), 
\end{align}
where the coefficients $a,b,c$ are given by (\ref{coeff rho a b}) and (\ref{little c in thm}), $X_{2}^{(3)} = X_{2}^{(3)}(\theta, \alpha)$ is given in (\ref{X2p3p}), and $W(\theta, \alpha)$ is defined by
\begin{align*}
W(\theta, \alpha) = &\; \frac{-3-12 \alpha - 6 \alpha^{2} + 2 \theta + \theta^{2}}{6(1+\theta)^{2}} \ln(\theta) + \frac{1+\alpha}{4\theta}\ln(2\pi)  
	\\
& + \frac{-3(1+\alpha)^{2} + (2+3\alpha)^{2} \theta - (1+6\alpha) \theta^{2}}{24 \theta^{2} (1+\theta)} + \frac{1+6\alpha(1+\alpha)-\theta^{2}}{12\theta^{2}} \ln(1+\theta) + \zeta'(-1).
\end{align*}
\end{lemma}
\begin{proof}
The proof involves long computations which use the definitions (\ref{def of c1...c8}) and (\ref{b1b2def}) of the constants $\{c_j\}_1^8$, $b_1$, and $b_2$, as well as the relations (\ref{lnb1b2}) satisfied by $\ln{b_1}$ and $\ln{b_2}$. 
Explicit expressions for the coefficients in (\ref{asymp for the diff identity final1}) are given in Propositions \ref{I1prop}, \ref{I2ZZXprop}, \ref{Xprop}, \ref{Zprop}, \ref{I3Kprop}, \ref{I4Kprop}.
After rather lengthy calculations, we find that the first six coefficients on the right-hand side of (\ref{asymp for the diff identity final1}) can be expressed as
\begin{align*}
& \mathcal{I}_{1}^{(1)} =  - \frac{2a}{\rho(1+\theta)^2},
	\\
& I_{1}^{(1)} =  - \partial_\theta a, 	
	\\
& \mathscr{X}_{1}^{(2)}+\mathscr{Z}^{(2)}=  0, 
	\\
& \mathcal{I}_{1}^{(2)} + \mathcal{X}_{1}^{(2)}+\mathcal{X}_{3}^{(2)}+\mathcal{Z}^{(2)} =  \frac{b}{\rho(1+\theta)^2}, 
	\\
& I_{1}^{(2)} + X_{1}^{(2)}+X_{3}^{(2)}+Z^{(2)} = \partial_\theta b, 
	\\
& \mathcal{I}_{1}^{(3)} + \mathcal{X}_{1}^{(3)} + \mathcal{X}_{3}^{(3)} + \mathcal{Z}^{(3)} + \mathcal{I}_{3}^{(3)} = \frac{\partial_\theta c}{\rho},
\end{align*}
where $a$, $b$, and $c$ are given by \eqref{coeff rho a b} and \eqref{little c in thm}. 
Integrating (\ref{asymp for the diff identity final1}) from $1$ to $\theta$ and using (\ref{asymp theta=1}) to compute the boundary term at $1$, this yields 
\begin{align*}\nonumber
\ln \det \Big( \left. 1-\mathbb{K} \right|_{[0,s]} \Big) = &  -a s^{2\rho}+b s^{\rho}+c \ln s + \ln G(1+\alpha) - \frac{\alpha}{2}\ln(2\pi)
	\\ 
& + \int_{1}^{\theta} \left.\Big( I_{1}^{(3)} + X_{1}^{(3)} + X_{2}^{(3)} + X_{3}^{(3)} + Z^{(3)} + I_{3}^{(3)} \Big)\right|_{\theta'}d\theta' + \bigO\big(s^{-\rho} \ln(s^{\rho})\big)
\end{align*}
as $s \to +\infty$ uniformly for $\theta$ in compact subsets of $(0,1]$.
The lemma will follow if we can show that
$$I_{1}^{(3)} + X_{1}^{(3)} + X_{2}^{(3)} + X_{3}^{(3)} + Z^{(3)} + I_{3}^{(3)}
= - \ln G \left( 1+ \frac{1+\alpha}{2\theta} \right) 
+ W + X_{2}^{(3)}.$$
This identity is a consequence of another long computation which also employs the identities
\begin{align*}
\ln(b_{1}) & \; = \frac{1-\theta}{1+\theta} \ln(\theta) + \ln(1+\theta) + i(\pi-\phi), \\
\ln(b_{2}) & \; = \frac{1-\theta}{1+\theta} \ln(\theta) + \ln(1+\theta) + i\phi,
\end{align*}
which are a consequence of (\ref{b1b2def}) and (\ref{def of Re b2}).
\end{proof}

\begin{remark}\upshape
Lemma \ref{int1thetalemma0} provides an alternative proof of the expressions (\ref{coeff rho a b}) and (\ref{little c in thm}) for $a,b$, and $c$ based on the differential identity in $\theta$. Note that this method yields an error term in (\ref{asymp diff identity almost final l}) of order $\bigO\big(s^{-\rho} \ln(s^{\rho})\big)$, which is slightly worse than than the optimal bound $\bigO(s^{-\rho})$ (which was proved via the differential identity in $s$ in \cite{ClaeysGirSti}). 
\end{remark}

To complete the proof of Theorem \ref{thm:main results} it only remains to verify that the sum of the terms of order $1$ on the right-hand side of (\ref{asymp diff identity almost final l}) equals $\ln C$, where $C$ is given by \eqref{big C in thm}. In order to verify this we need to compute the three integrals on the right-hand side of (\ref{asymp diff identity almost final l}). These integrals are computed in the following three lemmas.

\begin{lemma}\label{int1thetalemma1}
For $\alpha > -1$ and $\theta \in (0,1]$, it holds that
\begin{align}\nonumber
& -\int_{1}^{\theta} \ln G \Big(1 + \frac{1+\alpha}{2\theta'} \Big)d\theta' =  -\theta \ln G \left( 1+ \frac{1+\alpha}{2\theta} \right) + \ln G \left( 1+\frac{1+\alpha}{2} \right) 
	\\\label{part1 of final expression for C} 
& \hspace{1cm} - \frac{1+\alpha}{2} \left( \frac{\ln(2\pi)-1}{2}\ln \theta - \frac{1+\alpha}{2} \frac{\theta-1}{\theta} - \ln \Gamma \left(1+ \frac{1+\alpha}{2\theta} \right) + \ln \Gamma \left( \frac{3+\alpha}{2} \right) \right). 
\end{align}
\end{lemma}
\begin{proof}
A simple integration by parts shows that
\begin{equation}\label{int log G}
\int_{1}^{\theta} \ln G \Big( 1 + \frac{1+\alpha}{2\theta'} \Big)d\theta' = \bigg[\theta' \ln G\Big(1+\frac{1+\alpha}{2\theta'} \Big) \bigg]_{\theta'=1}^\theta - \int_{1}^{\theta} \theta' \partial_{\theta'} \left[ \ln G\left(1+\frac{1+\alpha}{2 \theta'} \right) \right]d\theta'.
\end{equation}
Using the identities (see \cite[Eq. 5.17.4]{NIST} for the first identity)
\begin{align*}
& (\ln G)^{\prime}(z) = \frac{1}{2}(\ln(2\pi) + 1) - z + (z-1)\psi(z) \quad \mbox{and} \quad \partial_{\theta'}\left(1+ \frac{1+\alpha}{2 \theta'} \right) = - \frac{1+\alpha}{2 \theta'^{2}},
\end{align*}
we obtain
\begin{align}\nonumber
& - \int_{1}^{\theta} \theta' \partial_{\theta'} \left[ \ln G\left(1+\frac{1+\alpha}{2\theta'} \right) \right]d\theta' 
	\\ \nonumber
& \hspace{2cm} = \frac{1+\alpha}{2} \int_{1}^{\theta} \frac{1}{\theta'} \left[ \frac{\ln(2\pi)-1}{2} - \frac{1+\alpha}{2\theta'}  + \frac{1+\alpha}{2\theta'} \psi\left(1+ \frac{1+\alpha}{2\theta'} \right) \right]d\theta' 
	\\\nonumber
& \hspace{2cm} = \frac{1+\alpha}{2} \left( \frac{\ln(2\pi)-1}{2}\ln \theta - \frac{1+\alpha}{2}\frac{\theta-1}{\theta}-\int_{1}^{\theta} \partial_{\theta'}\psi \left(1+ \frac{1+\alpha}{2\theta'} \right) d\theta' \right) 
	\\\label{minusint1thetalnG}
& \hspace{2cm} = \frac{1+\alpha}{2} \left( \frac{\ln(2\pi)-1}{2}\ln \theta - \frac{1+\alpha}{2}\frac{\theta-1}{\theta}-\left[  \psi \left(1+ \frac{1+\alpha}{2\theta'} \right) \right]_{\theta' = 1}^{\theta} \right).
\end{align}
Substituting (\ref{minusint1thetalnG}) into \eqref{int log G} and simplifying, we find (\ref{part1 of final expression for C}).
\end{proof}

\begin{lemma}\label{int1thetalemma2}
For $\alpha > -1$ and $\theta \in (0,1]$, it holds that
\begin{align}
\int_{1}^{\theta} W(\theta', \alpha) d\theta' 
= &\; \zeta'(-1) (\theta-1) + \left( \frac{1+\alpha}{4}\ln(2\pi) + \frac{9+30\alpha+24\alpha^{2}-3\theta-18 \alpha \theta + 4 \theta^{2}}{24(1+\theta)} \right) \ln (\theta) \nonumber \\
& + \frac{-1-6\alpha-6\alpha^{2}+2\theta + 6 \alpha \theta - \theta^{2}}{12\theta}\ln(1+\theta)-\frac{(\theta-1)\big( 3(1+\alpha)^{2}+2\theta \big)}{24\theta} + \frac{\alpha^{2}}{2} \ln 2. \nonumber
\end{align}
\end{lemma}
\begin{proof}
This follows from a long but straightforward computation.
\end{proof}

\begin{lemma}\label{int1thetalemma3}
For $\alpha > -1$ and $\theta \in (0,1]$, it holds that
\begin{align}\nonumber
\int_{1}^{\theta} X_{2}^{(3)}(\theta', \alpha) d\theta' 
= &\; \frac{\alpha}{2}\ln \theta - d(\theta,\alpha) + d(1,\alpha) - (\theta - 1)\zeta'(-1) 
 + \theta \ln G \left( 1 + \frac{1+\alpha}{2\theta} \right) 
	\\ \nonumber
&- \ln G \left( \frac{3+\alpha}{2} \right)  
- \frac{1+\alpha}{2} \ln \Gamma \left( 1+\frac{1+\alpha}{2\theta} \right) +\frac{1+\alpha}{2}\ln \Gamma \left( 1+\frac{1+\alpha}{2} \right)
	\\ \label{part3 of final expression for C}
&  - \frac{(\theta-1)\big( 3(1+\alpha)^{2} - 2 \theta \big)}{24 \theta}. 
\end{align}
\end{lemma}
\begin{proof}
Integrating the definition (\ref{X2p3p}) of $X_{2}^{(3)}$ from $1$ to $\theta$ and appealing to Fubini's theorem to interchange the order of integration and summation, we obtain
\begin{align}\nonumber
\int_{1}^{\theta} X_{2}^{(3)}(\theta', \alpha) d\theta' = &\; \frac{\alpha}{2}\ln \theta + \lim_{N\to + \infty}\sum_{k=1}^{N} \Bigg\{ -\ln \Gamma (1+\alpha + k \theta) + \ln \Gamma(1+\alpha + k)  	
	\\\nonumber
& + k (1-\theta) + \left(\frac{1+\alpha}{2} + k \theta \right) \ln\left(\frac{1+\alpha}{2} + k \theta \right) - \left(\frac{1+\alpha}{2} + k  \right) \ln\left(\frac{1+\alpha}{2} + k  \right) 
 	\\\label{int of X2p3p version1}
 & - \frac{1-6\alpha - 9 \alpha^{2}}{24 }\frac{1-\theta}{\theta} \ln \left( 1+\frac{1}{k} \right) + \frac{\alpha}{2} \ln \theta \Bigg\}. 
\end{align}

To simplify the sum in (\ref{int of X2p3p version1}), we first consider the sum of $\ln \Gamma(1+\alpha + k)$. Using the reproducing formula for Barnes' $G$-function (see \cite[Eq. 5.17.1]{NIST}),
\begin{align*}
G(z+1) = \Gamma(z) G(z),
\end{align*}
we can write
\begin{align*}
\sum_{k=1}^{N} \ln \Gamma (1+\alpha + k) = \ln G(2+\alpha+N)-\ln G (2+\alpha).
\end{align*}
The asymptotics \eqref{large z asymp for log Barnes G} of $\ln G$ then leads to the large $N$ asymptotics
\begin{align}\nonumber
\sum_{k=1}^{N} & \ln \Gamma(1+\alpha + k) =  \frac{N^{2}}{2}\ln N - \frac{3}{4} N^{2} + (1+\alpha) N \ln N + \left( \frac{\ln(2\pi)}{2}-(1+\alpha) \right) N 
	\\ \label{sumlnGamma}
& + \frac{5+12 \alpha + 6 \alpha^{2}}{12} \ln N + \Big( \zeta'(-1) + \frac{\ln(2\pi)}{2}(1+\alpha) - \ln G(2+\alpha) \Big) + \bigO(N^{-1}).
\end{align}

To simplify the terms in (\ref{int of X2p3p version1}) involving $\ln(\frac{1+\alpha}{2} + k \theta)$, we utilize the Hurwitz zeta function $\zeta(u,z)$ which is defined for $\re u > 1$ and $z \neq 0, -1, -2, \dots$ by
$$\zeta(u,z) = \sum_{n=0}^\infty \frac{1}{(n + z)^u}.$$
We recall that this function, which generalizes Riemann's zeta function $\zeta(u)$ in the sense that $\zeta(u,1) = \zeta(u)$, is defined for all $u \in \C \setminus \{1\}$ by analytic continuation.  
A simple shift of the summation index shows that
\begin{align}\label{Hurwitz identity difference}
\zeta(u,z)-\zeta(u,z+N) 
=  \sum_{n=0}^\infty \frac{1}{(n + z)^u} -  \sum_{n=N+1}^\infty \frac{1}{(n + z)^u}
= \sum_{n=0}^{N-1} \frac{1}{(n+z)^{u}},
\end{align}
whenever $\re u > 1$. By analyticity, (\ref{Hurwitz identity difference}) is in fact valid for all $u \in \C \setminus \{1\}$ and $z \in \C \setminus \{0,-1, \dots\}$. 
Differentiating \eqref{Hurwitz identity difference} with respect to $u$ and evaluating the resulting equation at $u=-1$, we obtain
\begin{align}\label{zetaprimezetaprime}
\zeta'(-1,z+N)-\zeta'(-1,z) = \sum_{n=0}^{N-1} (n+z)\ln(n+z),
\end{align}
where $\zeta'(-1,z) := \left.\partial_{u}  \zeta(u,z) \right|_{u=-1}$. It is a simple calculation to deduce from (\ref{zetaprimezetaprime}) that
\begin{align*}
\sum_{k=1}^{N} \left(\frac{1+\alpha}{2} + k \theta \right) \ln\left(\frac{1+\alpha}{2} + k \theta \right) = & \; \theta \bigg[ \zeta'\left( -1,1+N+\frac{1+\alpha}{2\theta}\right) - \zeta'\left( -1,1+\frac{1+\alpha}{2\theta} \right) \bigg] \\
& + \frac{\ln \theta}{2} N (1+\alpha + \theta + N \theta).
\end{align*}
Using the asymptotic formula \cite[Eq. 25.11.44]{NIST}
\begin{align*}
\zeta'(-1,z) = \frac{z^{2}}{2}\ln z - \frac{z^{2}}{4} - \frac{z}{2}\ln z + \frac{1}{12} \ln z + \frac{1}{12} + \bigO(z^{-2}), \qquad z \to \infty,
\end{align*}
which is valid in the sector $|\arg z| < \frac{\pi}{2} - \delta$ for any fixed $\delta > 0$, we obtain, for any $\theta\in (0,1]$,
\begin{align}\nonumber
& \sum_{k=1}^{N}\left(\frac{1+\alpha}{2} + k \theta \right) \ln\left(\frac{1+\alpha}{2} + k \theta \right) = \frac{\theta}{2}N^{2} \ln N + \frac{\theta}{4}(2\ln \theta - 1) N^{2} + \frac{1}{2}(1+\alpha+\theta) N \ln N  
	\\ \nonumber
& + \frac{1}{2}(1+\alpha+\theta) N \ln \theta + \frac{3(1+\alpha^{2})+2\theta(3+\theta) + 6 \alpha (1+\theta)}{24 \theta} \ln N 
	 \\ \label{sumk1Nleft}
& + \frac{3(1+\alpha^{2})+2\theta(3+\theta) + 6 \alpha (1+\theta)}{24 \theta} - \theta \zeta'\left(-1;\frac{1+\alpha + 2 \theta}{2 \theta}\right) + \bigO(N^{-1})
\end{align}
as $N \to + \infty$. 

The asymptotics of the terms in (\ref{int of X2p3p version1}) involving $\ln(\frac{1+\alpha}{2} + k)$ can be obtained by setting $\theta = 1$ in (\ref{sumk1Nleft}).
Moreover, it is easy to check that
\begin{align}\label{sumk1thetaalpha2}
& \sum_{k=1}^{N} \bigg\{k(1-\theta) + \frac{\alpha}{2} \ln \theta\bigg\} = \frac{1-\theta}{2}N^{2} + \frac{1-\theta}{2} N + \frac{\alpha}{2}N \ln \theta
\end{align}
and
\begin{align}\label{sum16alpha}
- \sum_{k=1}^{N} \frac{1-6\alpha - 9 \alpha^{2}}{24 }\frac{1-\theta}{\theta} \ln \left( 1+\frac{1}{k} \right) = - \frac{1-6\alpha - 9 \alpha^{2}}{24 }\frac{1-\theta}{\theta} \ln(N+1).
\end{align}

Substituting (\ref{sumlnGamma}), (\ref{sumk1Nleft}), (\ref{sumk1thetaalpha2}), and (\ref{sum16alpha}) into \eqref{int of X2p3p version1} and using that $\ln(N+1)$ can be replaced by $\ln N$ because $\ln(N+1) - \ln{N} \to 0$ as $N \to +\infty$, we obtain
\begin{align*}
& \int_{1}^{\theta} X_{2}^{\infty}(c) d\theta = \frac{\alpha}{2}\ln \theta + \lim_{N \to + \infty} \Bigg\{ \bigg( -\sum_{k=1}^{N}  \ln \Gamma (1+\alpha + k \theta) \bigg)  + \frac{\theta}{2} N^{2} \ln N +  (2 \ln \theta - 3)\frac{\theta}{4}N^{2} \\
& +\left(1+\alpha + \frac{\theta -1}{2}\right) N \ln N + \left( \frac{\ln(2\pi)}{2}-(1+\alpha)+\frac{1-\theta}{2} + \left( \alpha + \frac{1+\theta}{2} \right) \ln \theta \right) N \\
& +\frac{1+6\alpha^{2} + \theta (3+\theta) + 6\alpha(1+\theta)}{12 \theta} \ln N + \zeta'(-1) - \theta \zeta'\left( -1;\frac{1+\alpha + 2 \theta}{2 \theta} \right) + \zeta' \left( -1;\frac{3+\alpha}{2} \right) \\
& +\frac{\ln(2\pi)}{2}(1+\alpha) - \ln G(2+\alpha) + \frac{3(1+\alpha^{2})+2\theta(3+\theta) + 6 \alpha (1+\theta)}{24 \theta} - \frac{11+12\alpha + 3 \alpha^{2}}{24} \Bigg\},
\end{align*}
which, recalling the definition \eqref{def of the constant d} of the quantity $d(\theta,\alpha)$, 
can be rewritten as
\begin{align}\nonumber
\int_{1}^{\theta} X_{2}^{(3)}(\theta', \alpha) d\theta'  = &\; \frac{\alpha}{2}\ln \theta - d(\theta,\alpha) + d(1,\alpha)  
- \theta \zeta'\left( -1;\frac{1+\alpha + 2 \theta}{2 \theta} \right) + \zeta' \left( -1;\frac{3+\alpha}{2} \right) 
	\\ \label{lol22}
&  - \frac{(\theta-1)\big( 3(1+\alpha)^{2} - 2 \theta \big)}{24 \theta}. 
\end{align}
Using the following identity which relates the Barnes $G$-function to $\zeta'(-1,z)$ (see \cite[Eq. (18)]{A1998}):
\begin{equation*}
\ln G(z) = \zeta'(-1)-\zeta'(-1,z) + (z-1) \ln \Gamma(z),
\end{equation*} 
we can rewrite \eqref{lol22} as (\ref{part3 of final expression for C}).
\end{proof}

\begin{remark}\upshape
Note that the reasoning leading to (\ref{sumlnGamma}) cannot be applied to the sum 
\begin{align*}
\sum_{k=1}^{N} \ln \Gamma(1+\alpha + k \theta)
\end{align*}
for general values of $\theta$. In fact, this is the only finite $N$ sum in \eqref{int of X2p3p version1} which we are not able to evaluate in terms of known special functions. 
\end{remark}

Replacing the three integrals on the right-hand side of \eqref{asymp diff identity almost final l} with the expressions derived in Lemmas \ref{int1thetalemma1}-\ref{int1thetalemma3}, we obtain the following expression for the term of order $1$ in the large $s$ asymptotics of $\ln \det ( \left. 1-\mathbb{K} \right|_{[0,s]} )$:
\begin{align*}
& \ln G ( 1+\alpha ) - \frac{\alpha}{2} \ln(2\pi) + d(1,\alpha) - d(\theta,\alpha) \\
& + \frac{24 \alpha (\alpha +2)+15+3\theta + 4 \theta^{2}}{24(1+\theta)} \ln \theta 
+ \frac{6\alpha \theta - 6 \alpha (1+\alpha)-(\theta-1)^{2}}{12 \theta} \ln(1+\theta),
\end{align*}
which is precisely $\ln C$, where $C$ is defined by \eqref{big C in thm}. This finishes the proof in the case when $\theta \in (0,1]$; as explained in Section \ref{twocasessubsubsec} the result for $\theta \in [1, \infty)$ then follows by symmetry. The proof of Theorem \ref{thm:main results} is therefore complete.

\appendix

\section{Proof of Proposition \ref{prop: constant d for rational theta}}\label{subsection: d rational}
In this appendix, we establish the formula \eqref{d in terms of Barnes G intro} for $d(\theta, \alpha)$ for rational values of $\theta$ stated in Proposition \ref{prop: constant d for rational theta}.

Let $\theta = p/q$ where $p,q \geq 1$ are two (not necessarily relatively prime) integers. Let $N = mq$ where $m \geq 1$ is an integer (later we will take $m \to + \infty$). We have
\begin{align}\label{sum d theta rational def}
\prod_{k=1}^{N} \Gamma(1+\alpha + k \theta) = \prod_{k=1}^{mq} \Gamma\left(1+\alpha + \frac{kp}{q}\right) = \prod_{k=1}^{q} \prod_{j=0}^{m-1} \Gamma \left( 1+\alpha + jp + k\frac{p}{q} \right).
\end{align}
We recall that $\Gamma(z)$ satisfies the duplication formula (see \cite[Eq. 5.5.6]{NIST})
\begin{align} \label{duplication Gamma}
\Gamma(pz) = p^{pz-\frac{1}{2}}(2\pi)^{\frac{1-p}{2}} \prod_{\ell = 0}^{p-1} \Gamma \left( z + \frac{\ell}{p} \right).
\end{align}
Evaluating (\ref{duplication Gamma}) at $z = \frac{1+\alpha}{p}+j + \frac{k}{q}$, we find
\begin{equation}\label{lol28}
\Gamma \left( 1+\alpha + jp + k\frac{p}{q} \right) = p^{\frac{1}{2}+\alpha + jp + k\frac{p}{q}}(2\pi)^{\frac{1-p}{2}} \prod_{\ell = 0}^{p-1} \Gamma \left( \frac{1+\alpha}{p}+j + \frac{k}{q} + \frac{\ell}{p} \right).
\end{equation}
Substituting \eqref{lol28} into \eqref{sum d theta rational def}, we obtain
\begin{align}\label{prodGammaprodprod}
\prod_{k=1}^{N} \Gamma(1+\alpha + k \theta) = \prod_{k=1}^{q} \prod_{j=0}^{m-1} p^{\frac{1}{2}+\alpha + jp + k\frac{p}{q}}(2\pi)^{\frac{1-p}{2}} \times \prod_{k=1}^{q} \prod_{\ell=0}^{p-1} \prod_{j=0}^{m-1} \Gamma \left( \frac{1+\alpha}{p}+j + \frac{k}{q} + \frac{\ell}{p} \right).
\end{align}
The last product can be expressed in terms of Barnes' $G$-function:
\begin{align*}
\prod_{j=0}^{m-1} \Gamma \left( \frac{1+\alpha}{p}+j + \frac{k}{q} + \frac{\ell}{p} \right) = \frac{G \left( \frac{1+\alpha}{p} + m + \frac{k}{q} + \frac{\ell}{p} \right)}{G \left( \frac{1+\alpha}{p} + \frac{k}{q} + \frac{\ell}{p} \right)}.
\end{align*}
On the other hand, we have
\begin{align*}
\prod_{j=0}^{m-1} p^{\frac{1}{2}+\alpha + jp + k\frac{p}{q}}(2\pi)^{\frac{1-p}{2}} = p^{m \left( \frac{1}{2} + \alpha + k \frac{p}{q} \right)}(2\pi)^{\frac{(1-p)m}{2}} p^{\frac{pm(m-1)}{2}}.
\end{align*}
Therefore, taking the logarithm, we can write equation (\ref{prodGammaprodprod}) as
\begin{align*}
\sum_{k=1}^{N} \ln \Gamma(1+\alpha + k \theta) = \sum_{k=1}^{q} \left\{\left( \frac{1}{2} + \alpha + k \frac{p}{q} \right) m \ln p + \frac{(1-p)m}{2} \ln(2\pi) + \frac{pm(m-1)}{2} \ln p \right\} \\
+ \sum_{k=1}^{q} \sum_{\ell = 0}^{p-1} \left\{ \ln G \left( \frac{1+\alpha}{p} + m + \frac{k}{q} + \frac{\ell}{p} \right) - \ln G \left( \frac{1+\alpha}{p} + \frac{k}{q} + \frac{\ell}{p} \right) \right\}
\end{align*}
As $m \to + \infty$, by definition of $d$, the term of order $1$ in the above expression is given by
\begin{align*}
&d\Big(\theta = \frac{p}{q},\alpha\Big) +\frac{1+6\alpha^{2} + \theta (3+\theta) + 6\alpha(1+\theta)}{12 \theta} \ln q = \\
& \sum_{k=1}^{q} \sum_{\ell = 0}^{p-1} \left\{ \zeta'(-1) + \left( \frac{k}{2q} + \frac{1+\alpha + \ell -p}{2p} \right) \ln(2\pi) - \ln G \left( \frac{1+\alpha}{p} + \frac{k}{q} + \frac{\ell}{p} \right) \right\},
\end{align*}
where we have used that the term of order $1$ in the large $m$ expansion of $\ln G \left( \frac{1+\alpha}{p} + m + \frac{k}{q} + \frac{\ell}{p} \right)$ is given by
\begin{equation}
\zeta'(-1) + \left( \frac{k}{2q} + \frac{1+\alpha + \ell -p}{2p} \right) \ln(2\pi).
\end{equation}
Simplifying the double sum, we arrive at the following expression for $d$:
\begin{align*}
d\Big(\theta = \frac{p}{q},\alpha\Big) & = pq\zeta'(-1) + \left(\frac{p(q+1)}{4} + \frac{(1+\alpha)q}{2} -\frac{1}{2} pq + \frac{q(p-1)}{4} \right) \ln(2\pi) \\ 
& \quad - \sum_{k=1}^{q} \sum_{\ell = 0}^{p-1} \ln G \left( \frac{1+\alpha}{p} + \frac{k}{q} + \frac{\ell}{p} \right) - \frac{1+6\alpha^{2} + \theta (3+\theta) + 6\alpha(1+\theta)}{12 \theta} \ln q \\
& = pq\zeta'(-1) + \left( \frac{(1+\alpha)q}{2} + \frac{p-q}{4} \right) \ln(2\pi) - \sum_{k=1}^{q} \sum_{\ell = 0}^{p-1} \ln G \left( \frac{1+\alpha}{p} + \frac{k}{q} + \frac{\ell}{p} \right).
\end{align*}
After some simple cancellations and a simple change of indices, we obtain (\ref{d in terms of Barnes G intro}).

\section{Proof of Proposition \ref{dsymmprop}}\label{subsection: symmetry for d}
In this appendix, we prove the symmetry relation \eqref{symmetry for d} for $d$ given in Proposition \ref{dsymmprop}. We first use (\ref{d in terms of Barnes G intro}) to prove the relation for rational values of $\theta$. We then use continuity to extend it to all $\theta > 0$.

Let $\theta= p/q$ for some $p,q \in \N \setminus \{ 0 \}$. From \eqref{d in terms of Barnes G intro}, we have
\begin{align}\nonumber
d(\theta,\alpha)-d \left( \frac{1}{\theta},\frac{1+\alpha}{\theta}-1 \right) =&\; \frac{p-q}{2} \ln(2\pi)+\frac{\alpha -\theta \alpha+1-\theta }{ \theta} \ln p 
	\\ \label{dminusd}
& + \frac{1+6\alpha^{2} + \theta (3+\theta) + 6\alpha(1+\theta)}{12 \theta} \ln \theta 
+ D(p,q, \alpha),
\end{align}
where the function $D(p,q, \alpha)$ is defined by
\begin{align*}
D(p,q, \alpha) = - \sum_{k=1}^{q} \sum_{\ell = 1}^{p} \ln G \left( \frac{\ell+\alpha}{p} + \frac{k}{q} \right) + \sum_{k=1}^{q} \sum_{\ell = 1}^{p} \ln G \left( \frac{k-1 }{q} + \frac{\ell+1+\alpha}{p} \right).
\end{align*}
Simplification gives
\begin{align*}
D(p,q, \alpha) =& - \sum_{k=1}^{q} \sum_{\ell = 1}^{p} \ln G \left( \frac{\ell+\alpha}{p} + \frac{k}{q} \right) + \sum_{k=0}^{q-1} \sum_{\ell = 2}^{p+1} \ln G \left( \frac{k}{q} + \frac{\ell+\alpha}{p} \right)
\\
=& -\sum_{\ell = 1}^{p} \ln G \left( \frac{\ell+\alpha}{p} + 1 \right) - \sum_{k=1}^{q} \ln G \left( \frac{1+\alpha}{p} + \frac{k}{q} \right) 
	\\
&+\sum_{\ell = 2}^{p+1} \ln G \left(  \frac{\ell+\alpha}{p} \right)+\sum_{k=0}^{q-1}\ln G \left( \frac{k}{q} + \frac{1+\alpha}{p}+1 \right)
\\
=&\; \sum_{k=1}^{q-1} \bigg\{  \ln G \left( \frac{k}{q} + \frac{1+\alpha}{p}+1 \right)-\ln G \left( \frac{1+\alpha}{p} + \frac{k}{q} \right) \bigg\}
	\\
& +\sum_{\ell = 1}^{p-1}\bigg\{ \ln G \left(  \frac{\ell+1+\alpha}{p} \right)-\ln G \left( \frac{\ell+1+\alpha}{p} + 1 \right) \bigg\}.
\end{align*}
Using the identity $G(z+1)=\Gamma(z)G(z)$ and the duplication formula for $\Gamma$ (see \eqref{duplication Gamma}), we obtain 
\begin{align*}
D(p,q, \alpha) =&\;\sum_{k=1}^{q-1} \ln \Gamma \left( \frac{1+\alpha}{p} + \frac{k}{q} \right) -\sum_{\ell = 1}^{p-1} \ln \Gamma \left(  \frac{1+\alpha}{p} +\frac{\ell}{p} \right) 
	\\
= &\; \sum_{k=0}^{q-1} \ln \Gamma \left( \frac{1+\alpha}{p} + \frac{k}{q} \right) -\sum_{\ell = 0}^{p-1} \ln \Gamma \left(  \frac{1+\alpha}{p}+\frac{\ell}{p} \right) 
\\
=&\; \ln \Gamma \left( \frac{1+\alpha}{\theta} \right) - \left( \frac{1+\alpha}{\theta}-\frac{1}{2} \right)  \ln q -\frac{1-q}{2}\ln(2\pi)- \ln \Gamma \left(1+\alpha \right) 
	\\
& + \left(\frac{1}{2}+\alpha \right) \ln p+ \frac{1-p}{2} \ln (2\pi).
\end{align*}
Substituting this expression for $D$ into (\ref{dminusd}), we arrive at
\begin{align*}
d(\theta,\alpha)-d \left( \frac{1}{\theta},\frac{1+\alpha}{\theta}-1 \right) =&\; \ln \Gamma \left( \frac{1+\alpha}{\theta} \right) - \ln \Gamma \left(1+\alpha \right) +\left( \frac{1+\alpha}{\theta}-\frac{1}{2} \right)  \ln \theta
\\
&+ \frac{1+6\alpha^{2} + \theta (3+\theta) + 6\alpha(1+\theta)}{12 \theta} \ln \theta 
\\
=&\; \ln \Gamma \left( \frac{1+\alpha}{\theta} \right) - \ln \Gamma \left( 1+\alpha \right) + \frac{13 + 6 \alpha^{2} + \theta(\theta-3) + 6 \alpha (\theta + 3)}{12 \theta}\ln \theta,
\end{align*}
which proves (\ref{symmetry for d}) for rational values of $\theta$.

The definition (\ref{def of the constant d}) of $d(\theta, \alpha)$ can be written as
\begin{align}\label{ddN}
d(\theta,\alpha) = \lim_{N\to + \infty} d_N(\theta,\alpha), 
\end{align}
where the functions $d_N$ are defined by
\begin{align*}
d_N(\theta,\alpha) =&\; \sum_{k=1}^{N} \ln \Gamma (1+\alpha + k \theta) - \Bigg\{\frac{\theta}{2} N^{2} \ln N +  \frac{\theta(2 \ln \theta - 3)}{4}N^{2} \\
& +\left(1+\alpha + \frac{\theta -1}{2}\right) N \ln N + \left( \frac{\ln(2\pi)}{2}-(1+\alpha)+\frac{1-\theta}{2} + \left( \alpha + \frac{1+\theta}{2} \right) \ln \theta \right) N \nonumber \\
& +\frac{1+6\alpha^{2} + \theta (3+\theta) + 6\alpha(1+\theta)}{12 \theta} \ln N \Bigg\}.
\end{align*}
The proof of Proposition \ref{dsymmprop} will be complete if we can show that the convergence in (\ref{ddN}) is uniform for $(\theta,\alpha)$ in compact subsets of $\mathcal{U}$, where
$$\mathcal{U} = (\C \setminus [0,-\infty)) \times \{ \alpha \in \C: \re \alpha >-1\}.$$
Indeed, if this is the case, then since each function $d_N$ is holomorphic $\mathcal{U} \to \C$, so is $d$; thus (\ref{symmetry for d}) must hold also for irrational values of $\theta > 0$ by continuity.

Let $K \subset \mathcal{U}$ be compact. By \eqref{stirlingwithremainder}, we have
\begin{align*}
\sum_{k=1}^{N} \ln \Gamma (1+\alpha + k \theta) =&\; \sum_{k=1}^{N}\bigg[ (1+\alpha + k \theta)  \ln(1+\alpha + k \theta)-  (1+\alpha + k \theta) -\frac{1}{2}   \ln(1+\alpha + k \theta)
\\
&+\frac{1}{2} \ln(2\pi)+\frac{1}{12(1+\alpha + k \theta)}\bigg] + \sum_{k=1}^{N} \mathcal{D}_1(1+\alpha + k \theta),
\end{align*}
where $\mathcal{D}_1$ is the remainder defined in (\ref{calDNdef}).
Using the relation \eqref{Hurwitz identity difference}, we obtain
\begin{align}\nonumber
&\sum_{k=1}^{N}\bigg[ (1+\alpha + k \theta)  \ln(1+\alpha + k \theta)-  (1+\alpha + k \theta) -\frac{1}{2}   \ln(1+\alpha + k \theta)+\frac{1}{2} \ln(2\pi)+\frac{1}{12(1+\alpha + k \theta)}\bigg]
	\\\nonumber
= &\; \theta \bigg( \zeta'\bigg( -1,1+N+\frac{1+\alpha}{\theta} \bigg)-  \zeta'\bigg( -1,1+\frac{1+\alpha}{\theta} \bigg)\bigg) + \frac{\ln \theta}{2} N(2+2\alpha +\theta + N\theta) 
	\\\nonumber
&- N(1+\alpha)+\theta \frac{N^2 +N-2}{2}- \frac{1}{2} N \ln \theta -\frac{1}{2} \bigg( \ln \Gamma \bigg( 1+\frac{1+\alpha}{\theta}+N \bigg) - \ln \Gamma \bigg( 1+\frac{1+\alpha}{\theta} \bigg)\bigg)
\\ \label{sumk1N1alpha}
&+\frac{N}{2} \ln (2\pi)+ \frac{1}{12 \theta} \bigg( \psi  \bigg( 1+\frac{1+\alpha}{\theta}+N \bigg) - \psi \bigg( 1+\frac{1+\alpha}{\theta} \bigg)\bigg).
\end{align}
All the special functions on the right-hand side of (\ref{sumk1N1alpha}) have uniform expansions for large $N$ whenever the argument of
\[
1+\frac{1+\alpha}{\theta}+N 
\]
is bounded away from $\pm \pi$; in particular, this is the case for $(\theta,\alpha)\in K$. Furthermore, by \eqref{remainderO}, there are constants $C',C''>0$ (that only depend on $K$) such that
\begin{align*}
|\mathcal{D}_1(1+\alpha +k\theta)| \le \frac{C'}{(1+\alpha + k\theta)^3} \le \frac{C''}{k^3}
\end{align*}
and thus the series $\sum_{k=1}^{\infty} \mathcal{D}_1(1+\alpha +k\theta)$ converges uniformly for $(\theta,\alpha)\in K$. We conclude that the sequence of functions $d_N$ converges to $d$ uniformly for $(\theta,\alpha)$ in compact subsets of $\mathcal{U}$ and thus the proof of Proposition \ref{dsymmprop} is complete.

\section{Proof of Lemma \ref{lemma: integrals for I1}}\label{I1lemmaapp}
Assume $\alpha > -1$ and $\theta \in (0,1]$, so that $b_1$ and $b_2$ lie in the second and first quadrants, respectively.
For any integer $j$, a contour deformation shows that
\begin{align}\label{lol6}
2\int_{\gamma_{b_2b_1}} \frac{\zeta^j}{r(\zeta)} \frac{d\zeta}{2\pi i} 
= 2\int_{\gamma_{b_2b_1}} \frac{\zeta^j}{\tilde{r}(\zeta)} \frac{d\zeta}{2\pi i} =
\int_{\sigma} \frac{\zeta^j}{\tilde{r}(\zeta)} \frac{d\zeta}{2\pi i},
\end{align}
where $\sigma$ is a closed loop surrounding once $\Sigma_{5}$ in the positive direction but not surrounding $0$, and where we recall that the square roots $r$ and $\tilde{r}$ defined in (\ref{rdef}) and (\ref{rtildedef}) have branch cuts along $\Sigma_{5}$ and $[b_{1},b_{2}]$, respectively. If $j \geq 0$, then $\frac{\zeta^j}{\tilde{r}(\zeta)}$ is analytic in $\mathbb{C}\setminus [b_{1},b_{2}]$, and then deforming the contour $\sigma$ to infinity, we see that the right-hand side of (\ref{lol6}) equals the coefficient of $\zeta^{-1}$ in the large $\zeta$ expansion of $\frac{\zeta^j}{r(\zeta)}$. Since
$$\frac{1}{r(\zeta)} = \frac{1}{\zeta} + \frac{b_1 + b_2}{2\zeta^2} 
+ \frac{3 b_1^2+2 b_1 b_2+3   b_2^2}{8 \zeta^3 }
+ \bigO(\zeta^{-4}), \qquad \zeta \to \infty,$$
this proves the first three identities \eqref{int1overr}--\eqref{intzeta2overr} of the lemma. 

On the other hand, if $j \leq -1$ in \eqref{lol6}, then $\frac{\zeta^{j}}{\tilde{r}(\zeta)}$ has no residue at $\infty$ but has a pole of order $|j|$ at $0$. By deforming the contour $\sigma$ through $\infty$, we obtain
\begin{equation*}
2 \int_{\gamma_{b_{2}b_{1}}} \frac{\zeta^{j}}{ r(\zeta)}\frac{d\zeta}{2\pi i} = - \int_{C_{\epsilon}} \frac{\zeta^{j}}{ \tilde{r}(\zeta)}\frac{d\zeta}{2\pi i},
\end{equation*}
where $C_{\epsilon}$ denotes a small circle of radius $\epsilon$ centered at $0$ oriented positively. Therefore the right-hand side of (\ref{lol6}) is equal to the coefficient of $\zeta^{-1}$ in the expansion of $\frac{\zeta^{j}}{\tilde{r}(\zeta)}$ as $\zeta \to 0$. Since
\begin{equation*}
\frac{1}{\tilde{r}(\zeta)} = \frac{i}{|b_{2}|} + \bigO(\zeta) \qquad \mbox{as } \zeta \to 0,
\end{equation*}
this proves \eqref{int1overzetar}. 

To prove the remaining four identities, we note that the same kind of argument that gave \eqref{lol6}, shows that, for any $j \in \mathbb{Z}$, 
\begin{align}
2\int_{\gamma_{b_2b_1}} \frac{\zeta^j \ln{\zeta}}{r(\zeta)} \frac{d\zeta}{2\pi i} 
= &\; 2\int_{\gamma_{b_2b_1}} \frac{\zeta^j \ln{\zeta}}{\tilde{r}(\zeta)} \frac{d\zeta}{2\pi i} = \int_{\sigma} \frac{\zeta^j \ln{\zeta}}{\tilde{r}(\zeta)} \frac{d\zeta}{2\pi i}. \label{lol7}
\end{align}
If $j \geq 0$, then by deforming $\sigma$ to $C_{R} \cup \big((-R,0) + i 0^{+}\big) \cup \big((0,-R) - i 0^{+}\big)$ where $R > 0$ is any large radius, and noting that $r = \tilde{r}$ over the range of integration, we get
\begin{align*}
2\int_{\gamma_{b_2b_1}} \frac{\zeta^j \ln{\zeta}}{r(\zeta)} \frac{d\zeta}{2\pi i} = &\; \int_{C_{R}} \frac{\zeta^j \ln{\zeta}}{r(\zeta)} \frac{d\zeta}{2\pi i} 
+ \int_{-R}^0 \frac{\zeta^j 2\pi i}{r(\zeta)} \frac{d\zeta}{2\pi i}.
\end{align*}
Since the left-hand side is independent of $R$, by taking the limit $R \to + \infty$, we obtain
\begin{align}\nonumber
&2\int_{\gamma_{b_2b_1}} \frac{\zeta^j \ln{\zeta}}{r(\zeta)} \frac{d\zeta}{2\pi i}  
	\\  \label{lol8}
& =\lim_{R\to + \infty} \bigg\{\int_{-\pi}^{\pi} R^je^{i j \varphi} \ln(Re^{i\varphi})\bigg(\frac{1}{Re^{i\varphi}} + \frac{b_1 + b_2}{2 R^2 e^{2i \varphi}} + \bigO(R^{-3})\bigg) \frac{iRe^{i\varphi} d\varphi}{2\pi i} 
+ \int_{-R}^0 \frac{\zeta^j 2\pi i}{r(\zeta)} \frac{d\zeta}{2\pi i} \bigg\}. 
\end{align}

Taking $j = 0$ in (\ref{lol8}), we find
\begin{align*}
2\int_{\gamma_{b_2b_1}} \frac{\ln{\zeta}}{r(\zeta)} \frac{d\zeta}{2\pi i} 
=& \lim_{R \to \infty}
\bigg\{\int_{-\pi}^{\pi} \ln(Re^{i\varphi})\bigg(\frac{1}{Re^{i\varphi}} + \frac{b_1 + b_2}{2 R^2 e^{2i \varphi}} \bigg) \frac{iRe^{i\varphi} d\varphi}{2\pi i} 
 + \int_{-R}^0 \frac{1}{r(\zeta)} d\zeta \bigg\}
	\\
=& \lim_{R \to \infty}
\bigg\{\ln{R} + \int_{-R}^0 \frac{1}{r(\zeta)} d\zeta \bigg\}.
\end{align*}
Using that
$$\frac{d}{d\zeta} \ln\big(-2 r(\zeta) (1 + r'(\zeta))\big) = \frac{1}{r(\zeta)},$$
we can compute the large $R$ asymptotics of the integral from $-R$ to $0$:
\begin{align*}
& \int_{-R}^0 \frac{1}{r(\zeta)} d\zeta
= \bigg[\ln\big(-2 r(\zeta) (1 + r'(\zeta))\big)\bigg]_{\zeta = -R}^0  = -\ln(R) + \ln(i(|b_2| + \im{b_2})) - \ln{2} + \bigO(R^{-1}).
\end{align*}
This yields \eqref{logzetaoverr}. 

Taking $j = 1$ in \eqref{lol8}, we find
\begin{align*}
& 2\int_{\gamma_{b_2b_1}} \frac{\zeta \ln{\zeta}}{r(\zeta)} \frac{d\zeta}{2\pi i} 
= \lim_{R \to \infty}
\bigg\{\int_{-\pi}^{\pi} Re^{i \varphi} \ln(Re^{i\varphi})\bigg(\frac{1}{Re^{i\varphi}} + \frac{b_1 + b_2}{2 R^2 e^{2i \varphi}}  \bigg) \frac{iRe^{i\varphi} d\varphi}{2\pi i} + \int_{-R}^0  \frac{\zeta}{r(\zeta)} d\zeta
 \bigg\}.
\end{align*}
Using that
$$\frac{d}{d\zeta} \bigg\{r(\zeta) + \frac{b_1 + b_2}{2} \ln(-2 r(\zeta) (1 + r'(\zeta)))\bigg\} = \frac{\zeta}{r(\zeta)},$$
we obtain the following large $R$ asymptotics:
\begin{align*}\nonumber
& \int_{-R}^0 \frac{\zeta}{r(\zeta)} d\zeta
= \bigg[r(\zeta) + \frac{b_1 + b_2}{2} \ln[-2 r(\zeta) (1 + r'(\zeta))]\bigg]_{\zeta = -R}^0
	\\ \nonumber
& = r_{-}(0) + \frac{b_1 + b_2}{2} \ln[-2 r_{-}(0) (1 + r_{-}'(0))]
- \bigg(r(-R) + \frac{b_1 + b_2}{2} \ln[-2 r(-R) (1 + r'(-R))]\bigg)
	\\ \nonumber
& = -i|b_2| + i \im(b_2)\ln(2i(|b_2| + \im b_2))
	\\ \nonumber
&\quad - \bigg(-\sqrt{(R+b_1)(R+b_2)} + \frac{b_1 + b_2}{2} \ln[2 \sqrt{(R+b_1)(R+b_2)} + b_1 + b_2 + 2R]\bigg)
	\\
& = R - i\im(b_2) \ln R
-i |b_2| + i\Big(1 + \ln(2i(|b_2| + \im b_2)) - \ln 4\Big)\im{b_2} + \bigO(R^{-1}),
\end{align*}
where we have fixed the branch of $\sqrt{(R+b_1)(R+b_2)}$ so that $\sqrt{(R+b_1)(R+b_2)} \sim R$ as $R \to +\infty$. Since
\begin{align*}
& \int_{-\pi}^{\pi} \ln (Re^{i\varphi}) \frac{d \varphi}{2\pi} = \ln R, \\
& \int_{-\pi}^{\pi} \ln (Re^{i\varphi})Re^{i\varphi} \frac{d \varphi}{2\pi} = - R,
\end{align*}
this proves \eqref{zetalogzetaoverr}. 

Taking $j = 2$ in \eqref{lol8} and utilizing the fact that
\begin{align*}
& \int_{-\pi}^{\pi} \ln (Re^{i\varphi})R^{2}e^{2i\varphi} \frac{d \varphi}{2\pi} = \frac{R^{2}}{2},
\end{align*}
we find
\begin{align}
2\int_{\gamma_{b_2b_1}} \frac{\zeta^2 \ln{\zeta}}{r(\zeta)} \frac{d\zeta}{2\pi i} 
=&\; \lim_{R \to \infty}
\bigg\{\int_{-\pi}^{\pi} R^2e^{2i \varphi} \ln(Re^{i\varphi})\bigg(\frac{1}{Re^{i\varphi}}  + \frac{b_1 + b_2}{2 R^2 e^{2i \varphi}} \nonumber
	\\
&  + \frac{3 b_1^2+2 b_1 b_2+3   b_2^2}{8 R^3 e^{3i\varphi}}  \bigg) \frac{iRe^{i\varphi} d\varphi}{2\pi i} 
 + \int_{-R}^0 \frac{\zeta^2 }{r(\zeta)} d\zeta \bigg\} \nonumber
	\\
=&\; \lim_{R \to \infty}
\bigg\{\frac{R^2}{2} - \frac{b_1 + b_2}{2} R + \frac{(\re{b_2})^2 - 2(\im{b_2})^2}{2}\ln{R} + \int_{-R}^0 \frac{\zeta^2 }{r(\zeta)} d\zeta \bigg\}. \label{lol9}
\end{align}
With the help of the identity
$$\frac{d}{d\zeta} \bigg\{\frac{3b_1 + 2\zeta+ 3 b_2}{4}r(\zeta) + \frac{3b_1^2 + 2b_1 b_2 + 3b_2^2}{8} \ln(-2 r(\zeta) (1 + r'(\zeta)))\bigg\} = \frac{\zeta^2}{r(\zeta)},$$
we infer the following large $R$ asymptotics:
\begin{align*}
 \int_{-R}^0 \frac{\zeta^2}{r(\zeta)} d\zeta
= &\; \bigg[\frac{3(b_1 + b_2) + 2\zeta}{4}r(\zeta) + \frac{3b_1^2 + 2b_1 b_2 + 3b_2^2}{8} \ln(-2 r(\zeta) (1 + r'(\zeta)))\bigg]_{\zeta = -R}^0
	\\
 =& -\frac{R^2}{2} + i\im(b_2) R  + \frac{2(\im{b_2})^2 - (\re{b_2})^2}{2} \ln{R}
	\\
& +\frac{1}{4} \bigg(2 \left((\re{b_1})^2-2 (\im{b_2})^2\right) \ln(2 i (|b_2|+\im(b_2)))+6 |b_2| \im(b_2)
	\\
& +(\im{b_2})^2 (8\ln(2)-6)+(\re{b_1})^2 (1-4 \ln(2))\bigg)
 +  \bigO(R^{-1}).
\end{align*}
Substituting the above expansion into \eqref{lol9}, we obtain \eqref{zeta2logzetaoverr}. 

If $j \leq -1$ in \eqref{lol7}, we have
\begin{align*}
 2\int_{\gamma_{b_2b_1}} \frac{\zeta^{j}\ln{\zeta}}{ r(\zeta)} \frac{d\zeta}{2\pi i} 
= &\; \int_{\sigma} \frac{\ln{\zeta}}{\zeta^{|j|} r(\zeta)} \frac{d\zeta}{2\pi i} 
	\\
= &\; \int_{C_{R}} \frac{\ln{\zeta}}{\zeta^{|j|} r(\zeta)} \frac{d\zeta}{2\pi i} 
+ \int_{-R}^{-\epsilon} \frac{2\pi i}{\zeta^{|j|} r(\zeta)} \frac{d\zeta}{2\pi i} 
- \int_{C_{\epsilon}} \frac{\ln{\zeta}}{\zeta^{|j|} r(\zeta)} \frac{d\zeta}{2\pi i} 
	\\
=&\; \int_{-\pi}^{\pi} \frac{\ln(Re^{i\varphi})}{R^{|j|}e^{i |j| \varphi}} \bigg(\frac{1}{Re^{i\varphi}} + \frac{b_1 + b_2}{2 R^2 e^{2i \varphi}} + \bigO(R^{-3})\bigg) \frac{iRe^{i\varphi} d\varphi}{2\pi i} 
	\\
& + \int_{-R}^{-\epsilon} \frac{d\zeta }{\zeta^{|j|} r(\zeta)} 
- \int_{-\pi}^\pi \frac{\ln(\epsilon e^{i\varphi})}{\epsilon^{|j|}e^{i |j| \varphi}}\bigg(\frac{i}{|b_2|} + \frac{\im{b_2}}{|b_2|^3} \epsilon e^{i\varphi} + \bigO(\epsilon^2)\bigg) \frac{i\epsilon e^{i\varphi} d\varphi}{2\pi i},
\end{align*}
as $R \to + \infty$ and $\epsilon \to 0^{+}$. Taking the limit $R \to + \infty$, and then the limit $\epsilon \to 0^{+}$ gives
\begin{align*}
& 2\int_{\gamma_{b_2b_1}} \frac{\zeta^{j}\ln{\zeta}}{ r(\zeta)} \frac{d\zeta}{2\pi i} 
= \lim_{\epsilon \to 0^{+}}  \bigg\{ \int_{-\infty}^{-\epsilon} \frac{d\zeta }{\zeta^{|j|} r(\zeta)} 
- \int_{-\pi}^\pi \frac{\ln(\epsilon e^{i\varphi})}{\epsilon^{|j|}e^{i |j| \varphi}}\bigg(\frac{i}{|b_2|} + \frac{\im{b_2}}{|b_2|^3} \epsilon e^{i\varphi} + \bigO(\epsilon^2)\bigg) \frac{i\epsilon e^{i\varphi} d\varphi}{2\pi i} \bigg\}.
\end{align*}
Suppose now that $j = -1$. Using that
$$\frac{d}{d\zeta} \frac{\ln\bigg(\frac{i(b_2 - \zeta) - \frac{b_2}{|b_2|} r(\zeta)}{i(b_2 - \zeta) + \frac{b_2}{|b_2|} r(\zeta)}\bigg)}{i|b_2|} = \frac{1}{\zeta r(\zeta)},$$
we find that 
\begin{align*}
\int_{-R}^{-\epsilon} \frac{d\zeta}{\zeta r(\zeta)} & 
= \frac{\ln\bigg(\frac{i(b_2 - \zeta) - \frac{b_2}{|b_2|} r(\zeta)}{i(b_2 - \zeta) + \frac{b_2}{|b_2|} r(\zeta)}\bigg)}{i|b_2|}\bigg|_{-\infty}^{-\epsilon}
	\\
& = \frac{i \ln{\epsilon}}{|b_2|}  + \frac{\ln(\frac{2|b_2|^2}{\re{b_2}})}{i|b_2|} 
- \frac{\ln(-\frac{i(|b_2| + \im{b_2})}{\re{b_2}})}{i|b_2|} + \bigO(\epsilon)
\end{align*}
as $\epsilon \to 0^{+}$. Also, by a straightforward computation,
\begin{align*}
& \int_{-\pi}^\pi \frac{\ln(\epsilon e^{i\varphi})}{\epsilon e^{i \varphi}}\bigg(\frac{1}{-i|b_2|} + \frac{\im{b_2}}{|b_2|^3} \epsilon e^{i\varphi} + \cdots\bigg) \frac{i\epsilon e^{i\varphi} d\varphi}{2\pi i}
= \frac{i \ln{\epsilon}}{|b_2|} + \bigO(\epsilon)
\end{align*}
as $\epsilon \to 0^{+}$.
Hence
\begin{align*}\nonumber
 2\int_{\gamma_{b_2b_1}} \frac{\ln{\zeta}}{\zeta r(\zeta)} \frac{d\zeta}{2\pi i} 
= &\; \frac{\ln(\frac{2|b_2|^2}{\re{b_2}})}{i|b_2|} 
- \frac{\ln(-\frac{i(|b_2| + \im{b_2})}{\re{b_2}})}{i|b_2|}
=  \frac{\ln(\frac{2i|b_2|^2}{|b_2| + \im{b_2}})}{i|b_2|},
\end{align*}
which proves the eighth and last identity (\ref{lnzetaoverzetarintegral}). The proof of Lemma \ref{lemma: integrals for I1} is complete.

\section{Proof of Lemma \ref{Zintlemma}}\label{Zintlemmaapp}
Suppose $\alpha > -1$ and $\theta \in (0,1]$. Defining $f_+(\zeta)$ and $f_-(\zeta)$ by
\begin{align*}
& f_\pm(\zeta) = \ln(\pm i\zeta) \pm \ln\bigg(\frac{|b_2|^2 + i\zeta \im{b_2} - i|b_2| r(\zeta)}{(\pm r(\zeta) + \zeta - i\im{b_2})\zeta}\bigg), 
\end{align*}	
the definition (\ref{def of mathcal B}) of $\mathcal{B}(\zeta)$ can be written as
\begin{align}\label{calBff}
  \mathcal{B}(\zeta) = -c_5 f_+(\zeta) - c_6 f_-(\zeta) - c_{7}.
\end{align}
Thus,
\begin{align}\label{intBzetadzeta}
 \int_{\gamma_{b_2b_1}} \mathcal{B}(\zeta) d\zeta
= - c_5 \int_{\gamma_{b_2b_1}} f_+(\zeta) d\zeta
 - c_6  \int_{\gamma_{b_2b_1}} f_-(\zeta) d\zeta  - c_7 (b_1 - b_2).
\end{align}
Integrating by parts and using that
\begin{align}\label{fpmprime}
f_\pm'(\zeta)  = \frac{1}{\zeta} - \frac{1}{r(\zeta)} \mp \frac{i|b_2|}{\zeta r(\zeta)},
\end{align}
we find
\begin{align*}
\int_{\gamma_{b_2b_1}} f_\pm(\zeta) d\zeta
= b_1 f_\pm(b_1) - b_2 f_\pm(b_2) - \int_{\gamma_{b_2b_1}} \bigg(1 - \frac{\zeta}{r(\zeta)} \mp \frac{i|b_2|}{r(\zeta)}\bigg) d\zeta,
\end{align*}
Substituting these expressions into (\ref{intBzetadzeta}) and using (\ref{int1overr}) and (\ref{intzetaoverr}), the first assertion (\ref{calBintegral}) of the lemma follows after simplification. 

To prove (\ref{lnzetacalBintegral}), we use (\ref{calBff}) to write
\begin{align}\label{intlogBdzeta}
 \int_{\gamma_{b_2b_1}}  \ln(\zeta) \mathcal{B}(\zeta) d\zeta
= - c_5 \int_{\gamma_{b_2b_1}}  \ln(\zeta) f_+(\zeta) d\zeta
 - c_6  \int_{\gamma_{b_2b_1}}  \ln(\zeta) f_-(\zeta) d\zeta  - c_7 \int_{\gamma_{b_2b_1}} \ln( \zeta) d\zeta.
\end{align}
Employing (\ref{fpmprime}) and the fact that
$$\int \ln{\zeta} d\zeta = \zeta (\ln(\zeta) - 1),$$
partial integration gives
\begin{align*}
 \int_{\gamma_{b_2b_1}} \ln (\zeta) \mathcal{B}(\zeta) d\zeta
= &- c_5 \bigg\{\Big[ \zeta (\ln(\zeta) - 1) f_+(\zeta)\Big]_{\zeta=b_2}^{b_1}
-  \int_{\gamma_{b_2b_1}} (\ln(\zeta) - 1) \bigg(1 - \frac{\zeta}{r(\zeta)} - \frac{i|b_2|}{r(\zeta)}\bigg)d\zeta\bigg\}
 	\\
&  - c_6\bigg\{\Big[\zeta (\ln(\zeta) - 1) f_-(\zeta)\Big]_{\zeta=b_2}^{b_1}
 -   \int_{\gamma_{b_2b_1}} (\ln(\zeta) - 1) \bigg(1 - \frac{\zeta}{r(\zeta)} + \frac{i|b_2|}{r(\zeta)}\bigg)d\zeta\bigg\}
  	\\
& - c_7 \int_{\gamma_{b_2b_1}} \ln (\zeta) d\zeta,
\end{align*}
that is,
\begin{align}\nonumber
 \int_{\gamma_{b_2b_1}} \ln (\zeta) \mathcal{B}(\zeta) d\zeta
= &- c_5 \Big[ \zeta (\ln(\zeta) - 1) f_+(\zeta)\Big]_{\zeta=b_2}^{b_1}
- c_6\Big[\zeta (\ln(\zeta) - 1) f_-(\zeta)\Big]_{\zeta=b_2}^{b_1}
	\\\nonumber
&  + (c_5 + c_6)   \int_{\gamma_{b_2b_1}} (\ln(\zeta) - 1) d\zeta
 -(c_5 + c_6)  \int_{\gamma_{b_2b_1}}  \frac{\zeta \ln(\zeta)}{r(\zeta)} d\zeta
 	\\\nonumber
& + (c_5 + c_6) \int_{\gamma_{b_2b_1}}  \frac{\zeta}{r(\zeta)} d\zeta
 - (c_5 - c_6) i|b_2| \int_{\gamma_{b_2b_1}} \frac{\ln(\zeta)}{r(\zeta)}d\zeta
  	\\\label{lnBzetadzeta}
&+ (c_5 - c_6)  i|b_2| \int_{\gamma_{b_2b_1}} \frac{1}{r(\zeta)}d\zeta
  - c_7 \int_{\gamma_{b_2b_1}} \ln (\zeta) d\zeta.
\end{align}
An easy computation shows that
$$\int_{\gamma_{b_2b_1}} \ln(\zeta) d\zeta = b_2 - b_1 + b_1 \ln(b_1) - b_2 \ln(b_2)$$
and explicit expressions for the other integrals on the right-hand side of (\ref{lnBzetadzeta}) have already been obtained in (\ref{zetalogzetaoverr}), (\ref{intzetaoverr}), (\ref{logzetaoverr}), and (\ref{int1overr}).
Substituting these expressions into (\ref{lnBzetadzeta}), a long but straightforward simplification gives (\ref{lnzetacalBintegral}).

We next prove (\ref{calAintegral}). According to the definition \eqref{def of mathcal A} of $\mathcal{A}$, we have
\begin{align}\label{intgammacalA}
\int_{\gamma_{b_2b_1}} \mathcal{A}(\zeta) d\zeta 
= \frac{ic_8}{2}\int_{\gamma_{b_2b_1}} \frac{d\zeta}{\zeta} 
+  \frac{c_8-\frac{3\alpha^2-1}{12}}{2|b_2|} \int_{\gamma_{b_2b_1}} \frac{r(\zeta)}{\zeta}d\zeta.
\end{align} 
Clearly,
\begin{align}\label{int1overzeta}
\int_{\gamma_{b_{2}b_{1}}} \frac{d\zeta}{\zeta} = \ln b_{1} - \ln b_{2} = i(\arg b_{1} - \arg b_{2}),
\end{align}
so it only remains to compute $ \int_{\gamma_{b_2b_1}} \frac{r(\zeta)}{\zeta}d\zeta$. 
Let $\tilde{r}(\zeta)$ denote the analytic continuation of $r(\zeta)$ defined in (\ref{rtildedef}). 
For $r > 0$, let $C_r$ denote the positively oriented circle of radius $r$ centered at the origin. 
A contour deformation shows that
\begin{equation}\label{lol11}
\int_{\gamma_{b_{2}b_{1}}}\frac{r(\zeta)}{ \zeta}d\zeta = \int_{\gamma_{b_{2}b_{1}}}\frac{\tilde{r}(\zeta)}{ \zeta}d\zeta = \frac{1}{2} \int_{C_{R}}\frac{\tilde{r}(\zeta)}{ \zeta}d\zeta - \frac{1}{2}\int_{C_{\epsilon}}\frac{\tilde{r}(\zeta)}{ \zeta}d\zeta,
\end{equation}
where $R>|b_{1}|$ and $0 < \epsilon < \im b_2$.  
Using the expansions
\begin{align*}
& \frac{\tilde{r}(\zeta)}{\zeta} = 1 - \frac{b_{1}+b_{2}}{2\zeta} + \bigO(\zeta^{-2}), \qquad \zeta \to \infty, \\
& \frac{\tilde{r}(\zeta)}{\zeta} = - \frac{i |b_{2}|}{\zeta} + \bigO(1), \qquad \zeta \to 0,
\end{align*}
and letting $R \to + \infty$ and $\epsilon \to 0^{+}$ in \eqref{lol11}, we infer that
\begin{align}\label{introverzeta}
\int_{\gamma_{b_{2}b_{1}}}\frac{r(\zeta)}{ \zeta}d\zeta = - \pi \big( |b_{2}|-\im b_{2} \big).
\end{align}
Substituting (\ref{int1overzeta}) and (\ref{introverzeta}) into (\ref{intgammacalA}) and simplifying, we find \eqref{calAintegral}.

To prove (\ref{lnzetacalAintegral}), we use the definition \eqref{def of mathcal A} of $\mathcal{A}$ to write
\begin{equation}\label{lol14}
\int_{\gamma_{b_{2}b_{1}}} \ln ( \zeta) \mathcal{A}(\zeta)d\zeta = \frac{ic_{8}}{2}\int_{\gamma_{b_{2}b_{1}}} \frac{\ln (\zeta)}{\zeta}d\zeta + \frac{c_8-\frac{3\alpha^2-1}{12}}{2|b_2|}\int_{\gamma_{b_{2}b_{1}}} \frac{\ln (\zeta) r(\zeta)}{\zeta}d\zeta.
\end{equation}
The first integral on the right-hand side is easily computed:
\begin{align}\label{lol13}
 \int_{\gamma_{b_2b_1}} \frac{\ln{\zeta}}{\zeta } d\zeta 
= \frac{(\ln{b_1})^2 - (\ln{b_2})^2}{2}.
\end{align}
To compute the second integral, we use a contour deformation to obtain
\begin{align}\nonumber
\int_{\gamma_{b_2b_1}} r(\zeta)\frac{\ln \zeta}{\zeta}d\zeta & = \frac{1}{2}\int_{C_{R}}\frac{\tilde{r}(\zeta) \ln \zeta}{ \zeta}d\zeta - \frac{1}{2}\int_{C_{\epsilon}}\frac{\tilde{r}(\zeta)\ln \zeta}{ \zeta}d\zeta + \pi i \int_{-R}^{-\epsilon} \frac{\tilde{r}(\zeta)}{\zeta}d\zeta,
	\\ \label{lol12} 
& = \frac{i}{2}\int_{-\pi}^{\pi} \tilde{r}(Re^{i\varphi})\ln(Re^{i\varphi})d\varphi - \frac{i}{2}\int_{-\pi}^{\pi} \tilde{r}(\epsilon e^{i\varphi})\ln(\epsilon e^{i\varphi})d\varphi + \pi i \int_{-R}^{-\epsilon} \frac{\tilde{r}(\zeta)}{\zeta}d\zeta. 
\end{align}
where $R$ and $\epsilon$ are as in \eqref{lol11}. Since
\begin{align*}
\partial_{\zeta} \Bigg[ r(\zeta) + i |b_{2}| \ln \left( \frac{b_{2}r(\zeta)+i|b_{2}|(\zeta-b_{2})}{-b_{2} r(\zeta)+i|b_{2}|(\zeta-b_{2})} \right) - \frac{b_{1}+b_{2}}{2}\ln \big( 2r(\zeta) (  1+r'(\zeta)) \big) \Bigg] = \frac{r(\zeta)}{\zeta},
\end{align*}
we have
\begin{align*}
\int_{-R}^{-\epsilon} \frac{\tilde{r}(\zeta)}{\zeta}d\zeta = & - i \sqrt{- b_{1}-\epsilon}\sqrt{b_{2} + \epsilon}  + i |b_{2}| \ln \left( \frac{|b_{2}| \sqrt{b_{2}+\epsilon} + b_{2} \sqrt{-b_{1}-\epsilon}}{|b_{2}| \sqrt{b_{2}+\epsilon} - b_{2} \sqrt{-b_{1}-\epsilon}} \right) 
	\\
& - i \im (b_{2}) \left[ \ln \left( 2\epsilon + 2 i \im (b_{2}) + 2i \sqrt{- b_{1}-\epsilon}\sqrt{b_{2} + \epsilon}  \right) -\pi i \right] \\
& + \sqrt{b_{1}+R}\sqrt{b_{2}+R} - i \im(b_{2}) \ln \left( \frac{i |b_{2}| \sqrt{b_{2}+R} + b_{2} \sqrt{b_{1}+R}}{i |b_{2}| \sqrt{b_{2}+R} - b_{2} \sqrt{b_{1}+R}} \right) 
	\\
& +i \im(b_{2}) \left[ \ln\left( 2i \im(b_{2})+2R+2\sqrt{b_{1}+R}\sqrt{b_{2}+R} \right) -\pi i \right].
\end{align*}
Substituting this into \eqref{lol12} and letting $R \to + \infty$ and $\epsilon \to 0^{+}$ in the resulting equation, we find
\begin{align}\nonumber
\int_{\gamma_{b_2b_1}} r(\zeta)\frac{\ln \zeta}{\zeta}d\zeta = & \; \frac{i}{2}\int_{-\pi}^{\pi} \Big( Re^{i\varphi}- \frac{b_{1}+b_{2}}{2} + \bigO(R^{-1}) \Big)\ln(Re^{i\varphi})d\varphi 
	\\\nonumber
& -\frac{i}{2} \int_{-\pi}^{\pi} \Big( -i|b_{2}|  + \bigO(\epsilon) \Big) \ln(\epsilon e^{i\varphi}) d\varphi  +\pi i  \int_{-R}^{-\epsilon} \frac{\tilde{r}(\zeta)}{\zeta}d\zeta 
	\\ \label{intrlogzetaoverzeta}
= & \; \pi\bigg\{|b_2| - |b_2| \ln\bigg(\frac{2i|b_2|^2}{|b_2| + \im{b_2}}\bigg)
- \im{b_2} + \im(b_2)\ln\bigg(\frac{i(|b_2| + \im{b_2})}{2}\bigg)\bigg\}
\end{align}
Substituting of (\ref{lol13}) and (\ref{intrlogzetaoverzeta}) into (\ref{lol14}), we obtain (\ref{lnzetacalAintegral}).

We finally prove (\ref{calBoverzetaintegral}). By (\ref{calBff}), we have
\begin{align}\label{intBoverzeta}
 \int_{\gamma_{b_2b_1}}  \frac{\mathcal{B}(\zeta)}{\zeta} d\zeta
= - c_5 \int_{\gamma_{b_2b_1}}  \frac{f_+(\zeta)}{\zeta} d\zeta
 - c_6  \int_{\gamma_{b_2b_1}}  \frac{f_-(\zeta)}{\zeta} d\zeta  - c_7 \int_{\gamma_{b_2b_1}} \frac{d\zeta}{\zeta}.
\end{align}
Integration by parts using (\ref{fpmprime}) gives
\begin{align*}
 \int_{\gamma_{b_2b_1}}  \frac{\mathcal{B}(\zeta)}{\zeta} d\zeta
 = &- c_5 \bigg\{\Big[ \ln(\zeta) f_+(\zeta)\Big]_{\zeta=b_2}^{b_1}
-  \int_{\gamma_{b_2b_1}} \ln(\zeta) \bigg(\frac{1}{\zeta} - \frac{1}{r(\zeta)} - \frac{i|b_2|}{\zeta r(\zeta)}\bigg)d\zeta\bigg\}
 	\\
&  - c_6\bigg\{\Big[\ln(\zeta) f_-(\zeta)\Big]_{\zeta=b_2}^{b_1}
 -   \int_{\gamma_{b_2b_1}} \ln(\zeta) \bigg(\frac{1}{\zeta} - \frac{1}{r(\zeta)} + \frac{i|b_2|}{\zeta r(\zeta)}\bigg)d\zeta\bigg\}
 - c_7 \int_{\gamma_{b_2b_1}} \frac{d\zeta}{\zeta},
\end{align*}
i.e.,
\begin{align}\nonumber
 \int_{\gamma_{b_2b_1}}  \frac{\mathcal{B}(\zeta)}{\zeta} d\zeta
 = &- c_5\Big[ \ln(\zeta) f_+(\zeta)\Big]_{\zeta=b_2}^{b_1}
  - c_6\Big[\ln(\zeta) f_-(\zeta)\Big]_{\zeta=b_2}^{b_1}
+(c_5 + c_6) \int_{\gamma_{b_2b_1}} \frac{\ln(\zeta)}{\zeta} d\zeta
	\\ \label{intcalBoverzeta}
& - (c_5 + c_6) \int_{\gamma_{b_2b_1}} \frac{\ln(\zeta)}{r(\zeta)} d\zeta
+ i|b_2| (c_6- c_5) \int_{\gamma_{b_2b_1}} \frac{\ln(\zeta)}{\zeta r(\zeta)} d\zeta
  - c_7 \int_{\gamma_{b_2b_1}} \frac{d\zeta}{\zeta},
\end{align}
The four integrals on the right-hand side of (\ref{intcalBoverzeta}) have been computed in (\ref{lol13}), (\ref{logzetaoverr}), (\ref{lnzetaoverzetarintegral}), and (\ref{int1overzeta}), respectively. 
Substituting the expressions from these equations into (\ref{intcalBoverzeta}) and simplifying, (\ref{calBoverzetaintegral}) follows. This completes the proof of Lemma \ref{Zintlemma}.

\end{document}